\documentclass[11pt]{article}

\newcommand{\footremember}[2]{%
    \footnote{#2}
    \newcounter{#1}
    \setcounter{#1}{\value{footnote}}%
}

\usepackage[normalem]{ulem}

\usepackage[margin=1in]{geometry}
\usepackage{amsmath,amssymb,amsthm,mathtools,comment,bm, booktabs}
\usepackage{algorithm}
\usepackage{float}
\usepackage[noend]{algpseudocode}
\usepackage{natbib}

\usepackage{enumitem}
\usepackage{hyperref}
\usepackage{xcolor}
\usepackage{bbm}

\hypersetup{
  colorlinks=true,
  linkcolor=blue!60!black,
  citecolor=blue!80!black,
  urlcolor=blue!60!black
}

\newtheorem{theorem}{Theorem}
\newtheorem{lemma}{Lemma}
\newtheorem{proposition}{Proposition}

\newtheorem{definition}{Definition}
\newtheorem{corollary}{Corollary}

\newcommand{\R}{\mathbb{R}}
\newcommand{\E}{\mathbb{E}}

\newcommand{\argmin}{\operatorname*{argmin}}
\newcommand{\Reg}{\operatorname{Reg}}
\newcommand{\TV}{\operatorname{TV}}
\newcommand{\Hist}{\mathcal{H}}

\newcommand{\rhobar}{\bar\rho}

\newcommand{\norm}[1]{\left\lVert #1\right\rVert}
\newcommand{\inner}[2]{\left\langle #1,#2\right\rangle}

\newtheorem{assumption}{Assumption}

\newcommand{\Unif}{\operatorname{Unif}}

\usepackage{amsmath,amssymb,amsthm,mathtools}
\usepackage{enumitem}

\newcommand{\bits}{\operatorname{bits}}

\title{Independent Reinforcement Learning in Discounted Markov Games}
\author{Draft note}
\date{}

\author{%
  Asrın Efe Yorulmaz\footremember{ay}{Department of Electrical and Computer Engineering, University of Illinois Urbana-Champaign, Urbana, IL, 61801, USA, ay20@illinois.edu}%
  \and
  U\u{g}ur Ayd{\i}n\footremember{ua}{Department of Electrical and Computer Engineering, University of Illinois Urbana-Champaign, Urbana, IL, 61801, USA, uaydin2@illinois.edu}%
  \and Tamer Ba\c{s}ar\footremember{tb}{Department of Electrical and Computer Engineering, University of Illinois Urbana-Champaign, Urbana, IL, 61801, USA, basar1@illinois.edu}%
}

\begin{document}
\maketitle

\begin{abstract}
In this work, we study \emph{radically uncoupled} learning in discounted general-sum Markov games. Assuming ``$\mathsf{ETH}$ for $\mathsf{PPAD}$", we show that, for every fixed discount factor, there is no polynomial-time algorithm for computing inverse-polynomially accurate coarse correlated equilibria in discounted general-sum Markov games when players learn independently in decentralized settings. Complementing this hardness result, we provide what appears to be the first \emph{radically uncoupled} algorithm with sub-exponential convergence guarantees to coarse correlated equilibria in discounted general-sum Markov games without imposing any structural restrictions on the game. Our algorithm is a \emph{layered} variant of optimistic mirror descent with an increasing step-size schedule tailored to the multi-agent setting. Finally, we develop both full-feedback and partial feedback versions of the aforementioned algorithm and establish sub-exponential convergence guarantees for each case.

\end{abstract}

\section{Introduction}
\label{sec:introduction}
Multi-agent reinforcement learning (MARL) studies how multiple decision-making agents learn and act in a shared environment, where each agent seeks to optimize its own objective while adapting to the behavior of others \citep{BusoniuBabuskaSchutter2008,ZhangYangBasar2021}. This framework has become a central modeling paradigm for large-scale decision-making systems involving both cooperation and competition. Its recent applications range from strategic game playing---including Go, poker, Stratego, and Diplomacy \citep{SilverHuangMaddisonEtAl2016,BrownSandholm2018,PerolatEtAl2022,KramarEtAl2022,BakhtinEtAl2022}---to large language models \citep{llm_multiagent1,llm_multiagent2}, robotics \citep{marl_robotics}, quantitative finance \citep{marl_finance}, intelligent transportation \citep{marl_autonomous_driving}, cybersecurity \citep{MalialisKudenko2015}, and economic policy design \citep{ZhengEtAl2022}. These applications motivate algorithms that are not only computationally and statistically efficient, but also robust to decentralization, limited information, and strategic non-stationarity.

\vspace{6pt}
\noindent 
The standard mathematical model for the theoretical study of MARL is stochastic games, also known as Markov games \citep{Shapley1953,Littman1994}. In a Markov game, interaction unfolds over a sequence of states. At each state, all players choose actions simultaneously; each player receives a reward depending on the current state and the joint actions; and the next state is drawn from a transition kernel that also depends on the joint actions. Thus, Markov games extend Markov decision processes (MDPs) to multi-agent environments, while also extending normal-form games to dynamic environments with state-dependent rewards and transitions.
This combination creates difficulties that are absent from each of these two special cases. In contrast to single-agent MDPs, the value of a player's policy depends on the policies chosen by the other players, and thus optimality must be replaced by an equilibrium notion. In contrast to normal-form games, a player's current action affects not only her immediate reward, but also the future state distribution on which later strategic interactions take place. Consequently, equilibrium computation in Markov games combines the fixed-point nature of equilibrium computation in normal-form games with the intertemporal structure of dynamic decision-making. 

\vspace{6pt}
\noindent
In this work, we focus on a fully decentralized information model for MARL, namely \emph{independent} or \emph{radically uncoupled} learning \citep{noah_hardness,FosterPeyton2016}. We study this model in the discounted setting, which is widely used in the control and reinforcement learning literatures for a variety of purposes \citep{daskalakis2023, tessler2019action, zou2026claimhabit, perolat2017learning, SayinZhangLeslieBasarOzdaglar2021}. Here, \emph{radically uncoupled} refers to the restriction that players are not allowed to observe one another's policies, losses, value functions, or random bits. In particular, there is no available shared randomness, communication, policy revelation, or correlation device during learning. Consequently, each player updates her policy using only her own feedback and treats the effect of the other players as part of the environment. Independent learning dynamics are attractive because they scale without the communication and synchronization overhead required by centralized coordination. For instance, in large-scale systems, repeatedly aggregating global information or coordinating agents through a central learner can be costly and or infeasible under bandwidth, latency, or privacy constraints \citet{lian2017}. Moreover, these informational restrictions also rule out many of the mechanisms used in existing 
results for equilibrium computation in Markov games \citep{BaiJinYu2020,WeiEtAl2021, Vlearning}, thereby motivating the study of equilibrium computation under decentralized information structures.

\vspace{6pt}
\noindent
Beyond asymptotic convergence, one typically seeks algorithms that reach an approximate equilibrium with controlled computational cost and sample complexity, for example with polynomial or quasi-polynomial dependence on the natural problem parameters. To achieve these finite-time computational guarantees, centralized or communication-based methods often rely on shared value estimates, policy information, or coordinated equilibrium-computation subroutines. In contrast, under radically uncoupled information, each player must both compute and statistically estimate her update using only her own local feedback. Therefore, the main challenge in the independent learning setting is not merely to design a convergent learning dynamic, but to obtain convergence with meaningful computational and statistical efficiency under radically uncoupled information. For these reasons, in this work we investigate the following question:
\begin{quote}\centering\emph{Are there (quasi-)efficient, radically uncoupled learning algorithms for approximate equilibrium computation in discounted general-sum Markov games?}\end{quote}

\noindent
In normal-form games, the connection between independent learning mechanisms and equilibrium computation is well understood, as independent learning reduces naturally to no-regret learning in this setting: if all players have vanishing external regret, then their average play converges to a coarse correlated equilibrium (CCE), and if all players have vanishing swap regret, their average play converges to a correlated equilibrium (CE) \citep{hannan57,HartMasColell2000,CesaBianchiLugosi2006,BlumMansour2007}. Moreover, it has been shown that optimistic and regularized learning dynamics yield faster self-play rates in normal-form games \citep{DaskalakisDeckelbaumKim2011,SyrgkanisAgarwalLuoSchapire2015,DaskalakisFishelsonGolowich2021,AnagnostidesEtAl2022,SoleymaniPiliourasFarina2025}. 

\vspace{6pt}
\noindent
In contrast, the relationship between independent learning dynamics and equilibrium computation in Markov games remains much less understood. For two-player zero-sum Markov games there is a substantial body of work on decentralized learning, including settings with limited shared information and \emph{radically uncoupled} algorithms with provable guarantees \citep{BaiJinYu2020,WeiEtAl2021,Cai23,Chenetal2024}. Naturally, general-sum Markov games are expected to be more challenging, as there is no common value function or saddle-point structure. The challenge lies not only in the general-sum strategic interactions among multiple players, but also in the dynamic nature of the environment. The unilateral alternative policies may induce a different distribution over future states, and this state distribution depends jointly on the deviating player's policy, the transition dynamics, and the policies used by the other players. Since the opponents may also change their policies over time, these comparison distributions are themselves time-varying and unknown to players. This is a central obstruction to transferring the standard normal-form no-regret argument directly to Markov games.

\vspace{6pt}
\noindent
The existing positive results for general-sum Markov games typically rely on coordination or post-processing mechanisms that are not present in the independent learning problem considered in this paper \citep{ArslanYuksel2017,Vlearning,SongMeiBai2022,MaoBasar2023,daskalakis2023}. A frequently studied class of algorithms include V-learning and related sample-based algorithms \citep{Vlearning,SongMeiBai2022,MaoBasar2023}, which compute CCE or CE policies in decentralized MARL. The equilibrium policies produced by these methods are represented as carefully specified distributions over the learning history and require the players to sample a common index from that history during execution. Notably, the learning phase of these algorithms does not provide any online regret guarantees. Thus, while these methods are decentralized, they do not address the radically uncoupled online-regret objective considered in this paper. Being another closely related algorithm, the so-called SPoCMAR \citep{daskalakis2023} is a decentralized algorithm with polynomial sample and computational complexity for learning a nonstationary Markov-CCE, but the provided regret guarantees rely on the existence of common random bits during learning and for sampling the output policy, where Markov-CCE is an equilibrium with respect to the Markov policy deviations. Consequently, these works do not provide online regret guarantees for the independently played product policies during the learning process. 

\vspace{6pt}
\noindent
Another closely related precursor to our work is \citet{ErezLancewickiShermanKorenMansour2025}. Similar to our work, they study radically uncoupled learning in finite-horizon general-sum Markov games in the self-play setting. Their work, compared to ours, contains further restrictions. First, their regret bounds are obtained with respect to Markov policies, whereas our results allow non-Markov history-dependent policies. Second, their sublinear regret guarantee is obtained when agents may affect one another's losses but not one another's transition dynamics. In contrast, we study finite- and infinite-horizon discounted Markov games with action-dependent transition kernels and aim to develop radically uncoupled learning algorithms with online regret guarantees against non-Markov history-dependent deviations.

\vspace{6pt}
\noindent
On the other hand, there exist hardness results that demonstrate limitations in finding efficient algorithms for the MDPs and Markov games. In the single-agent case, no-regret learning in adversarially changing MDPs can be computationally intractable \citep{marl_hardness2}. Naturally, this intractability also arises in Markov games with adversarial opponents, where both computational and statistical barriers are known \citep{marl_hardness3,marl_hardness4}. 
More recently, motivated by the fact that computing Nash equilibria is $\mathsf{PPAD}$-complete in normal-form games \citep{DaskalakisGoldbergPapadimitriou2009}, it has been shown that similar hardness constraints can persist in Markov games \citep{daskalakis2023, noah_hardness}. In particular, computing stationary Markov-CCEs in discounted general-sum stochastic games is $\mathsf{PPAD}$-hard \citep{daskalakis2023}. Moreover, in finite-horizon general-sum Markov games, any independent learning algorithm that achieves no regret property is also $\mathsf{PPAD}$-hard \citep{noah_hardness}. In contrast, we seek nonstationary policies that approximate CCE in discounted Markov games while avoiding the intractability barriers described above in the self-play setting, which complements the existing literature.

\vspace{6pt}
\noindent
Our contributions in this paper can be summarized as follows.

\begin{enumerate}[leftmargin=1.45em]
    \item 
    In Section~\ref{sec:episodic-algorithm}, we first study finite-horizon
    general-sum Markov games under full-feedback.  We introduce
    a \emph{layered} smoothed-entropy optimistic-online-mirror-descent (OOMD) algorithm,  with an increasing step-size schedule (Algorithm \ref{alg:layered-omd}). Then, we show that this algorithm achieves an
    \( 
        \mathcal O\!\left(T^{-3/(3H+1)}\right)
    \) 
    -approximate CCE with respect to history-dependent deviations
    after \(T\) episodes (Theorem \ref{thm:main}). Then, when the underlying model is unknown to agents, in Subsection~\ref{subsec:episodic-partial-feedback}, under a reachability notion (Definition \ref{def:kappa-zeta-reachability}), we prove a high-probability finite-horizon CCE 
    guarantee with the same exponential-in-horizon dependence up to the polynomial terms (Theorem~\ref{thm:episodic-partial-feedback}).

    \item 
    In Section~\ref{sec:discounted-truncation}, we transfer the finite-horizon
    guarantees established in Section~\ref{sec:episodic-algorithm} to discounted Markov games (Theorem \ref{thm:discounted-transfer}). Unlike the undiscounted case, in the discounted case our algorithms yield
    quasi-polynomial-time and quasi-polynomial-sample guarantees for computing
    inverse-polynomially accurate sparse discounted CCEs for
    the full-feedback and partial feedback cases for both finite-horizon and infinite-horizon Markov games (Corollaries~\ref{cor:disc-full} and~\ref{cor:disc-bandit}). These quasi-polynomial results rely on a horizon length truncation and approximation error bounding arguments to reduce the error bound \(\mathcal O\!\left(T^{-3/(3H+1)}\right)\) to a quasi-polynomial one.

    \item
    In Section~\ref{sec:computational-lower-bound}, we complement the
    quasi-polynomial upper bounds established in Section~\ref{sec:discounted-truncation} for learning discounted-CCE with a complexity-theoretic lower bound.
    In particular, assuming the ``\(\mathsf{ETH}\) for \(\mathsf{PPAD}\)''
    \citep{Babichenko2016}, we show that for every fixed discount factor
    \(\gamma\in(0,1)\cap \mathbb Q\) and every polynomial support bound, there exists an
    inverse-polynomial accuracy level for which no polynomial-time algorithm can
    compute a \emph{sparse} approximate CCE in finite horizon or infinite horizon discounted general-sum Markov games
    (Theorem~\ref{thm:discounted-ppad-eth-lower-main}, Corollary \ref{cor:finite-discounted-ppad-eth-main}) i.e., when players learn independently in a decentralized setting, one cannot calculate approximate CCE in polynomial time.
\end{enumerate}

\section{Preliminaries and Notations}
\label{sec:preliminaries}

{\subsection{Notations}}

\vspace{6pt}
\noindent
For an integer \(q\ge 1\), we write
\( 
    [q]:=\{1,\ldots,q\}.
\) 
For a finite set \(X\), we let
\( 
    \Delta(X)
\) 
denote the probability simplex over \(X\).  
For a vector \(x\in\Delta(X)\), \(x(a)\) denotes the probability of occurrence of action \(a\).  All logarithms are natural base unless
otherwise stated. In the computational statements, every input object is represented by a finite binary string.

\subsection{Undiscounted Finite-horizon Markov games}

Let \(m\) be a positive integer. An \(m\)-player finite-horizon Markov game is specified by the tuple
\( 
    G
    =
    \Bigl(
        H,\{S_h\}_{h=1}^{H+1},
        \{A_i\}_{i=1}^m,
        \{P_h\}_{h=1}^H,
        \{\ell_h^i\}_{i\in[m],h\in[H]},
        s_1
    \Bigr).
\) 
We will denote by \(\mathcal M = [m]\) the set of players. The components of the the Markov game are defined as follows:

\begin{itemize}

    \item Horizon length is denoted by \(H\in\mathbb N_+\).

    \item We denote the finite set of states at layer \(h\) by \(S_h\). The initial state is deterministic and we denote the initial state by $s_1$.  The layer \(H+1\) is terminal and has no actions or
    costs.  We write
    \( 
        \mathcal S := \bigsqcup_{h=1}^{H+1} S_h
    \) , and \(S:=|\mathcal S|\).

    \item The action set of player \(i\in  \mathcal M\) is denoted by \(A_i\), which we will assume to be finite.  The joint action 
    space of all players will be denoted by
    \( 
        A := \prod_{i=1}^m A_i .
    \) 
    We write joint actions of the players as
    \( 
        \bm a=(a_1,\ldots,a_m)\in A,
    \) 
    and write
    \( 
        \bm a_{-i}:=(a_j)_{j\neq i}\in A_{-i}:=\prod_{j\neq i} A_j
    \) 
    for the action profile that excludes the action of player \(i\). By
    \( 
        A_{\max}:=\max_{i\in[m]} |A_i|,
    \)  
    we will denote the maximum available cardinality among all players.
 
    \item For any given \(h\in[H]\), by
    \( 
        P_h:S_h\times A\to \Delta(S_{h+1})
    \) 
    we denote the (state) transition probability kernel of all players.
    The quantity \(P_h(s'\mid s,\bm a)\) represents the probability of transitioning to
    \(s'\in S_{h+1}\) from \(s\in S_h\) when the joint action is
    \(\bm a\in A\).

    \item For each player \(i\in[m]\) and layer \(h\in[H]\), the cost function
    is \( 
        \ell_h^i:S_h\times A\to[0,1].
    \) 
\end{itemize}

\paragraph{General (non-Markov) policies.}
For \(h\in[H]\), let \(\mathcal H_h\) denote the set of histories before the
layer-\(h\) action, 
\( 
    \mathcal H_h
    :=
    S_1\times A\times S_2\times A\times\cdots\times S_{h-1}\times A\times S_h .
\) 
An element of \(\mathcal H_h\) is written as
\( 
    \tau_h=(s_1,\bm a_1,s_2,\bm a_2,\ldots,s_{h-1},\bm a_{h-1},s_h).
\) Furthermore, when needed, we let \( s(\tau_h) \in S_h\) be the concluding state of the history \(\tau_h \in \mathcal H_h \) at layer \(h\).
A non-Markov policy of player \(i\) is a sequence of maps
\( 
    \pi^i=\{\pi_h^i\}_{h=1}^H,
    \; \pi_h^i(\cdot\mid\tau_h)\in\Delta(A_i)
    \;
    \text{for every }\tau_h\in\mathcal H_h .
\)
Let \(\Pi_i^{\mathrm{gen}}\) denote the set of all such finite-horizon policies of player \(i\), and define
\( 
    \Pi^{\mathrm{gen}}:=\prod_{i=1}^m \Pi_i^{\mathrm{gen}}.
\) 
A non-Markov product policy profile is denoted by
\( 
    \bm\pi=(\pi^1,\ldots,\pi^m)\in\Pi^{\mathrm{gen}}.
\) 
In particular, for player \(i \in \mathcal M\) a general (non-Markov) policy \(\pi_i \in \Pi_i^{\mathrm{gen}}\) is specified by a collection \(\pi^i = (\pi^i_1, \cdots, \pi^i_H )\), where \(\pi^i_h:\tau_h \to \Delta(A_i) \). 
Under such a profile, conditional on the history, \(\tau_h\), the players' 
actions are sampled independently as
\( 
    a_h^i\sim \pi_h^i(\cdot\mid\tau_h),
    \; i\in[m],  h \in [H] .
\)  

\paragraph{Markov policies.}
A Markov policy of player \(i\) is a sequence
\( 
    \pi^i=\{\pi_h^i\}_{h=1}^H,
    \;
    \pi_h^i(\cdot\mid s)\in\Delta(A_i)
    \) \( 
    \text{for every }s\in S_h .
\) 
Let \(\Pi_i^{\mathrm{markov}}\) denote the set of all such finite-horizon policies of player \(i\), and define
\( 
    \Pi^{\mathrm{markov}}:=\prod_{i=1}^m \Pi_i^{\mathrm{markov}}
\) to
denote the space of product Markov policies, where each agent \(i\) independently follows a policy
in \(\Pi_i^{\mathrm{markov}}\) for a given state. In particular, a Markov product policy \(\bm \pi \in \Pi^{\mathrm{markov}}\) is specified by a collection \(\bm \pi = (\bm \pi_1, \cdots, \bm \pi_H)\), where \(\bm \pi_h: S_h \to \Delta(A_1) \times \cdots \times \Delta (A_m) .\)
At state \(s\in S_h\), it induces the product distribution
\( 
    \bm\pi_h(\bm a\mid s)
    :=
    \prod_{j=1}^m \pi_h^j(a_j\mid s),
    \;
    \bm a=(a_1,\ldots,a_m)\in A .
\) 
For player \(i\), write
\( 
    \bm\pi^{-i}
    :=
    (\pi^1,\ldots,\pi^{i-1},\pi^{i+1},\ldots,\pi^m)
\) 
for the policy profile of all players except \(i\).  If \(\rho^i\) is another
policy of player \(i\), either Markov or non-Markov, we write
\( 
    \rho^i\odot\bm\pi^{-i}
\)
for the profile obtained by replacing player \(i\)'s policy with \(\rho^i\) and
leaving the opponents' policies equal to \(\bm\pi^{-i}\).  In particular, if
\(\rho^i\) is Markov, then for \(s\in S_h\),
\( 
    (\rho^i\odot\bm\pi^{-i})_h(\bm a\mid s)
    :=
    \rho_h^i(a_i\mid s)
    \prod_{j\neq i}\pi_h^j(a_j\mid s).
\) 

\paragraph{Repeated interaction.}
The agents interact with environment over \(T\) episodes.  At the beginning of
episode \(t\in[T]\), each player \(i\) chooses a Markov policy \(\pi_t^i\).  Let
\( 
    \bm\pi_t=(\pi_t^1,\ldots,\pi_t^m)
\) 
denote the resulting product Markov profile. Each episode starts from \(s_{t,1}=s_1\).  For each \(h\in[H]\), player \(i\)
samples
\( 
    a_{t,h}^i\sim \pi_{t,h}^i(\cdot\mid s_{t,h}),
\) 
the joint action is
\( 
    \bm a_{t,h}:=(a_{t,h}^1,\ldots,a_{t,h}^m),
\) 
player \(i\) incurs cost
\( 
    \ell_h^i(s_{t,h},\bm a_{t,h}),
\) 
and the next state is sampled as
\( 
    s_{t,h+1}\sim P_h(\cdot\mid s_{t,h},\bm a_{t,h}).
\) 
The feedback observed by each player after the episode depends on the feedback
model defined below.

\begin{definition}[Full-feedback]
At the end of episode \(t\), player \(i\) observes the exact
\( 
    Q_t^i(s,a_i)
    =
    Q_h^{i,\bm\pi_t}(s,a_i),
    \;
    \forall h\in[H],\ s\in S_h,\ a_i\in A_i .
\) 
\end{definition}

\begin{definition}[Partial feedback]
At the end of episode \(t\), player \(i\) observes only its own trajectory
information
\( 
    \Bigl(
        s_{t,1},a_{t,1}^i,c_{t,1}^i,
        s_{t,2},a_{t,2}^i,c_{t,2}^i,
        \ldots,
        s_{t,H},a_{t,H}^i,c_{t,H}^i,
        s_{t,H+1}
    \Bigr),
\) 
where
\( 
    c_{t,h}^i
    :=
    \ell_h^i(s_{t,h},\bm a_{t,h})
\). 
\end{definition}

\vspace{6pt}
\noindent
For a policy profile \(\bm\pi\in\Pi^{\mathrm{markov}}\), let \(\mathbb P^{\bm\pi}\) denote the distribution over trajectories induced by the policy profile \(\bm\pi\) and the transition kernels \(\{P_h\}_{h=1}^H\), and define
\( 
    V_{H+1}^{i,\bm\pi}\equiv 0.
\)  All expectations \(\mathbb E^{\bm\pi}\) below are taken with respect to this trajectory distribution. For \(h\in[H]\), \(s\in S_h\), and \(i\in[m]\), we define the value
function
\( 
    V_h^{i,\bm\pi}(s)
    :=
    \mathbb E^{\bm\pi}
    \left[
        \sum_{r=h}^H \ell_r^i(s_r,\bm a_r)
        \,\middle|\, s_h=s
    \right].
\) 
Similarly, the state-action value is given as
\( 
    Q_h^{i,\bm\pi}(s,a_i)
    :=
    \mathbb E_{\bm a_{-i}\sim \bm\pi_h^{-i}(\cdot\mid s)}
    \left[
        \ell_h^i(s,\bm a)
    +
    \sum_{s'\in S_{h+1}}
        P_h(s'\mid s,\bm a)\,
        V_{h+1}^{i,\bm\pi}(s')
    \right],
\) where the opponents follow the
\(\bm\pi^{-i}\). 
Therefore,
\( 
    V_h^{i,\bm\pi}(s)
    =
    \mathbb E_{a_i\sim\pi_h^i(\cdot\mid s)}
    \left[
        Q_h^{i,\bm\pi}(s,a_i)
    \right].
\) 
For \(s_1\), we have
\( 
    V^{i,\bm\pi}(s_1):=V_1^{i,\bm\pi}(s_1).
\)

\paragraph{The induced MDP of a player.}
Fix an episode \(t\), a player \(i\), and the opponents' policies
\(\bm\pi_t^{-i}\). Then, player \(i\) faces an induced finite-horizon MDP
\( 
    M_t^i
    :=
    \Bigl(
        H,\{S_h\}_{h=1}^{H+1},A_i,
        \{P_{t,h}^i\}_{h=1}^H,
        \{\ell_{t,h}^i\}_{h=1}^H,
        s_1
    \Bigr),
\) 
where, for \(s\in S_h\), \(a_i\in A_i\), and \(s'\in S_{h+1}\),
\( 
    \ell_{t,h}^i(s,a_i)
    :=
    \mathbb E_{\bm a_{-i}\sim\bm\pi_{t,h}^{-i}(\cdot\mid s)}
    \left[
        \ell_h^i\bigl(s,(a_i,\bm a_{-i})\bigr)
    \right],
\) 
and
\( 
    P_{t,h}^i(s'\mid s,a_i)
    :=
    \mathbb E_{\bm a_{-i}\sim\bm\pi_{t,h}^{-i}(\cdot\mid s)}
    \left[
        P_h\bigl(s'\mid s,(a_i,\bm a_{-i})\bigr)
    \right].
\) 
For a Markov policy \(\mu^i\) in this induced MDP, define
\( 
    V_{t,H+1}^{i,\mu^i}\equiv 0
\) 
and, then, we have
\[
    Q_{t,h}^{i,\mu^i}(s,a_i)
    :=
    \ell_{t,h}^i(s,a_i)
    +
    \sum_{s'\in S_{h+1}}
        P_{t,h}^i(s'\mid s,a_i)\,
        V_{t,h+1}^{i,\mu^i}(s'),
\qquad
    V_{t,h}^{i,\mu^i}(s)
    :=
    \mathbb E_{a_i\sim\mu_h^i(\cdot\mid s)}
    \left[
        Q_{t,h}^{i,\mu^i}(s,a_i)
    \right].
\]
These induced-MDP quantities agree with the corresponding Markov-game
quantities,
\( 
    V_{t,h}^{i,\mu^i}(s)
    =
    V_h^{i,\mu^i\odot\bm\pi_t^{-i}}(s),
    \;
    Q_{t,h}^{i,\mu^i}(s,a_i)
    =
    Q_h^{i,\mu^i\odot\bm\pi_t^{-i}}(s,a_i).
\) 
For the actually played policy \(\pi_t^i\), we use the shorthand
\( 
    V_{t,h}^{i,\pi}(s)
    :=
    V_h^{i,\pi_t^i\odot\bm\pi^{-i}_t}(s),
    \;
    Q_{t,h}^{i,\pi}(s,a_i)
    :=
    Q_h^{i,\pi_t^i\odot\bm\pi^{-i}_t}(s,a_i),
    \; s\in S_h .
\)  Moreover, we drop policy-dependent indices from the value-functions' notation whenever they are clear from the context.

\paragraph{Distributional policies.}
We will also consider distributions over non-Markov
general policies; we refer to such elements of \(\Delta(\Pi^{\mathrm{gen}})\) as
\emph{distributional policies}. Playing a distributional policy
\(P \in \Delta(\Pi^{\mathrm{gen}})\) consists of first sampling a
randomized policy
\( 
    \bm\pi \sim P,
\) 
and then executing the sampled policy \(\bm\pi\).

\vspace{6pt}
\noindent 
For any policy \(\bm\pi \in \Pi^{\mathrm{gen}}\), let
\( 
    \delta_{\bm\pi} \in \Delta(\Pi^{\mathrm{gen}})
\) 
denote the distribution placing unit
mass on \(\bm\pi\). 
Accordingly, the empirical distribution over the iterates
\(\bm\pi_1,\ldots,\bm\pi_T\) is
\( 
    \widehat{\Pi}_T
    :=
    \frac1T \sum_{t=1}^T \delta_{\bm\pi_t}  \in  \Delta(\Pi^{\mathrm{gen}}).
\) 
More generally, if \(\mathfrak P\in\Delta(\Pi^{\mathrm{gen}})\) is a
distribution over policy profiles, we define
\( 
    V^{i,\mathfrak P}(s_1)
    :=
    \mathbb E_{\bm\sigma\sim\mathfrak P}
    \left[
        V^{i,\bm\sigma}(s_1)
    \right].
\) 
Then, for a deviation \(\mu^i\in\Pi_i^{\mathrm{gen}}\), we let
\( 
    V^{i,\mu^i\odot\mathfrak P^{-i}}(s_1)
    :=
    \mathbb E_{\bm\sigma\sim\mathfrak P}
    \left[
        V^{i,\mu^i\odot\bm\sigma^{-i}}(s_1)
    \right].
\) 

\begin{definition}[General-policy external regret]\label{def:gen-ext-reg}
Given product Markov policy profiles
\(\bm\pi_1,\ldots,\bm\pi_T\), the general-policy external regret of player
\(i\) is
\( 
    \operatorname{Reg}^{\mathrm{gen},i}_T
    :=
    \sup_{\mu^i\in\Pi_i^{\mathrm{gen}}}
    \sum_{t=1}^T
    \left[
        V^{i,\bm\pi_t}(s_1)
        -
        V^{i,\mu^i\odot\bm\pi_t^{-i}}(s_1)
    \right].
\) 
\end{definition}

\begin{definition}[Coarse Correlated Equilibrium]\label{def:gen-ext-cce}
The empirical distribution
\( 
    \widehat{\Pi}_T
    =
    \frac1T\sum_{t=1}^T \delta_{\bm\pi_t}
\) 
is an \(\varepsilon\)-approximate CCE if, for every player \(i\in[m]\),
\( 
    \sup_{\mu^i\in\Pi_i^{\mathrm{gen}}}
    \left[
        V^{i,\widehat{\Pi}_T}(s_1)
        -
        V^{i,\mu^i\odot\widehat{\Pi}_T^{-i}}(s_1)
    \right]
    \le \varepsilon .
\) 
\end{definition}

\subsection{Discounted Markov games}
\label{subsec:discounted-preliminaries}

We use the same notation for finite and infinite-horizon discounted games.  Let
\( 
    \bar H\in\mathbb N_+\cup\{\infty\}
\) 
denote the horizon.  When \(\bar H<\infty\), an \(m\)-player
finite-horizon discounted Markov game is specified by
\( 
    G_{\gamma,\bar H}
    =
    \Bigl(
        \bar H,
        \{S_h\}_{h=1}^{\bar H+1},
        \{A_i\}_{i=1}^m,
        \{P_h\}_{h=1}^{\bar H},
        \{\ell_h^i\}_{i\in[m],h\in[\bar H]},
        \gamma,
        s_{\mathrm{1}}
    \Bigr).
\) 
When \(\bar H=\infty\), we recover the discounted infinite-horizon
Markov game
\( 
    G_\gamma
    =
    \Bigl(
        S,
        \{A_i\}_{i=1}^m,
        P,
        \{\ell^i\}_{i=1}^m,
        \gamma,
        s_{\mathrm{1}}
    \Bigr)
\) . In both cases, we will denote by \(\mathcal M = [m]\) the set of players. Corresponding games are defined as follows.

\begin{itemize}
    \item  For horizon \(\bar H\), we use the convention
    \( 
        [\bar H]=
        \begin{cases}
            \{1,\ldots,\bar H\}, & \bar H<\infty,\\
            \mathbb N_+, & \bar H=\infty.
        \end{cases}
    \) 

    \item In the finite-horizon
    case, the state spaces \(S_1,\ldots,S_{\bar H+1}\) are layer-labeled, and
    \(S_{\bar H+1}\) is terminal.  In the infinite-horizon case, we have \(S_h=S\) for all
    \(h\ge 1.\) 
    The initial state is deterministic and we denote it by $s_1$.
    
    \item The set \(A_i\) is the finite action set of player \(i \in \mathcal M\).  The joint action
    space is
    \( 
        A:=\prod_{i=1}^m A_i .
    \) 
    We write a joint action as
    \( 
        \bm a=(a_1,\ldots,a_m)\in A,
    \) 
    and we write
    \( 
        \bm a_{-i}:=(a_j)_{j\neq i}\in A_{-i}:=\prod_{j\neq i}A_j
    \).

    \item For each layer \(h\in[\bar H]\), the transition kernel is
    \( 
        P_h:S_h\times A\to\Delta(S_{h+1}).
    \) 
    In the infinite-horizon case, this reduces to the stationary kernel
    \( 
        P:S\times A\to\Delta(S).
    \)     

    \item For each player \(i\in[m]\) and layer \(h\in[\bar H]\), the cost
    function is
    \( 
    \ell_h^i:S_h\times A\to[0,1].
    \) 
    In the infinite-horizon case, the cost function is
    \( 
        \ell^i:S\times A\to[0,1].
    \) 

    \item The discount factor is
    \( 
        \gamma\in(0,1).
    \) 
\end{itemize}

\paragraph{General (non-Markov) and Markov policies.}
For each \(h\le \bar H\), let
\( 
    \mathcal H_h^{\bar H}
    :=
    S_1\times A\times S_2\times A\times\cdots\times S_{h-1}\times A\times S_h
\) 
be the set of histories before the time-\(h\) action.  When \(\bar H=\infty\),
this definition is understood with the convention \(S_h=S\) for every \(h\).
An element of \(\mathcal H_h^{\bar H}\) is written as
\( 
    \tau_h=(s_1,\bm a_1,s_2,\bm a_2,\ldots,s_{h-1},\bm a_{h-1},s_h).
\) 

\vspace{6pt}
\noindent 
A non-Markov policy of player \(i\) over horizon \(\bar H\) is a sequence of maps
\( 
    \mu^i=\{\mu_h^i\}_{h\le \bar H},
    \; \mu_h^i(\cdot\mid\tau_h)\in\Delta(A_i)
    \;
    \text{for every }\tau_h\in\mathcal H_h^{\bar H}.
\) 
Let \(\Pi_i^{\mathrm{gen},\bar H}\) denote the set of such policies of player
\(i\), and define
\( 
    \Pi^{\mathrm{gen},\bar H}
    :=
    \prod_{i=1}^m \Pi_i^{\mathrm{gen},\bar H}.
\) 
For \(\bar H=\infty\), we also write
\(\Pi_i^{\mathrm{gen},\infty}\) and
\(\Pi^{\mathrm{gen},\infty}\).

\vspace{6pt}
\noindent 
A Markov policy of player \(i\) over horizon \(\bar H\) is a sequence
\( 
    \pi^i=\{\pi_h^i\}_{h\le \bar H},
    \;
    \pi_h^i(\cdot\mid s)\in\Delta(A_i),
    \; s\in S_h.
\) 
In the infinite-horizon case, as a special case our notion of Markov policies corresponds to non-stationary Markov policies. Let \(\Pi_i^{\mathrm{markov},\bar H}\) denote the set of such policies of player
\(i\), and define
\( 
    \Pi^{\mathrm{markov},\bar H}
    :=
    \prod_{i=1}^m \Pi_i^{\mathrm{markov},\bar H}.
\) 
For \(\bar H=\infty\), we also write
\(\Pi_i^{\mathrm{markov},\infty}\) and
\(\Pi^{\mathrm{markov},\infty}\).

\paragraph{Discounted values.}
For a policy profile \(\bm\sigma\in\Pi^{\mathrm{gen},\bar H}\) and a finite
\(L\le \bar H\), define the \(L\)-truncated discounted cost of player \(i\) by
\( 
    J_{i,L}^{\gamma}(\bm\sigma)
    :=
    \mathbb E^{\bm\sigma}
    \left[
        \sum_{h=1}^{L}\gamma^{h-1}\ell_h^i(s_h,\bm a_h)
        \,\middle|\,
        s_1=s_{\mathrm{init}}
    \right].
\) 
The full discounted objective is
\( 
    J_{i,\bar H}^{\gamma}(\bm\sigma)
    :=
    \mathbb E^{\bm\sigma}
    \left[
        \sum_{h=1}^{\bar H}\gamma^{h-1}\ell_h^i(s_h,\bm a_h)
        \,\middle|\,
        s_1
    \right],
\) 
In the infinite-horizon case,
we also write
\( 
    J_i^\gamma(\bm\sigma)
    :=
    J_{i,\infty}^\gamma(\bm\sigma).
\) 
Since \(0\le \ell_h^i\le 1\), for every profile \(\bm\sigma\),
\( 
    0\le J_{i,\bar H}^{\gamma}(\bm\sigma)
    \le
    \sum_{h=1}^{\bar H}\gamma^{h-1}
    \le
    \frac{1}{1-\gamma}.
\) 

\vspace{6pt}
\noindent 
For a distribution \(\mathfrak P\in\Delta(\Pi^{\mathrm{gen},\bar H})\) over
policy profiles, define
\( 
    J_{i,\bar H}^{\gamma}(\mathfrak P)
    :=
    \mathbb E_{\bm\sigma\sim\mathfrak P}
    \left[J_{i,\bar H}^{\gamma}(\bm\sigma)\right].
\) 
For a unilateral deviation \(\mu^i\in\Pi_i^{\mathrm{gen},\bar H}\), define
\( 
    J_{i,\bar H}^{\gamma}(\mu^i\odot\mathfrak P^{-i})
    :=
    \mathbb E_{\bm\sigma\sim\mathfrak P}
    \left[J_{i,\bar H}^{\gamma}(\mu^i\odot\bm\sigma^{-i})\right].
\) Then, we define the corresponding discounted CCE and regret notions as follows.

\begin{definition}
\label{def:discounted-general-policy-regret}
For a sequence of \(\bar H\)-horizon product policy profiles
\(\bar{\bm{\pi}}_1,\ldots,\bar{\bm{\pi}}_T\), the discounted general-policy regret of
player \(i\) is defined as
\( 
    \Reg_{\gamma,\bar H,T}^{\mathrm{gen},i}
    :=
    \sup_{\mu^i\in\Pi_i^{\mathrm{gen},\bar H}}
    \sum_{t=1}^T
    \left[
        J_{i,\bar H}^{\gamma}(\bar{\bm{\pi}}_t)
        -
        J_{i,\bar H}^{\gamma}(\mu^i\odot\bar{\bm{\pi}}_t^{-i})
    \right].
\) 
\end{definition}

\begin{definition}
\label{def:disc-cce-unified}
The empirical distribution
\( 
    \widehat{\Pi}_T
\) 
is an \(\varepsilon\)-approximate CCE if, for every player \(i\in[m]\),
\( 
    \sup_{\mu^i\in\Pi_i^{\mathrm{gen},\bar H}}
    \left[
        J_{i,\bar H}^{\gamma}(\widehat{\Pi}_T)
        -
        J_{i,\bar H}^{\gamma}(\mu^i\odot\widehat{\Pi}_T^{-i})
    \right]
    \le
    \varepsilon.
\) 
If \(\bar H=\infty\), it is the discounted infinite-horizon CCE.
\end{definition}

\subsection{Truncated Markov Games}\label{subsect:truncated:quaso}
Let \(L\le \bar H\). For a given discounted Markov game \(G_{\gamma,\bar H}\), we define a corresponding \(L\)-step discounted truncation
\(G_{\gamma,\bar H}^{[L]}\) whose one-step cost function is defined as
\( 
    c_h^i(s,\bm a)
    :=
    \gamma^{h-1}\ell_h^i(s,\bm a),
    \;
    h\in[L].
\) 
The state spaces, action sets, transition kernels, and initial state of \(G_{\gamma,\bar H}^{[L]}\) are inherited from the original game for layers \(1,\ldots,L\).

\vspace{6pt}
\noindent 
For an \(\bar H\)-horizon policy
\( 
\mu^i = (\mu^i_1,\mu^i_2,\ldots,\mu^i_{\bar H}),
\) define
\( 
\mu^{i,[L]} := (\mu^i_h)_{h=1}^{L}.
\) 
That is, \(\mu^{i,[L]}\) is the \(L\)-step policy obtained by truncating \(\mu^i\) after stage \(L\). We define policy \(\bm{\pi}^{L,\mathrm{ref}}\) for the layers after
\(L\), in which every player plays uniformly at every state.  For an \(L\)-step product profile \(\bm{\pi}^{[L]}\), let
\( 
    \mathrm{Ext}_L(\bm{\pi}^{[L]})
\) 
denote the \(\bar H\)-horizon profile that plays \(\bm{\pi}^{[L]}\) during layers
\(1,\ldots,L\) and then plays \(\bm{\pi}^{L,\mathrm{ref}}\) afterward. Finally, define the discounted tail after layer \(L\) or truncation error by
\( 
    \tau_{L,\bar H}(\gamma)
    :=
    \sum_{h=L+1}^{\bar H}\gamma^{h-1}\le
    \frac{\gamma^L}{1-\gamma}.
\)

\section{Main Algorithm and Results for Episodic Markov Games}
\label{sec:episodic-algorithm}

Let \( 
    G
    =
    \Bigl(
        H,\{S_h\}_{h=1}^{H+1},
        \{A_i\}_{i=1}^m,
        \{P_h\}_{h=1}^H,
        \{\ell_h^i\}_{i\in[m],h\in[H]},
        s_1
    \Bigr)
\)
be a finite-horizon Markov game. In this section, we establish the undiscounted finite-horizon regret guarantees for $G$ that underlie our later discounted-game results.
We consider undiscounted general-sum Markov games within episodic framework in the self-play setting, where the players repeatedly generate product Markov profiles and update their own policies independently. The performance criterion is regret with respect to history-dependent deviations, which, through the standard no-regret-to-CCE connection \citep{CesaBianchiLugosi2006}, yields approximate CCE for the empirical distribution of play.

\vspace{6pt}
\noindent
The main obstruction, compared with normal-form games, is that the regret comparison in a Markov game is weighted by the history distribution induced by the policy deviations.  A unilateral deviation of policies may change the distribution of future states, and therefore the future decision points at which the deviation is evaluated.  Hence, the comparison places more weight on histories that the deviating policy visits frequently and less weight on histories that it reaches rarely.

\vspace{6pt}
\noindent
From the learner's perspective, these history-occupancy weights are unknown and time-varying: they depend on the transition kernel, the deviation being compared against, and the opponents' evolving policies.  Thus, the analysis must control the \emph{drift} of these unknown weights in addition to the online learning dynamics.  We first handle this issue in the full-feedback model, and then extend the argument to the trajectory-level partial-feedback model by introducing value estimation and an exploration floor under a reachability condition.

\subsection{Episodic Markov games under full-feedback}
\label{subsec-full}

In this subsection, we present 
the Layered Optimistic Online Mirror Descent (OOMD) algorithm for the self-play setting in game $G$ (see Algorithm~\ref{alg:layered-omd}) and prove its regret guarantee under full-feedback.  The algorithm follows an independent self-play structure: at each player--state--layer tuple \((i,s,h)\), player \(i\) performs a local policy update using the state--action value vector generated by the product Markov profile played in the previous episode.

\begin{algorithm}[H]
\caption{Layered Optimistic Online Mirror Descent}
\label{alg:layered-omd}
\begin{algorithmic}[1]
\State \textbf{Input:}
Finite-horizon game \(G\), horizon \(H\), episode count \(T\), base step-size
\(\eta_0\in\mathbb{R}\).

\State Define
\( 
    \alpha_h:=\frac{3(H-h)+1}{3H+1},
    \;
    \eta_h:=\eta_0T^{-\alpha_h},
    \;  \bar x^i(a_i):=\frac{1}{|A_i|},\lambda_i=1/|A_i|,
    \; a_i\in A_i, h\in[H].
\) 

\State For every \(i\in[m]\), \(h\in[H]\), and \(s\in S_h\), set
\( 
    \widetilde x_{0,h}^{i,s}:=\bar x^i\in\argmin_{x\in\Delta(A_i)}\Psi_i(x),
    \;
    Q_0^i(s,\cdot)\equiv 0\in\mathbb R^{A_i}.
\) 

\For{\(t=1,\ldots,T\)}

    \For{each player \(i\in[m]\), layer \(h\in[H]\), and state \(s\in S_h\)}
        \Comment{Optimistic mirror step}

        \State 
        \( 
            x_{t,h}^{i,s}
            =
            \argmin_{x\in\Delta(A_i)}
            \left\{
                \eta_h\,
                \bigl\langle Q_{t-1,h}^i(s,\cdot),x\bigr\rangle
                +
                D_i\!\left(x,\widetilde x_{t-1,h}^{i,s}\right)
            \right\}.
        \) 

        \State Set the played policy at \((h,s)\) by
        \( 
            \pi_{t,h}^i(\cdot\mid s)
            =
            x_{t,h}^{i,s}(\cdot).
        \) 

    \EndFor

    \State Execute episode \(t\) using the
    \( 
        \bm\pi_t=(\pi_t^1,\ldots,\pi_t^m).
    \) 

    \State Each player \(i\) observes
    \( 
        Q_{t,h}^i(s,a_i),
        \;
        \forall h\in[H],\ s\in S_h,\ a_i\in A_i.
    \) 

    \For{each player \(i\in[m]\), layer \(h\in[H]\), and state \(s\in S_h\)}
        \Comment{Mirror-descent update}

        \State 
        \( 
            \widetilde x_{t,h}^{i,s}
            =
            \argmin_{x\in\Delta(A_i)}
            \left\{
                \eta_h\,
                \bigl\langle Q_{t,h}^i(s,\cdot),x\bigr\rangle
                +
                D_i\!\left(x,\widetilde x_{t-1,h}^{i,s}\right)
            \right\}.
        \) 
    \EndFor

\EndFor
\end{algorithmic}
\end{algorithm}

\vspace{6pt}
\noindent
The full-feedback setting removes the statistical difficulty of estimating value functions from trajectories, since each player observes the exact vectors \(Q^i_{t,h}(s,\cdot)\) after every episode.  The remaining challenge is therefore purely dynamical: each player's OOMD updates must be stable enough to control the induced drift of the hidden history-occupancy weights.  Algorithm~\ref{alg:layered-omd} addresses this difficulty by combining a smoothed-entropy regularizer with an increasing layer-wise step-size schedule.

\vspace{6pt}
\noindent
Our first design choice is the smoothed entropy function as the regularizer choice.
Given a smoothing parameter
\(\lambda_i>0\), the smoothed negative-entropy regularizer over \(\Delta(A_i)\) is defined as
\[
    \Psi_i(x)
    :=
    \sum_{a\in A_i}
    \bigl(x(a)+\lambda_i\bigr)
    \log\bigl(x(a)+\lambda_i\bigr),
    \qquad
    x\in \Delta(A_i).
\]
The associated Bregman divergence is
\( 
    D_i(x,y)
    =
    \sum_{a\in A_i}
    \bigl(x(a)+\lambda_i\bigr)
    \log
    \frac{x(a)+\lambda_i}{y(a)+\lambda_i},
    \;
    x,y\in \Delta(A_i).
\) Here, the idea of shift $\lambda_i$ is reminiscent of the fixed-share smoothing idea used in
tracking and adaptive-regret algorithms
\cite{Herbster98,CesaBianchiGaillardLugosiStoltz2012}. To the best of our knowledge, this type of \emph{fixed-share} idea, implemented through regularizers, was first proposed in~\cite{shiftedent}, although we use
it for a different purpose. In our setting, the shift provides a bounded
Bregman divergence on the full simplex and a clean Lipschitz bound,
while still allowing the comparator
\(\mu^i_h(\cdot\mid\tau_h)\) to be an arbitrary distribution in
\(\Delta(A_i)\). Furthermore, under the choice
\(\lambda_i=1/|A_i|\), the upper bound over $D_i(x,y)$ grows only logarithmically with the action-set size.

\vspace{6pt}
\noindent 
The second design choice is the increasing step-size \vspace{1.5pt}schedule. For a given horizon \(H\), defining 
\( 
    \alpha_h:=\frac{3(H-h)+1}{3H+1} 
    \), we set our learning-rate as \vspace{1.5pt}
\( 
    \eta_h:=\eta_0 T^{-\alpha_h},
    \; h\in[H], 
\) given \(\eta_0\in(0,1]\). Simply, our step-size schedule is designed to make the drift of unknown weights manageable. In particular, the earlier
layers use smaller learning rates, and thus their policies move \emph{slowly} and do not create abrupt changes over later histories. Later layers
use larger learning rates, allowing the downstream parts of the game to adapt more \emph{quickly} by providing them a more stable regime. In a  sense, this creates a nested \emph{time-scale separation} across the horizon as deep layers react faster, while earlier layers remain comparatively stable.  
Combined with optimism, this schedule turns the
problem into a controlled collection of slowly varying weighted online learning problems.

\vspace{6pt}
\noindent
This layer-wise schedule is mathematically reminiscent of increasing
learning-rate schemes used in adaptive online learning and adversarial
bandit/MDP analyses
\cite{bubeck2017kernel,agarwal2017corralling,wei2018adaptive,lee2020bias},
but the motivation and exact schedule here is fundamentally different. In those settings,
learning rates increase across episodes to improve exploration by discounting
early high-variance loss estimates and avoiding commitment to
suboptimal actions. In contrast, our schedule is designed to make the drift of unknown history-occupancy weights manageable. 

\vspace{6pt}
\noindent
Formally, the performance
of Algorithm \ref{alg:layered-omd} is provided in the statement below.

\begin{theorem}
\label{thm:main}
Suppose all players run Algorithm~\ref{alg:layered-omd}. Define
\[
C_*^{\mathrm{SE}}
:=
\max_{i\in[m]}
\left\{
\frac{2H\log(|A_i|+1)}{\eta_0}
+12mH^3\log(|A_i|+1)
+\eta_0H^3
+144\eta_0^3m^2H^7
\right\}.
\]
Given \(\epsilon\in(0,1]\), choose
\( 
T_\varepsilon
:=
\left\lceil
\left(
\frac{C_*^{\mathrm{SE}}}{\varepsilon}
\right)^{(3H+1)/3}
\right\rceil .
\) 
Then, after running Algorithm~\ref{alg:layered-omd} for \(T_\varepsilon\) episodes, for every player
\(i\in[m]\), 
\( 
\mathrm{Reg}^{\mathrm{gen},i}_{T_\varepsilon}
\le
\varepsilon T_\varepsilon .
\) 
Equivalently, the empirical distribution
\( 
\widehat{\Pi}_{T_\varepsilon}
\) 
is an \(\varepsilon\)-approximate CCE with respect to all history-dependent deviations.
\end{theorem}

\begin{proof}[Proof Outline.]

We provide here an overview of the proof of Theorem~\ref{thm:main}.  The full
argument is deferred to Appendix~\ref{app:proof-main}.  The proof follows three main steps.  First,
the Markov-game regret is rewritten as a sum of local weighted regret terms.
Second, the movement of the hidden weights and the movement of the \(Q\)-functions
are controlled through the stability of the policy sequences.  Third, the
layer-wise learning-rate schedule is chosen so that all resulting terms have
the same order.

\vspace{6pt}
\noindent
For a player \(i\) and a history-dependent policy \(\mu^i\in\Pi_i^{\mathrm{gen}}\),
let \(d_{t,h}^{i,\mu}\in\Delta(\mathcal H_h)\) denote the distribution over
histories at layer \(h\), \(\tau_h\), generated by \(\mu^i\) and the opponents' episode-\(t\)
Markov profile \(\bm{\pi}_t^{-i}\). For each layer \(h\), define the local weighted
regret
\[
    R_{i,h}(\mu^i)
    :=
    \sum_{\tau_h\in\mathcal H_h}
    \sum_{t=1}^T
    d_{t,h}^{i,\mu}(\tau_h)
    \left\langle
        Q_{t,h}^i(s(\tau_h),\cdot),
        \pi_{t,h}^i(\cdot\mid s(\tau_h))
        -
        \mu_h^i(\cdot\mid\tau_h)
    \right\rangle .
\]
The weighted value-difference decomposition, Lemma~\ref{lem:history-vdiff},
shows that the regret against \(\mu^i\) is exactly
\( 
    \sum_{t=1}^T
    \!\!\left[
        V^{i,\pi}_{t,h}(s_1)
        \!-\!
        V^{i,\mu^i}_{t,h}(s_1)
    \right]\!\!
    = \!\!
    \sum_{h=1}^H \!\!R_{i,h}(\mu^i).
\) 
Thus, the problem reduces to bounding \(R_{i,h}(\mu^i)\) for all $h$.  

\vspace{6pt}
\noindent
The first step is the weighted optimistic-online-mirror-descent bound in
Lemma~\ref{lem:weighted-oomd}. Applied to the learner at layer \(h\), and then
summed over histories, it leads to a bound of the form
\[
    R_{i,h}(\mu^i)
    \leq
    \frac{B_i(1+\mathrm{TV}_{i,h}^{\mu})}{\eta_h}
    +
    \frac{\eta_h\rho_i}{2}
    \sum_{\tau\in\mathcal H_h}
    \sum_{t=1}^T
    d_{t,h}^{i,\mu}(\tau)
    \left\|
        Q_{t,h}^i(s(\tau),\cdot)
        -
        Q_{t-1,h}^i(s(\tau),\cdot)
    \right\|_\infty^2 ,
\]
where
\( 
    \mathrm{TV}_{i,h}^{\mu}
    :=
    \sum_{t=1}^{T-1}
    \left\|
        d_{t+1,h}^{i,\mu}
        -
        d_{t,h}^{i,\mu}
    \right\|_1 .
\) 
The first term is controlled by the bounded Bregman diameter of the
smoothed-entropy regularizer.  The second term is where optimism enters: the
cost of learning depends on the variation of consecutive \(Q\)-functions rather
than on their magnitudes.

\vspace{6pt}
\noindent
The second step controls the variation of the hidden weights. Regarding the smooth-entropy regularizer, define
    \( 
        \rho_i := 1 + |A_i|\lambda_i,
        \;
        \bar\rho := \max_{i\in[m]} \rho_i,
    \) 
    and
    \( 
        B_i
        :=
        \rho_i
        \log\frac{1+\lambda_i}{\lambda_i}.
    \) 
By the
proximal Lipschitz property of the smoothed-entropy OOMD update given in Lemma \ref{lem:prox-lip}, Lemma~\ref{lem:movement}
shows that the one-step movement of player \(i\)'s layer-\(h\) policy satisfies
\[
    \max_{s\in S_h}
    \left\|
        \pi_{t+1,h}^i(\cdot\mid s)
        -
        \pi_{t,h}^i(\cdot\mid s)
    \right\|_1
    \le
    3\rho_iH\eta_h .
\]
Then, Lemma~\ref{lem:history-tv} uses the layered structure of the game to show
that the layer-\(h\) history weights depend only on policy movement in earlier
layers:
\( 
    \mathrm{TV}_{i,h}^{\mu}
    \le
    3m\rhobar H T
    \sum_{\ell<h}\eta_\ell .
\) 
Since
earlier layers use smaller step sizes, their movement is slow enough to keep
the later-layer visitation weights stable.

\vspace{6pt}
\noindent
The third step of the proof controls the \(Q\)-variation term.  Lemma~\ref{lem:q-perturb}
shows that a small change in the product Markov profile leads to a small change
in \(Q\)-functions.  Combining this perturbation estimate with the
policy-movement bound leads to Lemma~\ref{lem:qvar}, namely
\[
    \left\|
        Q_{t,h}^i(s,\cdot)
        -
        Q_{t-1,h}^i(s,\cdot)
    \right\|_\infty^2
    \le
    36m^2\rhobar^2H^6\eta_H^2
    \qquad (t\ge 2),
\]
with the \(t=1\) round bounded separately by \(H^2\).

\vspace{6pt}
\noindent
Substituting these estimates into the weighted OOMD bound yields, up to
absolute constants,
\[
    R_{i,h}(\mu^i)
    \leq
    \frac{B_i}{\eta_h}
    +
    \frac{3m\bar\rho HB_i}{\eta_h}
    \left(
        T\sum_{\ell<h}\eta_\ell
    \right)
    +
    \frac{1}{2}\eta_h\rho_iH^2
    +
    18 \eta_h\rho_i m^2\rhobar^2H^6\eta_H^2T .
\]
The learning-rate schedule
\( 
    \eta_h=\eta_0T^{-\alpha_h},
    \;
    \alpha_h=\frac{3(H-h)+1}{3H+1},
\) 
is chosen precisely to balance these terms after summing over layers. Therefore the total regret against any \(\mu^i\in\Pi_i^{\mathrm{gen}}\) is
bounded by
\[
    \left(
        \frac{H B_i}{\eta_0}
        +
        3mB_i\rhobar H^3
        +
        \frac12\eta_0\rho_iH^3
        +
        18\eta_0^3\rho_i m^2\rhobar^2H^7
    \right)T^\beta ,
\]
where \(\beta=\frac{3H-2}{3H+1}\). Finally, taking the supremum over \(\mu^i\) and choosing \(\lambda_i=1/|A_i|\) gives
\( 
    \Reg^{\mathrm{gen},i}_T \le C_i^{\mathrm{SE}}T^\beta.
\) 
Since \(\beta-1=-3/(3H+1)\), choosing
\( 
    T_\varepsilon
    :=
    \left\lceil
    \left(
        \frac{C_*^{\mathrm{SE}}}{\varepsilon}
    \right)^{(3H+1)/3}
    \right\rceil,
    \;
    C_*^{\mathrm{SE}}:=\max_{i\in[m]} C_i^{\mathrm{SE}},
\) 
ensures that, for every player \(i\in[m]\),
\( 
    \Reg^{\mathrm{gen},i}_{T_\varepsilon} 
    \le \varepsilon T_\varepsilon .
\) 
The CCE guarantee then follows from the standard no-regret-to-CCE equivalence applied to the
empirical distribution over the generated policy profiles.

\end{proof}

\subsection{Episodic Markov games under partial feedback}
\label{subsec:episodic-partial-feedback}

In this subsection, we extend the results of Subsection \ref{subsec-full} to the partial-feedback setting. The primary distinction between the full-feedback and partial-feedback settings is the restrictions on the accessibility to \(Q\)-functions. We tackle this difficulty by estimating the \(Q\)-functions as follows. First, players will maintain a small action-exploration floor throughout the play. This action-exploration floor is maintained by taking mirror descent steps over the simplex
\[
    \Delta_i^\zeta
    :=
    \left\{
        x\in\Delta(A_i):
        x(a_i)\ge \frac{\zeta}{|A_i|}
        \text{ for every }a_i\in A_i
    \right\}
\]
instead of the full simplex \(\Delta(A_i)\), where \(\zeta\in(0,1]\). Second, unlike the full-information case, the players update policies in \emph{blocks}, which is inspired by \citet{ErezLancewickiShermanKorenMansour2023}. This update proceeds as follows. During a block, the
same product Markov profile is repeated for \(B\) independent episodes, and the
realized trajectories are used to estimate the \(Q\)-functions for that
block by visiting each state-action \((s,a)\) pair sufficiently many times. The policy is then updated using the estimated \(Q\)-functions.

\vspace{6pt}
\noindent
To ensure that every player state is reached often enough throughout the play, we impose the following reachability condition in a similar way to \citet{ErezLancewickiShermanKorenMansour2023, WeiEtAl2021}.

\begin{definition}[\((\kappa,\zeta)\)-reachability]
\label{def:kappa-zeta-reachability}
A finite-horizon Markov game is said to be \((\kappa,\zeta)\)-reachable if, for every Markov policy profile \(\bm\pi\) satisfying
\( 
\pi_h^i(a_i\mid s)
\ge
\frac{\zeta}{|A_i|}
\) 
for all \(i\in[m]\), \(h\in[H]\), \(s\in S_h\), and \(a_i\in A_i\), the induced state-occupancy probabilities satisfy
\( 
\Pr_{\bm\pi}(s_h=s)
\ge
\kappa
\)
for all \(h\in[H]\) and \(s\in S_h\).
\end{definition}

\vspace{6pt}
\noindent
At block \(k\), all players use the same product Markov profile \(\pi_k\) for
\(B\) independent episodes.  For episode \(r\in[B]\) in block \(k\), we write the
realized trajectory as
\( 
    (s_{1}^{k,r},\bm a_{1}^{k,r},s_{2}^{k,r},\bm a_{2}^{k,r},
    \ldots,s_{H}^{k,r},\bm a_{H}^{k,r},s_{H+1}^{k,r}).
\) 
The realized tail cost of player \(i\) at layer \(h\) is denoted by
\(
    Y_{h}^{i,k,r}
    :=
    \sum_{\tau=h}^{H}
        \ell_\tau^i(s_{\tau}^{k,r},\bm a_{\tau}^{k,r}).
\) 
By \(N_{k,i}\) we denote the number of times that player \(i\) visits \((s,a)\in S_h \in A_i\) during the play, which can be formally written as
\begin{equation}\label{eq:quasocount}
    N_{k,i}(s,a_i)
    :=
    \sum_{r=1}^{B}
    \mathbf 1\{s_h^{k,r}=s,\ a_{h}^{i,k,r}=a_i\}.
\end{equation}
Then, we define the empirical estimate of the \(Q\)-function of player \(i\) at block \(k\) as
\[
    \widehat Q_{k,h}^i(s,a_i)
    :=
    \begin{cases}
    \displaystyle
    \frac{1}{N_{k,i}(s,a_i)}
    \sum_{r=1}^{B}
    \mathbf 1\{s_h^{k,r}=s,\ a_{h}^{i,k,r}=a_i\}
    Y_h^{i,k,r},
    & N_{k,i}(s,a_i)>0,\\[2ex]
    0, & N_{k,i}(s,a_i)=0.
    \end{cases}
\]
\noindent 
Our Algorithm in the partial-feedback setting, compared to Algorithm~\ref{alg:layered-omd}, differs in only two aspects: the use of a blocked estimation procedure and the replacement of exact values with their corresponding estimates; see Algorithm \ref{alg:blocked-bandit-oomd}.

\begin{algorithm}[H]
\caption{Blocked Layered OOMD}
\label{alg:blocked-bandit-oomd}
\begin{algorithmic}[1]
\State \textbf{Input:} finite-horizon game \(G\), horizon \(H\), number of blocks
\(K\), block length \(B\), exploration floor \(\zeta\in(0,1]\), base step-size
\(\eta_0\in(0,1]\).
\State Define
\( 
    \alpha_h:=\frac{3(H-h)+1}{3H+1},
    \;
    \eta_h:=\eta_0K^{-\alpha_h}, \; \bar x^i(a_i):=\frac{1}{|A_i|},\lambda_i=1/|A_i|,
    \; a_i\in A_i, h\in[H].
\) 
\State For every \(i\in[m]\), \(h\in[H]\), and \(s\in S_h\), set
\( 
    \widetilde x_{0,h}^{i,s} = \bar x^i\in
    \argmin_{x\in\Delta_i^\zeta}\Psi_i(x),
    \;
    \widehat Q_0^i(s,\cdot)\equiv 0\in\mathbb R^{A_i}.
\) 
\For{\(k=1,\ldots,K\)}
    \For{each player \(i\in[m]\), layer \(h\in[H]\), and state \(s\in S_h\)}
        \State 
        \( 
            x_{k,h}^{i,s}
            =
            \argmin_{x\in\Delta_i^\zeta}
            \left\{
                \eta_h\bigl\langle \widehat Q_{k-1,h}^i(s,\cdot),x\bigr\rangle
                +D_i\bigl(x,\widetilde x_{k-1,h}^{i,s}\bigr)
            \right\}.
        \) 
        \State Set \(\pi_{k,h}^i(\cdot\mid s)=x_{k,h}^{i,s}(\cdot)\).
    \EndFor
    \State Execute \(B\) independent episodes using the fixed product profile
    \(\pi_k=(\pi_k^1,\ldots,\pi_k^m)\).
    \State Each player \(i\) forms the trajectory estimate
    \(\widehat Q_{k,h}^i(s,a_i)\) for all \(h\in[H]\), \(s\in S_h\), and
    \(a_i\in A_i\).
    \For{each player \(i\in[m]\), layer \(h\in[H]\), and state \(s\in S_h\)}
        \State 
        \( 
            \widetilde x_{k,h}^{i,s}
            =
            \argmin_{x\in\Delta_i^\zeta}
            \left\{
                \eta_h\bigl\langle \widehat Q_{k,h}^i(s,\cdot),x\bigr\rangle
                +D_i\bigl(x,\widetilde x_{k-1,h}^{i,s}\bigr)
            \right\}.
        \) 
    \EndFor
\EndFor
\end{algorithmic}
\end{algorithm}

\vspace{6pt}
\noindent
For the regret guarantee of Algorithm~\ref{alg:blocked-bandit-oomd}, we first introduce the following:
\[
    \widetilde C_*^{\mathrm{SE}}
    \le
    \max_{i\in[m]}
    \left\{
        \frac{2H\log(|A_i|+1)}{\eta_0}
        +12mH^3\log(|A_i|+1)
        +\eta_0H^3
        +288\eta_0^3m^2H^7
    \right\},
    \quad
    \widetilde C_*^{\mathrm{SE}}
    :=
    \max_{i\in[m]}\widetilde C_i^{\mathrm{SE}}.
\]
The constant above differs from the constant $C_*^{\mathrm{SE}}$ in the full-feedback case only in the last term, which is due to the \(Q\)-function estimations that we use. The following result is the corresponding regret guarantee of Algorithm~\ref{alg:blocked-bandit-oomd}.

\begin{theorem}
\label{thm:episodic-partial-feedback}
Suppose the underlying Markov game is \((\kappa,\zeta_{\epsilon})\)-reachable, where \(\zeta_\varepsilon:=\frac{\varepsilon}{8H^2}\). Fix \(\varepsilon,\delta\in(0,1]\). \vspace{2pt} Choose
\( 
    K_\varepsilon
    :=
    \left\lceil
        \left(
            \frac{4\widetilde C_*^{\mathrm{SE}}}{\varepsilon}
        \right)^{(3H+1)/3}
    \right\rceil, B_\varepsilon
    :=
    \left\lceil
        \frac{8A_{\max}}{\kappa\zeta_\varepsilon}
        \bigl(n_\varepsilon+u_\varepsilon\bigr)
    \right\rceil,
\) 
where \vspace{2pt}
\( 
    M_\varepsilon
    :=
    K_\varepsilon m A_{\max}S,
    \xi_\varepsilon
    :=
    \min\left\{
        \frac{\varepsilon}{8H},
        \sqrt{\frac{\varepsilon}{32\eta_0 H}}
    \right\}, u_\varepsilon:=\log\frac{4M_\varepsilon}{\delta},
    n_\varepsilon
    :=
    \left\lceil
        \frac{H^2}{2\xi_\varepsilon^2}u_\varepsilon
    \right\rceil.
\) 
Then, running  Algorithm~\ref{alg:blocked-bandit-oomd} for every player \(i\in[m]\), for \(
    N_\varepsilon
    =
    \widetilde{O}
    \left(
        \frac{A_{\max}H^2}{\kappa\varepsilon}
        \left(\frac{H^2}{\xi_\varepsilon^2}+1\right)
        \left(
            \frac{\widetilde C_*^{\mathrm{SE}}}{\varepsilon}
        \right)^{(3H+1)/3}
    \right)
\) episodes\vspace{2pt} results in
\( \Reg^{\mathrm{gen},i}_{N_\varepsilon}
    \le
    \varepsilon N_\varepsilon,
\) with probability at least \vspace{2pt}
\(1-\delta\).  Consequently, the empirical distribution
\( 
    \widehat\Pi_{N_\varepsilon}
\) 
is an \(\varepsilon\)-approximate CCE.  
\end{theorem}

\noindent
Thus, our regret guarantees for the partial-feedback implementation of Algorithm~\ref{alg:blocked-bandit-oomd} inherits the same exponential dependence on the horizon length as the one we had in the full-feedback finite-horizon case, together with an additional polynomial overhead arising from exploration, and value-function estimation parts.

\begin{proof}[Proof Outline.]
The proof follows the same stability argument as in the full-feedback case,
with two additional steps.  First, because the algorithm uses only
partial feedback, it relies on an estimation of the exact $Q$-function values. The exploration floor
and the \((\kappa,\zeta)\)-reachability ensure that, during each
block, every state-action pair \((s,a_i)\) is sampled often enough. Later, Lemma~\ref{lem:pf-uniform-q-estimation}
shows that, if the block length is chosen as in Theorem \ref{thm:episodic-partial-feedback}, i.e. \( B_\varepsilon
    :=
    \left\lceil
        \frac{8A_{\max}}{\kappa\zeta_\varepsilon}
        \bigl(n_\varepsilon+u_\varepsilon\bigr)
    \right\rceil\), then with probability at
least \(1-\delta\), the \(Q\)-functions can be estimated uniformly with the desired precision, i.e.,
\( 
    \max_{k,i,s,a_i}
    \left|
        \widehat Q_{k,h}^i(s,a_i)-Q^{i}_{k,h}(s,a_i)
    \right|
    \le
    \xi_\varepsilon .
\) 

\vspace{6pt}
\noindent
Second, the imposed action floor for exploration restricts the policies played by the algorithm to
\(\Delta_i^{\zeta_\varepsilon}\).  To compare policies within \(\Delta_i^{\zeta_\varepsilon}\) against an arbitrary
history-dependent policy \(\mu^i\), we smooth the comparator by mixing it with
the uniform distribution.  Lemma~\ref{lem:pf-comparator-smoothing} then shows that
this comparison changes the value by at most \(2\zeta H^2\).  Thus, the error bound obtained due to the exploration is controlled by choosing \(\zeta_\varepsilon=\varepsilon/(8H^2)\).

\vspace{6pt}
\noindent
On the \emph{good event} that all \(Q\)-estimates are uniformly accurate, the proof then
runs the full-feedback OOMD analysis with \(\widehat Q_{k,h}^i\) in place of
\(Q_{k,h}^i\).  Then, Proposition~\ref{prop:pf-accurate-estimates} gives the following regret bound:
\[
\begin{aligned}
    \operatorname{Reg}^{\mathrm{gen},i}_N \!\!= \!\!\sup_{\mu^i\in\Pi_i^{\mathrm{gen}}}
    \frac1K\sum_{k=1}^K
    \left[
        V^{i,\pi_k}(s_1)
        \!-\!
        V^{i,\mu^i\odot\pi_k^{-i}}(s_1)
    \right]
    \!\le
    \widetilde C_i^{\mathrm{SE}}K^{-3/(3H+1)}
    \!+\!2H\xi
    \!+\!2\zeta H^2
    \!+\!
    8\eta_0H\xi^2K^{-\alpha_H}.
\end{aligned}
\]
The four terms have clear roles: the first is the finite-horizon OOMD error,
the second is the error from replacing exact \(Q\)-functions by (uniform)
estimates, the third is the cost of smoothing the comparator, and the last one is
an additional variation term caused by using estimated \(Q\)-functions in the
optimistic update.  The choices of \(K_\varepsilon\), \(\xi_\varepsilon\), and
\(\zeta_\varepsilon\) make these four terms at most \(\varepsilon/4\) each. Combining
this deterministic regret bound with high probability estimation guarantee given in Lemma~\ref{lem:pf-uniform-q-estimation} proves
Theorem~\ref{thm:episodic-partial-feedback}.
\end{proof}

\section{Approximating Equilibrium Discounted Markov Games}
\label{sec:discounted-truncation}

In this section, we extend the finite-horizon non-discounted Markov game regret guarantees from
Section~\ref{sec:episodic-algorithm} to discounted Markov games in both finite horizon and infinite horizon cases. The argument for this extension is independent of the feedback model, and it only requires a
finite-horizon regret guarantee on the truncated discounted game. With this extension, we show that the corresponding regret can be reduced to a \emph{sub-exponential} bound in the discounted case. We first
state this general transfer principle in Theorem \ref{thm:discounted-transfer}, and then instantiate it for the
full-feedback and partial-feedback algorithms in Corollaries \ref{cor:disc-full} and \ref{cor:disc-bandit} .

\vspace{6pt}
\noindent
\subsection{A transfer principle}
In the infinite-horizon discounted game case, a non-Markov nonstationary policy consists of infinitely many decision rules, and hence cannot be represented by a finite-sized object. Thus, a finite-output algorithm must either restrict attention to a compactly represented policy class, such as stationary policies, or approximate the discounted game by a finite-horizon truncation. It is known that, for general-sum discounted Markov games, computing stationary equilibria is already \(\mathsf{PPAD}\)-hard \citep{daskalakis2023}. We therefore follow the second route and work with \emph{sparse} distributions over finite-horizon product-policy prefixes. In the infinite-horizon case, each such prefix is extended beyond the truncation point by a fixed reference continuation. The same formulation also covers finite-horizon discounted games: when \(\bar H<\infty\), we choose a truncation length \(L\leq \bar H\), extend the learned \(L\)-step policies only if \(L<\bar H\), and the continuation is vacuous when \(L=\bar H\).

\vspace{6pt}
\noindent
Ultimately, discounting is what makes this finite approximation idea possible.  Since the
contribution of future layers decays geometrically, controlling the first
\(L\) discounted layers is sufficient to control the original discounted
objective up to a tail error.  This observation allows the finite-horizon
regret guarantees from Section~\ref{sec:episodic-algorithm} to be transferred
directly to discounted Markov games.  Specifically, we run the finite-horizon
algorithm on the \(L\)-step discounted truncated game, obtain a regret guarantee on
the first \(L\) layers, and then account separately for the remaining
discounted tail to approximate the equilibrium of the discounted game with horizon \(\bar H\). We now formalize this transfer argument as follows.

\begin{theorem}
\label{thm:discounted-transfer}
Let
\( 
    \bar H\in\mathbb N_+\cup\{\infty\},
\) \( \gamma \in (0,1), \)
and fix \(L\le \bar H\). Suppose that a learning procedure is run on the game \(G_{\gamma,\bar H}^{[L]}\) for \(N\) episodes, and outputs an
\(L\)-step product Markov profile at each episode,
\( 
    \{\bm\pi_1^{[L]},\ldots,\bm\pi_N^{[L]}\},
\) 
where
\( 
    \bm\pi_t^{[L]}\in
    \Pi^{\mathrm{markov},L}
    \;
    \text{for every } t\in[N].
\)  
Assume that, for every player
\(i\in[m]\), the learning procedure satisfies
\[
    \sup_{\nu^i\in\Pi_i^{\mathrm{gen},L}}
    \frac1N
    \sum_{n=1}^N
    \left[
        J_{i,L}^\gamma(\bm\pi_n^{[L]})
        -
        J_{i,L}^\gamma(\nu^i\odot\bm\pi_n^{[L],-i})
    \right]
    \le
    \varepsilon_{\mathrm{alg}} .
\]
For each \(n\in[N]\), let
\( 
    \bar{\bm\pi}_n
    :=
    \mathrm{Ext}_L(\bm\pi_n^{[L]}).
\) 
Then, 
for every player \(i\in[m]\),
\[
    \sup_{\mu^i\in\Pi_i^{\mathrm{gen},\bar H}}
    \frac1N
    \sum_{n=1}^N
    \left[
        J_{i,\bar H}^\gamma(\bar{\bm\pi}_n)
        -
        J_{i,\bar H}^\gamma(\mu^i\odot\bar{\bm\pi}_n^{-i})
    \right]
    \le
    \varepsilon_{\mathrm{alg}}
    +
    \tau_{L,\bar H}(\gamma).
\]
Equivalently, \( \frac1N\sum_{n=1}^N \delta_{\bar{\bm\pi}_n} \) is an
\((\varepsilon_{\mathrm{alg}}+\tau_{L,\bar H}(\gamma))\)-approximate
discounted CCE over horizon \(\bar H\).
\end{theorem}

\begin{proof}[Proof Outline.]
The complete proof is deferred to Appendix~\ref{app:discounted-truncation}. To
compare an extended learned policy \(\bar{\bm\pi}_n=\mathrm{Ext}_L(\bm\pi_n^{[L]})\)
against an arbitrary \(\bar H\)-horizon deviation \(\mu^i\), the finite-horizon
algorithm only needs to compete with the first \(L\) decision rules of that
deviation,
\( 
    \mu^{i,[L]}.
\) 
The formal tail bound in Lemma~\ref{lem:discounted-tail-unified} leads to
\[
    0
    \le
    J_{i,\bar H}^\gamma(\bm\sigma)
    -
    J_{i,L}^\gamma(\bm\sigma^{[L]})
    \le
    \tau_{L,\bar H}(\gamma)
    \le
    \frac{\gamma^L}{1-\gamma}
\]
for every policy profile \(\sigma\).  Combining this estimate for the learned
policy with the non-negativity of costs leads to the one-step comparison
\[
\begin{aligned}
    J_{i,\bar H}^\gamma(\bar{\bm\pi}_n)
    -
    J_{i,\bar H}^\gamma(\mu^i\odot\bar{\bm\pi}_n^{-i})
    \le\;&
    J_{i,L}^\gamma(\bm{\pi}_n^{[L]})
    -
    J_{i,L}^\gamma(\mu^{i,[L]}\odot\bm{\pi}_n^{[L],-i})
    +
    \tau_{L,\bar H}(\gamma),
\end{aligned}
\]
which is Lemma~\ref{lem:discounted-single-policy-transfer}.  Averaging this
inequality over the policies produced by the finite-horizon algorithm and then
using the finite-horizon regret guarantee proves the theorem.\end{proof}

\vspace{6pt}
\noindent
The corollaries below instantiate Theorem~\ref{thm:discounted-transfer} with a similar choice of truncation length. For a target accuracy \(\varepsilon\), we choose
\( 
H_\varepsilon^\gamma
:=
\left\lceil
\frac{\log\left(\frac{2}{(1-\gamma)\varepsilon}\right)}
{\log(1/\gamma)}
\right\rceil ,
\)  
so that
\( 
\frac{\gamma^{H_\varepsilon^\gamma}}{1-\gamma}
\le
\frac{\varepsilon}{2}.
\) 
Thus, for a discounted game with horizon \(\bar H\in\mathbb N_+\cup{\infty}\), we set
\( 
L_\varepsilon := \min{\{\bar H,H_\varepsilon^\gamma\}}
\) 
to guarantee, 
\( 
\tau_{L_\varepsilon,\bar H}(\gamma)\le \frac{\varepsilon}{2}.
\) 
We, then run the corresponding finite-horizon algorithms, Algorithms \ref{alg:layered-omd} and \ref{alg:blocked-bandit-oomd}, on the \(L_\varepsilon\)-step truncated discounted game \(G_{\gamma,\bar H}^{[L_\varepsilon]}\) long enough to make its average regret at most \(\varepsilon/2\). The finite-horizon regret bounds in Theorems~\ref{thm:main} and~\ref{thm:episodic-partial-feedback} have exponents that depend on the horizon length, and thus replacing the original horizon by the effective discounted horizon \(L_\varepsilon\) is crucial. For every fixed \(\gamma<1\), \(H_\varepsilon^\gamma=O(\log(1/\varepsilon))\), and hence the resulting discounted-game bounds become quasi-polynomial in \(1/\varepsilon\), rather than exponential in the original horizon.

\subsection{Full-feedback and Partial-feedback consequences}
\label{subsec:discounted-full-feedback}

Let \(C_*^{\mathrm{SE}}(L)\) denote the full-feedback constant in
Theorem~\ref{thm:main} with the horizon parameter replaced by
\(L\), and set
\( 
    A_{\log}:=\max_{i\in[m]}\log(|A_i|+1).
\) Then, formally, we have the following guarantees. 

\begin{corollary}
\label{cor:disc-full}
Let
\( 
    \bar H\in\mathbb N_+\cup\{\infty\},
    \;
    \varepsilon\in(0,1],
\) 
and define
\( 
    H_\varepsilon^\gamma
    :=
    \left\lceil
        \frac{\log\!\left(\frac{2}{(1-\gamma)\varepsilon}\right)}
             {\log(1/\gamma)}
    \right\rceil .
\)  Choose
\( 
    T_\varepsilon
    :=
    \left\lceil
        \left(
            \frac{2C_*^{\mathrm{SE}}(L_\varepsilon)}{\varepsilon}
        \right)^{(3L_\varepsilon+1)/3}
    \right\rceil , 
\) 
\( 
    L_\varepsilon
    :=
    \min\{\bar H,H_\varepsilon^\gamma\}.
\)
Run Algorithm~\ref{alg:layered-omd} on the \(L_\varepsilon\)-step discounted
truncation
\( 
    G_{\gamma,\bar H}^{[L_\varepsilon]}
\) 
for \(T_\varepsilon\) episodes.  For each \(t\in[T_\varepsilon]\), let
\(
\bar{\bm\pi}_t
=
\mathrm{Ext}_{L_\varepsilon}(\bm\pi_t^{[L_\varepsilon]})
\).
Then, for every player \(i\in[m]\),
\[
    \sup_{\mu^i\in\Pi_i^{\mathrm{gen},\bar H}}
    \sum_{t=1}^{T_\varepsilon}
    \left[
        J_{i,\bar H}^\gamma(\bar{\bm\pi}_t)
        -
        J_{i,\bar H}^\gamma(\mu^i\odot\bar{\bm\pi}_t^{-i})
    \right]
    \le
    \varepsilon T_\varepsilon .
\]
Consequently,
\( 
    \frac1{T_\varepsilon}
    \sum_{t=1}^{T_\varepsilon}\delta_{\bar{\bm\pi}_t}
\) 
is an \(\varepsilon\)-approximate discounted CCE over horizon \(\bar H\). Moreover, for every fixed \(\gamma<1\),
\( 
    T_\varepsilon
     \le
        \left[
            \frac{
                \left(
                    \frac{A_{\log}}{\eta_0}
                    +
                    mA_{\log}
                    +
                    m^2
                    +
                    1
                \right)
                \min\{\bar H,\log(1/\varepsilon)\}^7
            }{
                \varepsilon
            }
        \right]^{
            O(\min\{\bar H,\log(1/\varepsilon)\})
        }.
\) 
\end{corollary}
\begin{proof} See Appendix \ref{subsec:corollaryproof-1}
\end{proof}

\vspace{6pt}
\noindent 
We now apply Theorem \ref{thm:discounted-transfer} to the partial-feedback
case.  By \(\widetilde C_*^{\mathrm{SE}}(L)\) denote the finite-horizon constant from Theorem~\ref{thm:episodic-partial-feedback} with horizon
\(L\). For a target finite-horizon accuracy \(\bar\varepsilon\in(0,1]\), horizon
\(L\), and confidence level \(\delta\in(0,1)\), choose
\( 
    K_{\bar\varepsilon,L}
    :=
    \left\lceil
        \left(
            \frac{4\widetilde C_*^{\mathrm{SE}}(L)}{\bar\varepsilon}
        \right)^{(3L+1)/3}
    \right\rceil,\; 
    \zeta_{\bar\varepsilon,L}
    :=
    \frac{\bar\varepsilon}{8L^2},
    \;
    B_{\bar\varepsilon,L,\delta}
    \!:=\!
    \left\lceil\!
        \frac{8A_{\max}}{\kappa\zeta_{\bar\varepsilon,L}}
    \left(n_{\bar\varepsilon,L,\delta}+u_{\bar\varepsilon,L,\delta}\right)
    \!\right\rceil,
\) 
where, 
\( 
    M_{\bar\varepsilon,L}
    \!:=\!
    K_{\bar\varepsilon,L}mA_{\max}S,
    \) \( 
    u_{\bar\varepsilon,L,\delta}
    \!:=\!
    \log\frac{4M_{\bar\varepsilon,L}}{\delta}, \; 
    n_{\bar\varepsilon,L,\delta}
    \!:=\!
    \left\lceil\!
        \frac{L^2}{2\xi_{\bar\varepsilon,L}^2}
        u_{\bar\varepsilon,L,\delta}
    \!\right\rceil\!,
    \;
    \xi_{\bar\varepsilon,L}
    :=
    \min\left\{
        \frac{\bar\varepsilon}{8L},
        \sqrt{\frac{\bar\varepsilon}{32\eta_0\bar L}}
    \right\}. 
\) 
Then, formally, we have the following guarantee. 

\begin{corollary}
\label{cor:disc-bandit}
Let
\( 
    \bar H\in\mathbb N_+\cup\{\infty\},
    \;
    \varepsilon,\delta\in(0,1],
\)
and define 
\( 
    H_\varepsilon^\gamma
    :=
    \left\lceil
        \frac{\log\!\left(\frac{2}{(1-\gamma)\varepsilon}\right)}
             {\log(1/\gamma)}
    \right\rceil .
\) 
Set
\( 
    L_\varepsilon
    :=
    \min\{\bar H,H_\varepsilon^\gamma\},
    \;
    \bar\varepsilon:=\frac{\varepsilon}{2}
\) 
. 
Assume that the truncated discounted game
\(
    G_{\gamma,\bar H}^{[L_\varepsilon]}
\) 
is \((\kappa,\zeta_{\bar\varepsilon,L_\varepsilon})\)-reachable.  Run
Algorithm~\ref{alg:blocked-bandit-oomd} on
\(G_{\gamma,\bar H}^{[L_\varepsilon]}\) with parameters 
\(K_{\bar\varepsilon,L_\varepsilon}\),
\(B_{\bar\varepsilon,L_\varepsilon,\delta}\), and \(\zeta_{\bar\varepsilon,L_\varepsilon}\) for \(N_{\varepsilon,\bar H,\delta}\) episodes. For each
\(
k\in[K_{\bar\varepsilon,L_\varepsilon}],
\)
let
\( 
    \bar\pi_k
    :=
    \mathrm{Ext}_{L_\varepsilon}
    \bigl(\pi_k^{[L_\varepsilon]}\bigr).
\) 
Then, with probability at least \(1-\delta\), for every player \(i\in[m]\),
\[
    \sup_{\mu^i\in\Pi_i^{\mathrm{gen},\bar H}}
    \frac1{N_{\varepsilon,\bar H,\delta}}
    \sum_{k=1}^{K_{\bar\varepsilon,L_\varepsilon}}
    \sum_{t=1}^{B_{\varepsilon,L_\varepsilon,\delta}}
    \left[
        J_{i,\bar H}^\gamma(\bar\pi_k)
        -
        J_{i,\bar H}^\gamma(\mu^i\odot\bar\pi_k^{-i})
    \right]
    \le
    \varepsilon .
\]
Consequently,
\( 
    \frac1{N_{\varepsilon,\bar H,\delta}}
    \sum_{t=1}^{N_{\varepsilon,\bar H,\delta}}
    \delta_{\bar\pi_t}
\) 
is an \(\varepsilon\)-approximate discounted CCE.
Moreover, for constants \( \mathcal C_{\mathrm{band}}
    :=
    \frac{2A_{\log}}{\eta_0}
    +12mA_{\log}
    +\eta_0
    +288m^2 \eta^3_0,\) \( 
    \mathcal L_{\bar H,\varepsilon}
    :=
    \min\{\bar H,\log(1/\varepsilon)\},
\)
and \(\gamma<1\), we have
\[
\begin{aligned}
    N_{\varepsilon,\bar H,\delta}
    \le
    \left[
        \frac{
            A_{\max}\mathcal L_{\bar H,\varepsilon}^{2}
        }{
            \kappa\varepsilon
        }\!
        \left(
            1
            +
            \frac{\mathcal L_{\bar H,\varepsilon}^{4}}{\varepsilon^2}
            +
            \frac{\eta_0\mathcal L_{\bar H,\varepsilon}^{3}}{\varepsilon}\!
        \right)\!\!
        \times\!\!
        \left(\!
            \mathcal L_{\bar H,\varepsilon}\!
            \log
            \frac{
                \mathcal C_{\mathrm{band}}
                \mathcal L_{\bar H,\varepsilon}^{7}
            }{
                \varepsilon
            }
            \!+\!
            \log
            \frac{\!
                mA_{\max}S\!
            }{
                \delta
            }
        \right)\!
    \right]\!\!
    \left[
        \frac{\!
            \mathcal C_{\mathrm{band}}
            \mathcal L_{\bar H,\varepsilon}^{7}\!
        }{
            \varepsilon
        }
    \right]^{
        O(\mathcal L_{\bar H,\varepsilon})
    }
\end{aligned}
\]
\end{corollary}

\begin{proof}
See Appendix \ref{subsec:corollaryproof-2}
\end{proof}

\section{Computational Lower Bound for Discounted Markov Games}
\label{sec:computational-lower-bound}

In this section, we provide computational lower bounds for the task of finding a CCE in discounted Markov games to complement the quasi-polynomial upper bounds we have found in the earlier sections. In particular, we show that under the ``\(\mathsf{ETH}\) for \(\mathsf{PPAD}\)'', one should not expect a polynomial-time
polynomial-support algorithm for the same discounted CCE objective in
general-sum discounted Markov games when players learn independently in decentralized settings. To show this, we make simple modifications on the main results of \citet{noah_hardness}, which achieves this result for finite-horizon undiscounted Markov games for a problem called \emph{sparse Markov CCE}. Accordingly, we define a search problem that subsumes their problem as a special case.

\begin{definition}
\label{def:discounted-sparse-markov-cce}
Fix a discount factor \(\gamma\in(0,1]\) for the finite horizon case and fix a discount factor \(\gamma\in(0,1)\) for the infinite horizon case . For an \(m\)-player
discounted Markov game \(G_{\gamma,\bar H}\) and parameters
\(T\in\mathbb{N}\) and \(\varepsilon>0\), which may depend on the
size of \(G_{\gamma,\bar H}\), the
\( 
    (T,\varepsilon)\text{-}
    \mathsf{DiscSparseMarkovCCE}^{\mathrm{gen}}_{\gamma}
\) 
problem is the problem of finding a sequence
\( 
    \bm{\sigma}^{(1)},\ldots,\bm{\sigma}^{(T)},
    \;
    \bm{\sigma}^{(t)}\in\Pi^{\mathrm{markov},\bar H}
    \;\text{for every }t\in[T],
\) 
such that the distributional policy
\( 
    \widehat{\bm{\sigma}}_{T}
    :=
    \frac{1}{T}
    \sum_{t=1}^{T}
    \delta_{\bm{\sigma}^{(t)}}
    \in
    \Delta\!\left(\Pi^{\mathrm{gen},\bar H}\right)
\) 
is an \(\varepsilon\)-approximate discounted CCE of \(G_{\gamma,\bar H}\). Equivalently, the sequence must satisfy
\( 
    \sup_{\mu^i\in\Pi_i^{\mathrm{gen},\bar H}}
    \sum_{t=1}^{T}
    \left[
        J_i^\gamma\!\left(\bm{\sigma}^{(t)}\right)
        -
        J_i^\gamma\!\left(
            \mu^i\odot\bm{\sigma}^{(t),-i}
        \right)
    \right]
    \leq
    \varepsilon T
\) 
for every player \(i\in[m]\). We call such a sequence a
\(T\)-sparse discounted CCE.
\end{definition}

\noindent
Definition~\ref{def:discounted-sparse-markov-cce} formulates the
equilibrium-computation objective studied in this section as a search
problem over finite sequences of product Markov policy profiles. The
following proposition says that the outputs constructed
from Algorithms~\ref{alg:layered-omd} and
\ref{alg:blocked-bandit-oomd} are valid solutions to this search
problem.

\begin{proposition}
\label{prop:algorithms-solve-discounted-sparse-cce}
Let
\( G_{\gamma,\bar H}\) be an \(m\)-player discounted
Markov game, and fix \(\varepsilon\in(0,1]\).

\begin{enumerate}
    \item Let \(L_\varepsilon\) and \(T_\varepsilon\) be chosen as in
    Corollary~\ref{cor:disc-full}. Run
    Algorithm~\ref{alg:layered-omd} on the truncated Markov game
    \( G_{\gamma,\bar H}^{[L_\varepsilon]}\), and define
    \( 
        \bar{\bm\pi}_t
        :=
        \operatorname{Ext}_{L_\varepsilon}
        \bigl(\bm\pi_t^{[L_\varepsilon]}\bigr),
        \; t\in[T_\varepsilon].
    \) 
    Then
    \( 
        \bigl(
            \bar{\bm\pi}_1,\ldots,
            \bar{\bm\pi}_{T_\varepsilon}
        \bigr)
    \) 
    is a solution to
    \( 
        (T_\varepsilon,\varepsilon)\text{-}
        \mathsf{DiscSparseMarkovCCE}_{\gamma}^{\mathrm{gen}}.
    \) 

    \item Suppose that the reachability condition of
    Corollary~\ref{cor:disc-bandit} holds, and choose the parameters
    \(L_\varepsilon\), \(K_\varepsilon\), and \(B_\varepsilon\) as in Corollary~\ref{cor:disc-bandit}. Run
    Algorithm~\ref{alg:blocked-bandit-oomd} on
    \( G_{\gamma,\bar H}^{[L_\varepsilon]}\), and define
    \( 
        \bar{\bm\pi}_k
        :=
        \operatorname{Ext}_{L_\varepsilon}
        \bigl(\bm\pi_k^{[L_\varepsilon]}\bigr),
        \; k\in[K_\varepsilon].
    \)
    Let \(N_\varepsilon:=K_\varepsilon B_\varepsilon\), and define the
    policy sequence by
    \(
        \bm\sigma^{(k-1)B_\varepsilon+r}
        :=
        \bar{\bm\pi}_k,
        \;
        k\in[K_\varepsilon],\; r\in[B_\varepsilon].
    \)
    Then, with probability \(1-\delta\),
    \( 
        \bigl(
            \bm\sigma^{1},\ldots,
            \bm\sigma^{N_\varepsilon}
        \bigr)
    \) 
    is a solution to
    \( 
        (N_\varepsilon,\varepsilon)\text{-}
        \mathsf{DiscSparseMarkovCCE}_{\gamma}^{\mathrm{gen}}.
    \) 
\end{enumerate}
\end{proposition}

\begin{proof}
See Appendix~\ref{subsec:proof-search-problem-output}.
\end{proof}

\noindent
Before proceeding, we emphasize that
\( 
    (T,\varepsilon)\text{-}
    \mathsf{DiscSparseMarkovCCE}^{\mathrm{gen}}_{\gamma}
\) 
is a relaxation of the radically uncoupled learning task addressed by
Algorithms~\ref{alg:layered-omd} and
\ref{alg:blocked-bandit-oomd} through discounted policy extensions. The search problem only asks for
a finite sequence of product Markov policy profiles whose uniform
empirical distribution is an \(\varepsilon\)-approximate discounted CCE.
It does not require that this sequence be generated online by
independent players; in particular, a centralized algorithm with full
access to the game description is also permitted as stated in
\citet{noah_hardness}. Therefore, any lower bound provided for this search problem, is also valid for our radically uncoupled learning problem. 

\vspace{6pt}
\noindent
Next, we state the arguments that lay the foundation for our computational lower bound result.

\begin{lemma}\label{lemma:quaso:cheese}
    Let \(M\) be an \(n\times n\) bimatrix game. Then, there exists a corresponding  \(2\)-player Markov game \(F_H(M)\) with a predesignated horizon-length \(H\) such that computation of the CCE of \(F_H(M)\) implies computation of NE of \(M\). Furthermore, this construction can be done in polynomial time.
\end{lemma}
\begin{proof}
    See Appendix \ref{subsect:quaso0}.
\end{proof}

\noindent For any given finite-horizon Markov game, by forcing the game to enter an absorbing state after the step \(h=H\), we obtain an infinite-horizon discounted extension of it.

\begin{lemma} \label{lemma:quaso:hyper}
    Let \(G_H\) be a Markov game with horizon length \(H\). Then there exists an infinite-horizon \(\gamma\)-discounted Markov game \(\widetilde G_{H,\gamma}\) such that if \(\tilde L\) is an \(\varepsilon\)-approximate discounted CCE of \(\widetilde G_{H,\gamma}\), then there exists a list \(\widetilde L^{[H]}\) that is an \((\varepsilon/\gamma^{H-1})\)-approximate CCE of \(G_H\). Furthermore, this construction can be done in polynomial time. We denote this correspondence in the case of \(F_H(M)\) by \(\widetilde F_{H,\gamma} (M)\).
\end{lemma}
\begin{proof}
    See Appendix \ref{subsect:quaso1}.
\end{proof}

\vspace{6pt}
\noindent 
We refer to Appendix~\ref{app:discounted-lower-bound} for the detailed encoding conventions. 
Before introducing our main result of this section, we first recall the complexity-theoretic assumption underlying our lower bound. Introduced by \citet{Papadimitriou1994}, \(\mathsf{PPAD}\) is a widely studied complexity class for equilibrium computation in game theory. For the sake of completeness, we provide a complete statement of the so-called Exponential Time Hypothesis for \(\mathsf{PPAD}\). We refer to \citet{Papadimitriou1994} and \citet{Rubinstein2016} for further details.

\begin{definition}
 [$\mathsf{EndOfALine}$; \citep{DaskalakisGoldbergPapadimitriou2009}] Given two binary circuits \(S,P,\) each with \(m\)-input bits and \(m\) output bits such that \(P(0^m)=0^m \not = S(s^m)\), find an input \(x \in \{0,1\}^m\) such that \(P(S(x))\not = x\) or \(S(P(x))\not = x \not = 0^m\).   
\end{definition}

\begin{assumption}[\(\mathsf{ETH}\) for \(\mathsf{PPAD}\);
\citet{Babichenko2016}]
\label{ass:eth-ppad}
Solving the \(\mathsf{PPAD}\)-complete problem
\(\mathsf{EndOfALine}\), on instances of size
\(\tilde n\), requires time
\( 
2^{\Omega(\tilde n)}.
\) 
\end{assumption}

\noindent 
Our hardness results rely on the following consequence of the ``\(\mathsf{ETH}\) for \(\mathsf{PPAD}\)''.

\begin{theorem}[\citet{Rubinstein2016}]
\label{thm:rubinstein-qpoly-hardness}
Assume \(\mathsf{ETH}\) for \(\mathsf{PPAD}\). Then, there exists a universal
constant \(\varepsilon_\star>0\) such that every algorithm that computes an
\(\varepsilon_\star\)-approximate Nash equilibrium of every two-player
\(n\times n\) bimatrix game requires time
\( 
n^{\log^{1-o(1)} n}.
\) 
\end{theorem}

\begin{lemma}\label{lemma:quaso}
    There exists a countable infinite subset \(\mathcal S \subset \mathbb N\) such that for all \(n \in \mathcal S\) there exists a bimatrix game \(\widehat M\) with the following properties: 
    \begin{enumerate}
        \item Both players have \(n\) pure actions, and assuming \(\mathsf{ETH}\) for \(\mathsf{PPAD}\), computing an
\(\varepsilon_\star\)-approximate Nash equilibrium of \(\widehat M\) requires time
\( 
n^{\log^{1-o(1)} n}.
\)  
        \item Let \(|\widehat M|\) be the size of the game \(\widehat M\) (see Appendix \ref{app:lower-bound-search-problems}). There exist absolute constants \(q_{\rm src},n_{\rm src}\ge 1\), such that
    \( 
    n
    \le
    |\widehat M|
    \le
    (n+2)^{q_{\rm src}}
    \)
    for all \(n\ge n_{\rm src}\).
    \item Let \(\gamma\) be a discount factor. There are absolute constants \(A_H,B_H,n_{\mathrm{size}}(\gamma),d_0\) such that whenever \(n \ge n_{\mathrm{size}}(\gamma)\) and \(A_H\log(n+2)\!\le H \le \!B_H\log(n+2)\), \(n \!\le |\tilde F_{H,\gamma}(\widehat M)| \le \!(n+2)^{d_0}\) holds.
    \end{enumerate}
\end{lemma}
\begin{proof}
    See Appendices \ref{app:lower-bound-search-problems} and  \ref{subsect:quaso1}.
\end{proof}

\vspace{6pt}
\noindent
The next result is the main result of this section, which provides a computational lower bound for \((N^C,N^{-K})\text{-}\mathsf{DiscSparseMarkovCCE}^{\mathrm{gen}}_\gamma\) for arbitrarily large \(N\), where \(N\) is the \emph{size} of the discounted Markov game instance the search problem is concerned with (see Appendix \ref{app:lower-bound-search-problems}). Here, \(C\) is the \emph{support exponent} and \(K\) is the \emph{accuracy exponent}, which are absolute constants. 

\begin{theorem}
\label{thm:discounted-ppad-eth-lower-main}

Fix a constant discount factor \(\gamma\in(0,1)\cap \mathbb Q\) and a support
exponent \(C>0\). We define \(A_H :=
        \frac{2(Cd_0+1)}{c_F\varepsilon_\star^2},
        \;
        B_H := A_H+2,\) \(a_\gamma = B_H \log(\frac{1}{\gamma})\). By \(\widehat M\) denote the bimatrix game in Lemma \ref{lemma:quaso}. Under Assumption~\ref{ass:eth-ppad}, there exists an accuracy exponent \(K=K(\gamma,C)\ge a_\gamma+1\) such that, for any constructed discounted Markov game \(\widetilde F_{H,\gamma}(\widehat M)\), \( n \in \mathcal S\), satisfying \( n \ge n_{\mathrm{src}},n_{\mathrm{size}}(\gamma),\) \( n^{\hspace{2pt}-K}( n+2)^{a_{\gamma}}\le \frac{\varepsilon_{\star}}{4}\), and such that there exists an even integer \(H\) with 
\( 
B_H\log( n+2)\ge H \ge A_H\log( n+2),
\)
there exists a discounted two-player general-sum Markov game such that no polynomial-time algorithm solves \(\bigl(|\widetilde F_{H,\gamma}(\widehat M)|^C,|\widetilde F_{H,\gamma}(\widehat M)|^{-K}\bigr)\text{-}\mathsf{DiscSparseMarkovCCE}^{\mathrm{gen}}_\gamma\).
\end{theorem}

\noindent
Before proceeding with the proof of Theorem \ref{thm:discounted-ppad-eth-lower-main}, some remarks are in order. Theorem \ref{thm:discounted-ppad-eth-lower-main} demonstrates that with our current understanding of computational complexity, it is \emph{unlikely} that there exists a polynomial time algorithm for discounted Markov games. In particular, a quasi-polynomial time algorithm (which we have constructed) is \emph{about} the best one can hope for.

\begin{proof}[Proof of Theorem \ref{thm:discounted-ppad-eth-lower-main}]
    Suppose not. Take \(K:= a_{\gamma}+1\) that violates the running hypothesis. Let \(\widehat M\) denote the bimatrix game whose infinite-horizon discounted extension under \(\gamma\) admits a polynomial-time algorithm \(\mathcal A\) that solves \(\bigl(|\widetilde F_{H,\gamma}(\widehat M)|^C,|\widetilde F_{H,\gamma}(\widehat M)|^{-K}\bigr)\text{-}\mathsf{DiscSparseMarkovCCE}^{\mathrm{gen}}_\gamma\) with the aforementioned properties. Running \(\mathcal A\) on \(\widetilde F_{H,\gamma}(\widehat M)\), we obtain a list \(\widetilde L=(\pmb \sigma_1,\cdots,\pmb \sigma_T)\) such that
    \[
        T\le \lceil |\widetilde F_{H,\gamma}(\widehat M)|^C\rceil,
        \qquad
        \max_{i\in\{1,2\}}
        \mathrm{Gap}_i^\gamma
        (\widetilde L;\widetilde F_{H,\gamma}(\widehat M))
        \le |\widetilde F_{H,\gamma}(\widehat M)|^{-K},
    \]
    whose restrictions to the first \(H\)-layers is denoted by \( 
        L_H := \widetilde L^{[H]}
        =
        (\bm \sigma_1^{[H]},\ldots,\bm \sigma_T^{[H]})
\)
that satisfies the properties in Lemma \ref{lemma:quaso:hyper}. In particular, it holds that
\[
\begin{aligned}
\max_{i\in\{1,2\}}
\mathrm{Gap}_i^H(L_H;F_H(\widehat M))
&=
\gamma ^{-H+1}
\max_{i\in\{1,2\}}
\mathrm{Gap}_i^\gamma
(\widetilde L;\widetilde F_{H,\gamma}(\widehat M)) \le
|\widetilde F_{H,\gamma}(\widehat M)|^{-K}\gamma^{-H+1}.
\end{aligned}
\]
Since \(|\widetilde F_{H,\gamma}(\widehat M)|\ge n\) and \(H\le B_H\log(n+2)\) by Lemma \ref{lemma:quaso},
\[
\begin{aligned}
|\widetilde F_{H,\gamma}(\widehat M)|^{-K}\gamma^{-(H-1)}
&\le
n^{-K}\gamma^{-H}                                      \\
&=
n^{-K}\exp\!\left(H\log(1/\gamma)\right)                \\
&\le
n^{-K}\exp\!\left(B_H\log(n+2)\log(1/\gamma)\right)     \\
&=
n^{-K}(n+2)^{a_\gamma}
\le
\frac{\varepsilon_\star}{4}.
\end{aligned}
\]
Therefore \(L_H\) is an \(\varepsilon_\star/4\)-approximate \(T\)-sparse CCE
of \(F_H(\widehat M)\), against non-Markov deviations.
Since \(T\le \lceil |\widetilde F_{H,\gamma}(\widehat M)|^C\rceil\), \(|\widetilde F_{H,\gamma}(\widehat M)| \le (n+2)^{d_0}\), and that \(H \ge A_h\log(n+2)\), a straightforward inequality implies that \(T < \exp(c_F\varepsilon_\star^2 H)\), meaning that Lemma \ref{lem:fgk-extraction-app} is applicable. In particular, through \(\mathcal A\), one obtains an \(\varepsilon_\star\)-approximate Nash equilibrium of \(\widehat M\) in polynomial time.

\vspace{6pt}
\noindent 
Note that \(TH \le (n+2)^{Cd_0+1}O(\log n)\), which can be written in polynomial time in \(n\). Furthermore, for any given mixed strategy in \(\widehat M\), its Nash gap can be checked in polynomial time. Since \(\mathcal A\) runs in polynomial time and \(|\widetilde F_{H,\gamma}(\widehat M)|\le (n+2)^{d_0}\), all output probabilities have polynomial bit length in \(n\). Thus enumerating the
candidate pairs and outputting one with gap at most \(\varepsilon_\star\) is feasible in 
polynomial time. Construction of \(F_H(\widehat M)\) and \(\widetilde F_{H,\gamma}(\widehat M)\) takes polynomial time by Lemmas \ref{lemma:quaso:hyper} and \ref{lemma:quaso:cheese}. In particular, we have obtained a contradiction to Theorem \ref{thm:rubinstein-qpoly-hardness}.
\end{proof}

\noindent
The following corollary extends our hardness result for the
finite-horizon discounted setting. 

\begin{corollary}
\label{cor:finite-discounted-ppad-eth-main}
Fix a constant discount factor \(\gamma\in(0,1)\) and a support
exponent \(C_{\mathrm{fh}}>0\).  Under
Assumption~\ref{ass:eth-ppad}, there exists a finite-horizon discounted Markov game and an accuracy exponent
\( 
    K_{\mathrm{fh}}=K_{\mathrm{fh}}(\gamma,C_{\mathrm{fh}})>0
\) 
such that no polynomial-time algorithm can solve the following problem: given a game
\(G_{\gamma,\bar H}\) of input length at most \(N\), output a list of at most
\(N^{C_{\mathrm{fh}}}\) product Markov profiles whose \(\bar H\)-horizon
discounted CCE gap, against all history-dependent deviations, is at
most \(N^{-K_{\mathrm{fh}}}\).
\end{corollary}

\begin{proof}
Due to equilibrium correspondence between \(F_H(\widehat M)\) and \(\widetilde F_{H,\gamma}(\widehat M)\), given in Lemma \ref{lemma:quaso:hyper},  Theorem \ref{thm:discounted-ppad-eth-lower-main} implies that no polynomial time algorithm is feasible for a class of Markov games \(F_{H}(\widehat M)\) under some restrictions on \(H\) and \(n\). In such a game \(F_H(\widehat M)\), suppose that player \(i\) has an objective function of the form 
\(
J_{i,H}=\mathbb E\left[ \sum_{h=1}^H \gamma^{H-1}\ell^i_h(x_h,a_h) \right].
\)
Note that \(J_{i,H} = \mathbb E\left[ \sum_{h=1}^H \gamma^{h-1}\gamma^{H-1} \frac{\ell^i_h(x_h,a_h)}{\gamma^{h-1}} \right],\) and thus rewriting the cost of player \(i\) instead as \(\{\ell^i_h\gamma^{H-1}/\gamma^{h-1}\}_{h=1}^H\), we obtain that there is no polynomial time algorithm for a discounted Markov game.
\end{proof}

\section{Numerical Results}
\label{sec:num_results}



In this section, we evaluate the independent learning dynamics under both
information structures considered in the paper.  First, we run
Algorithm~\ref{alg:layered-omd} with full-feedback, as in
Theorem~\ref{thm:main}.  Second, we run Algorithm~\ref{alg:blocked-bandit-oomd} using only independently
sampled partial feedbacks, as in Theorem~\ref{thm:episodic-partial-feedback}.  These two experiments use exactly the
same four finite-horizon Markov games.  Finally, we use a two-state discounted linear-quadratic zero-sum Markov game with a closed-form stationary solution as a benchmark for the discounted truncation approximation result stated in Corollary~\ref{cor:disc-full}.

\subsection{Finite-horizon episodic Markov games}

We consider a Markov battle-of-the-sexes chain, a routing/congestion game, a
three-player public-goods game, and a transition-trap game.  The transition and
cost tables are given in Appendix~\ref{app:numerical-details}.  Every state appearing in these games has occupancy probability at least \(0.05\) under every product Markov profile.  Hence, all four games are
\((0.05,\zeta)\)-reachable for every \(\zeta\in(0,1]\). Letting \(N\) denote the
number of executed episodes, given the episode-indexed sequence of product
Markov profiles \((\rho_n)_{n=1}^N\), we report
\( 
    \operatorname{Gap}_N
    :=
    \max_{i\in[m]}
    \sup_{\mu^i\in\Pi_i^{\mathrm{gen}}}
    \frac{1}{N}
    \sum_{n=1}^{N}
    \left[
        V^{i,\rho_n}(s_1)
        -
        V^{i,\mu^i\odot\rho_n^{-i}}(s_1)
    \right].
\) 
For Algorithm~\ref{alg:layered-omd}, one policy update is performed per
episode, so \(N=t\) and \(\rho_n=\pi_n\).  For Algorithm~\ref{alg:blocked-bandit-oomd}, the profile
\(\pi_k\) is held fixed for \(B\) episodes in block \(k\); hence
\(N=kB\) and
\(\rho_{(k-1),B+r}=\pi_k\) for \(r\in[B]\).  

\vspace{6pt}
\noindent
The full-feedback experiment uses \(T=1000\), \(\eta_0=0.5\), and evaluates
the gap every \(50\) episodes.  The partial-feedback experiment uses
\( 
    K=1000,
    \;
    B=500,
    \;
    \zeta=0.05,
    \;
    \eta_0=0.1,
\) 
and evaluates the gap every \(50\) blocks, or equivalently in every \(25{,}000\)
episodes.  We perform \(20\) independent runs of the trajectory-feedback
experiment.  Figure~\ref{fig:finite-horizon-feedback-comparison} reports their
mean gap, with the shaded region representing plus or minus one empirical
standard deviation. In all four games, the CCE gap decreases under both feedback models.
Thus, the two panels respectively support the finite-horizon conclusions of
Theorems~\ref{thm:main} and~\ref{thm:episodic-partial-feedback}.

\begin{figure}[t]
    \centering
    \begin{minipage}{0.48\linewidth}
        \centering
        
        \includegraphics[width=\linewidth]{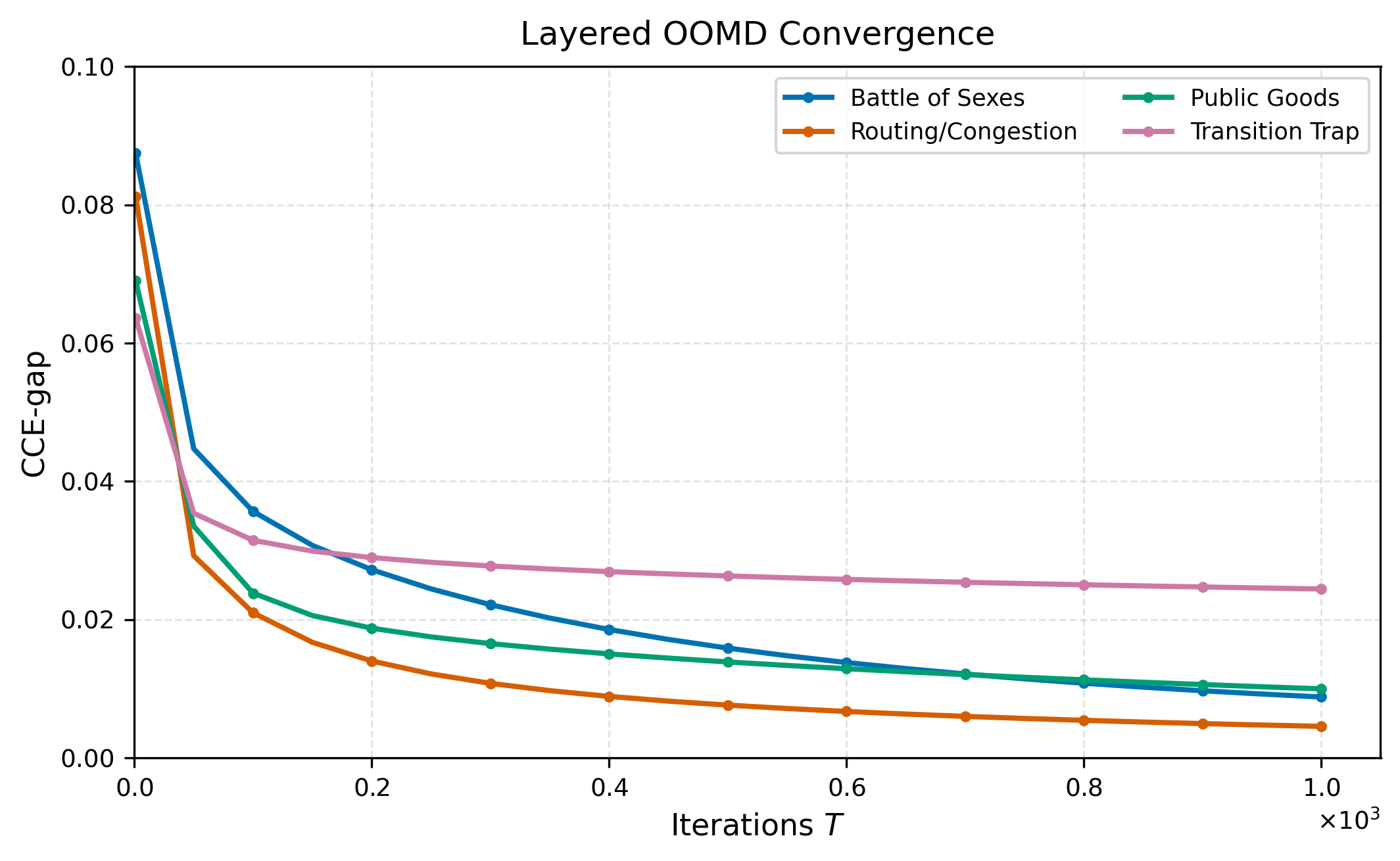}%
        \vspace{3pt}
        \textbf{(a)} Algorithm~\ref{alg:layered-omd}: full-feedback.
    \end{minipage}
    \hfill
    \begin{minipage}{0.48\linewidth}
        \centering        
            \includegraphics[width=\linewidth]{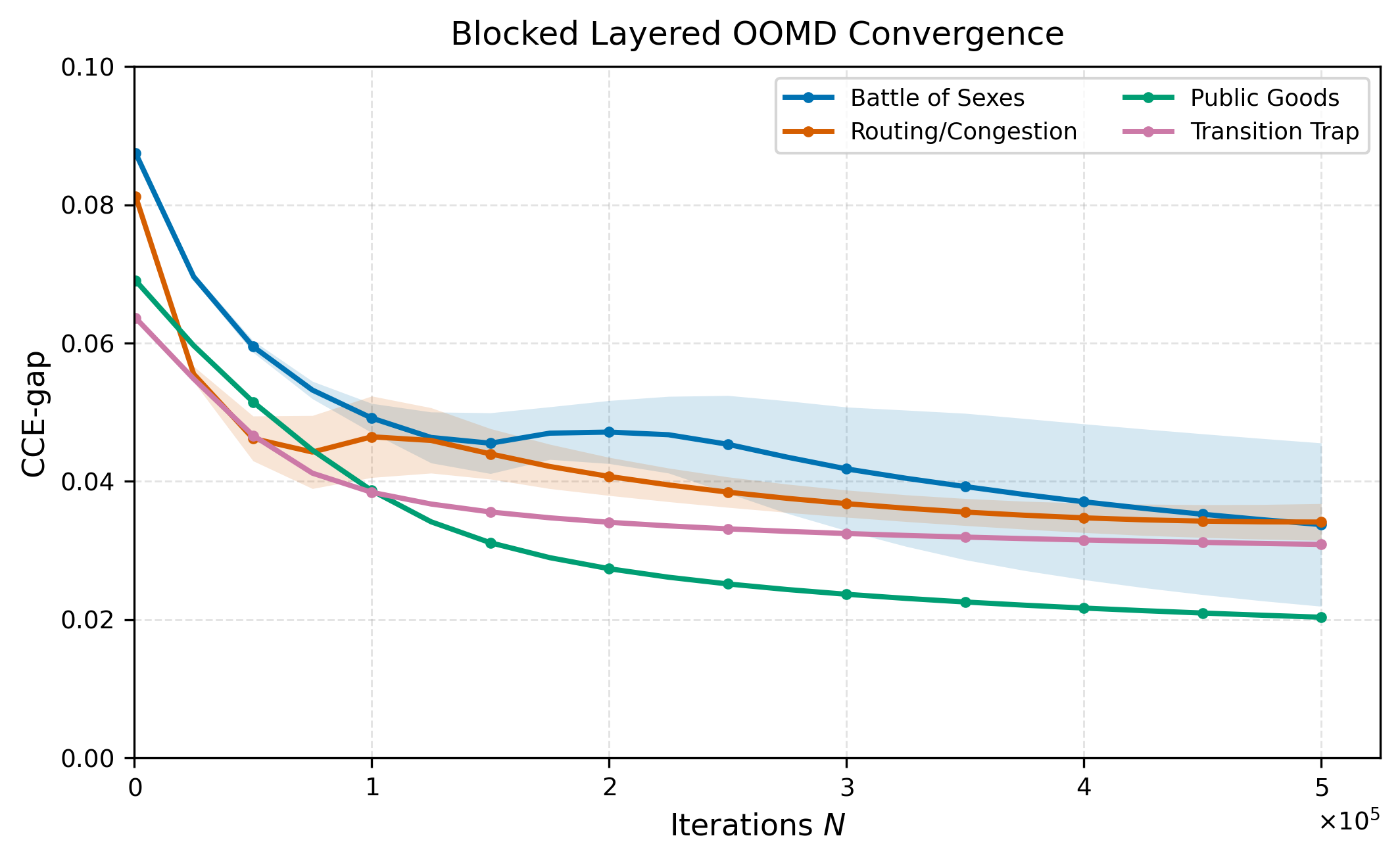}%
        \vspace{3pt}
        \textbf{(b)} Algorithm~\ref{alg:blocked-bandit-oomd}: partial feedback.
    \end{minipage}

    \caption{
    Finite-horizon behavior of the independent product-policy learning
    dynamics on the same four Markov games. Panel~(a) reports the
    deterministic feedback generated full-feedback.  Panel~(b)
    reports the mean over \(20\) independent partial feedback runs; shaded
    regions show one empirical standard deviation.
    }
    \label{fig:finite-horizon-feedback-comparison}
\end{figure}

\subsection{Discounted linear-quadratic benchmark}

The final experiment studies a two-player, two-state, zero-sum discounted
linear-quadratic Markov game with continuous scalar actions.  The benchmark is
chosen so that its infinite-horizon discounted saddle-point solution can be
computed in closed form.  In the simulation, we discretize the action interval
while ensuring that the grid contains the closed-form saddle-point actions.

\vspace{6pt}
\noindent 
Let \(J_0^\star(s_0)\) denote the normalized infinite-horizon discounted
saddle-point value obtained from the closed-form solution.  We report
\( 
    \operatorname{Err}_{T'}
    :=
    \left|
        \frac{1}{T'}\sum_{t=1}^{T'}J_0(\pi_t)
        -
        J_0^\star(s_0)
    \right|,
\) 
where \(J_0(\pi_t)\) is the value of the finite approximation used in the
simulation.  Figure~\ref{fig:discounted-lq-value-error} shows that this error
decreases rapidly. Thus, while the first two experiment evaluates CCE violations
directly for the episodic case, the LQ benchmark tests our approximation on a problem
where the infinite-horizon value is known. This experiment uses
\(
    T=1000,
    \ H=20,
    \ \gamma=0.90,
    \ \text{base grid size}=9,
    \ \eta_0=1.0
\).

\begin{figure}[t]
    \centering
    \IfFileExists{images/value_error_to_saddle.png}{%
        \includegraphics[width=0.64\linewidth]{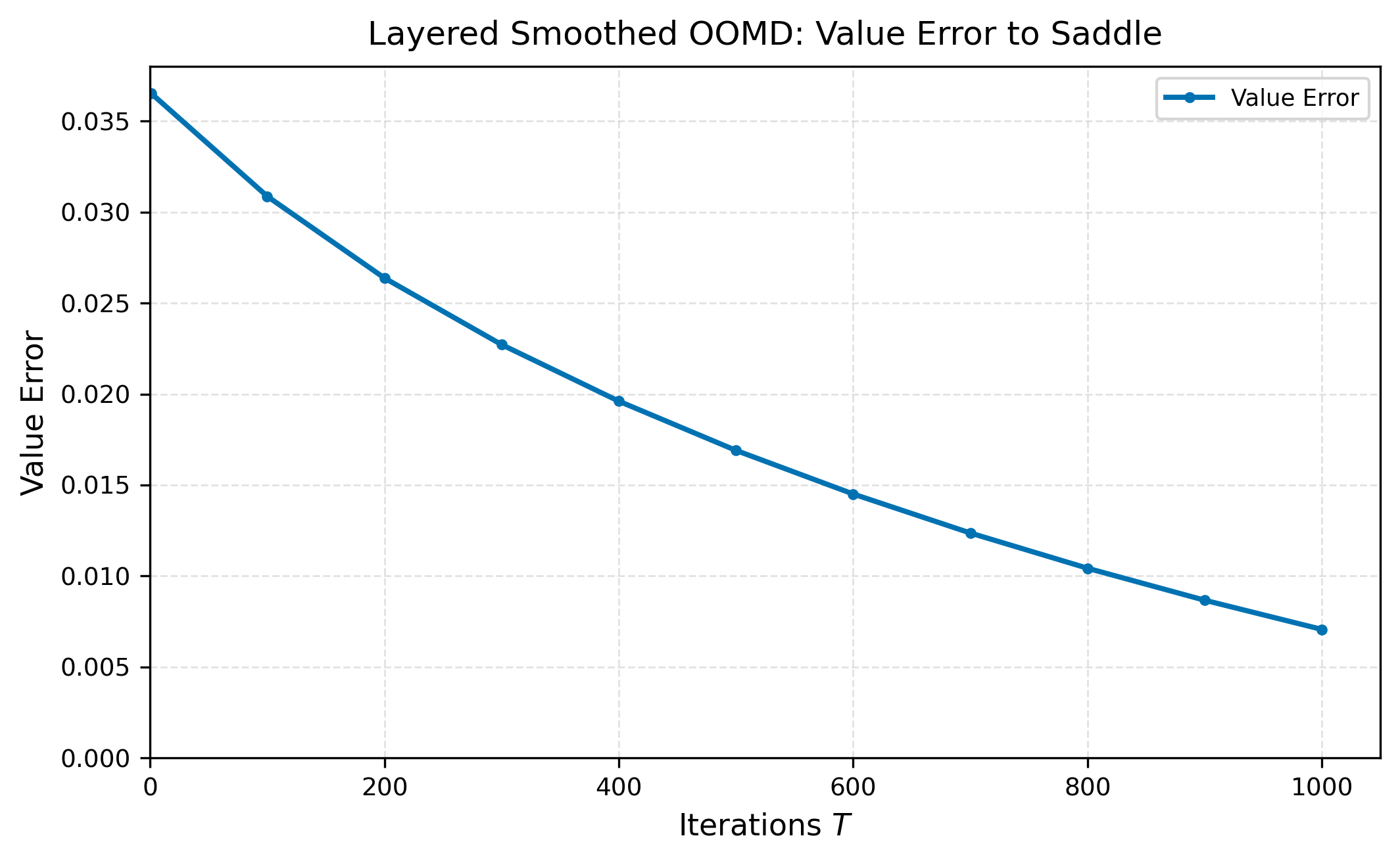}%
    }{%
        \fbox{\parbox{0.55\linewidth}{\centering
        Place \texttt{value\_error\_to\_saddle.png} here.}}%
    }
    \caption{
    Discounted linear-quadratic benchmark.  The plot reports the value error
    relative to the closed-form infinite-horizon saddle-point value.
    }
    \label{fig:discounted-lq-value-error}
\end{figure}

\section{Conclusion}

In this work, we have shown that, under the so-called ``\(\mathsf{ETH}\) for \(\mathsf{PPAD}\)'', there does not exist any polynomial time algorithm for computing CCE in discounted general-sum Markov games. We have complemented this hardness result with a quasi-polynomial-time algorithm for online self-play under a radically uncoupled information structure, without imposing any structural assumptions on the game. Our layered OOMD algorithm appears to be the first provably convergent algorithm in this setting with a sub-exponential convergence guarantee. Furthermore, we have shown that the layered OOMD algorithm extends to the partial-feedback setting while retaining similar convergence guarantees.

\vspace{6pt}
\noindent
Our work suggests several directions for future research. First, a gap remains between our quasi-polynomial upper bounds and the computational lower bounds established in this paper; determining the precise complexity of computing sparse CCEs in discounted Markov games remains an important open problem. Second, it would be interesting to determine whether polynomial-time radically uncoupled online-learning algorithms exist for general-sum Markov games when the considered policy classes for the players are restricted to Markov policies rather than general history-dependent policies. Finally, another direction is to characterize the boundary between radically uncoupled learning and more permissive information models. In particular, it would be interesting to investigate whether information models weaker than shared randomness also make tractable algorithms possible.   

\section{Acknowledgment}
Research was supported in part by the Army Research Office (ARO), under Grant Number W911NF-24-1-0085.

\newpage

\bibliographystyle{plainnat}

\newpage
\appendix 

\section{Related Work}

\noindent
The connection between no-regret learning and equilibrium computation is a cornerstone of learning in games. In repeated normal-form games, if every player uses an external no-regret algorithm with regret $\mathcal{O}(T^\alpha)$, then the empirical distribution of play converges to an $\mathcal{O}(T^{\alpha-1})$-approximate CCE; in two-player zero-sum games this further implies convergence to the saddle-point/Nash value, while swap or internal regret yields CE in general-sum games \citep{hannan57,HartMasColell2000,CesaBianchiLugosi2006,BlumMansour2007}. Classical online learning algorithms such as multiplicative weights, online mirror descent, and follow-the-regularized-leader provide the underlying no-regret guarantees in the fully adversarial setting, yielding the standard $\widetilde{\mathcal O}(T^{-1/2})$ convergence rate to equilibrium \citep{SchapireFreund95,NemirovskiYudin1983,AbernethyHazanRakhlin2008}.

\vspace{6pt}
\noindent
It is known that agents in the adversarial setting can be \emph{pessimistic} when their learning algorithms interact with one another in self-play \citet{farina2022}. Starting with the optimistic dynamics of \citet{DaskalakisDeckelbaumKim2011} for two-player zero-sum games and continuing through the RVU framework of \citet{SyrgkanisAgarwalLuoSchapire2015}, optimism and stability have emerged as central tools for leveraging this predictability. In normal-form games, optimistic variants of FTRL and OMD, together with their regularized extensions, attain convergence rates to coarse correlated and correlated equilibria that improve upon the baseline $\widetilde{\mathcal O}(T^{-1/2})$ rate, with recent results approaching $O(1/T)$ rates in general-sum normal-form games \citep{RakhlinSridharan2013,DaskalakisFishelsonGolowich2021,AnagnostidesEtAl2022,SoleymaniPiliourasFarina2025}. Our work is motivated by this optimistic self-play perspective. However, extending these ideas to Markov games introduces a challenge absent from normal-form games: deviations in the joint policy profile affect not only instantaneous losses but also future state occupancies.

\vspace{6pt}
\noindent
Early learning approaches in Markov games were mostly value-based. Specifically, the convergence of Nash-Q learning \citep{HuWellman2003} and Friend-or-Foe Q-learning \citep{Littman2001FriendFoe}, is established only under restrictions on the games encountered during learning. Related approaches based on CE computation further developed this equilibrium-computation perspective while maintaining similar structural requirements \citep{GreenwaldHall2003}. A separate line of research on decentralized \(Q\)-learning with convergence guarantees for weakly acyclic stochastic teams and games also emerged, which includes stochastic team problems, but again relies on strong structural assumptions on the games \citep{ArslanYuksel2017}. 

\vspace{6pt}
\noindent
Compared to the general-sum case, the two-player zero-sum structure in Markov games has been studied more thoroughly. There exist tractable algorithms with optimism and value-iteration variants, achieve finite sample-complexity guarantees for learning approximate Nash equilibria \citep{BaiJinYu2020,ZhangKakadeBasarYang2020}.  Policy-optimization methods have also been analyzed in the decentralized learning setting, including independent policy-gradient and optimistic gradient-descent/ascent approaches \citep{DaskalakisFosterGolowich2020,WeiEtAl2021}. Furthermore, \citet{Cai23, Chenetal2024} have developed uncoupled algorithms for two-player zero-sum Markov games with bandit feedback and finite-time last-iterate guarantees, relaxing the coordination and prior-knowledge constraints used in \citet{WeiEtAl2021}. From the value-learning angle, \citet{SayinZhangLeslieBasarOzdaglar2021} develops a radically uncoupled Q-learning algorithm for discounted zero-sum Markov games with asymptotic convergence guarantees.  

\vspace{6pt}
\noindent
Even for multi-player general-sum normal-form games, Nash equilibrium computation is known to be intractable \citet{DaskalakisGoldbergPapadimitriou2009}; and thus in Markov games provably efficient algorithms in the existing literature targets either CE or CCE. For instance, in sample-based MARL literature, V-learning \citep{Vlearning,SongMeiBai2022,MaoBasar2023}, SPoCMAR \citep{daskalakis2023} and related algorithms are utilized to compute CCE/CE policies in a decentralized setting and avoid the exponential dependence on the number of agents that appears in centralized settings \citet{BaiJinYu2020,Liu21central}. Near-optimal $\widetilde{\mathcal{O}}(1/T)$ convergence rates for both CCE and CE obtained in the full-information setting by using either value-based or stage-based policy-optimization schemes \citep{CaiLuoWeiZheng2024,MaoEtAl2024,YorulmazBasar2025}. These results require post-processing of the policies after the learning phase, which relies on the so-called \emph{shared randomness}. In contrast to our work, their learning phases do not provide any online regret guarantees, making their operating setting simpler than ours since they allow coordination mechanisms among the players.

\section{Convergence Analysis of Algorithm \ref{alg:layered-omd}}\label{app:proof-main}

\subsection{Properties of smoothed entropy}

\begin{lemma}
\label{lem:smoothed-entropy-properties}
Fix a player \(i\in[m]\) and a smoothing parameter \(\lambda_i>0\). Define the smoothed negative-entropy regularizer
\[
    \Psi_i(x)
    :=
    \sum_{a\in A_i}
    \bigl(x(a)+\lambda_i\bigr)
    \log\bigl(x(a)+\lambda_i\bigr),
    \qquad
    x\in\Delta(A_i).
\]
Then, 
\begin{equation}\label{eq:quaso}
    D_i(x,y)
    =
    \sum_{a\in A_i}
    \bigl(x(a)+\lambda_i\bigr)
    \log
    \frac{x(a)+\lambda_i}{y(a)+\lambda_i}.
\end{equation}
Moreover, defining
\( 
    \rho_i := 1+|A_i|\lambda_i,
    \;
    \bar\rho := \max_{i\in[m]}\rho_i,
\) 
and
\( 
    B_i
    :=
    \rho_i
    \log\frac{1+\lambda_i}{\lambda_i},
\) 
the following properties hold:
\begin{enumerate}
    \item \textbf{Boundedness.} For all \(x,y\in\Delta(A_i)\),
    \( 
        D_i(x,y)\le B_i.
    \) 

    \item \textbf{Strong convexity.} The function \(\Psi_i\) is \(1/\rho_i\)-strongly convex with respect to the \(\|\cdot\|_1\) on \(\Delta(A_i)\). Equivalently, for all \(x,y\in\Delta(A_i)\),
    \( 
        D_i(x,y)
        \ge
        \frac{1}{2\rho_i}
        \|x-y\|_1^2.
    \) 
\end{enumerate}
\end{lemma}

\begin{proof}
For each \(a\in A_i\), we have
\[
    \nabla_a \Psi_i(x)
    =
    \log\bigl(x(a)+\lambda_i\bigr)+1.
\]
Therefore,
\begin{align*}
    D_i(x,y)
    &=
    \sum_{a\in A_i}
    \bigl(x(a)+\lambda_i\bigr)
    \log\bigl(x(a)+\lambda_i\bigr)
    -
    \sum_{a\in A_i}
    \bigl(y(a)+\lambda_i\bigr)
    \log\bigl(y(a)+\lambda_i\bigr) \\
    &\qquad
    -
    \sum_{a\in A_i}
    \bigl(\log(y(a)+\lambda_i)+1\bigr)
    \bigl(x(a)-y(a)\bigr).
\end{align*}
Since \(x,y\in\Delta(A_i)\), we have
\( 
    \sum_{a\in A_i}\bigl(x(a)-y(a)\bigr)=0.
\) 
The linear terms cancel, and we obtain \eqref{eq:quaso}. Next, we prove that \(D_i\) is a bounded function. For every \(a\in A_i\),
\( 
    x(a)+\lambda_i \le 1+\lambda_i,
    \;
    y(a)+\lambda_i \ge \lambda_i.
\) 
Thus,
\[
    0\le\frac{x(a)+\lambda_i}{y(a)+\lambda_i}
    \le
    \frac{1+\lambda_i}{\lambda_i} \implies D_i(x,y)
    \le
    \sum_{a\in A_i}
    \bigl(x(a)+\lambda_i\bigr)
    \log\frac{1+\lambda_i}{\lambda_i}.
\]
Since
\begin{equation}\label{eq:quaso1}
    \sum_{a\in A_i}\bigl(x(a)+\lambda_i\bigr)
    =
    1+|A_i|\lambda_i
    =
    \rho_i,
\end{equation}
\(D_i\) is bounded above by $\rho_i
    \log\left((1+\lambda_i)/\lambda_i\right)
    =
    B_i$. Next, note that the Hessian of \(\Psi_i\) is diagonal and given by
\( 
    \nabla^2\Psi_i(x)
    =
    \operatorname{diag}
    \left(
        1/(x(a)+\lambda_i)
    \right)_{a\in A_i}.
\) 
Thus, for any \(z\in\mathbb{R}^{|A_i|}\), by Cauchy--Schwarz,
\[
    \|z\|_1^2
    =
    \left(\sum_{a\in A_i}|z(a)|\right)^2
    \le
    \left(\sum_{a\in A_i}(x(a)+\lambda_i)\right)
    \left(
        \sum_{a\in A_i}
        \frac{z(a)^2}{x(a)+\lambda_i}
    \right) = \rho_i\, z^\top \nabla^2\Psi_i(x)z,
\]
where the last inequality follows from \eqref{eq:quaso1}. Therefore, \(\Psi_i\) is \(1/\rho_i\)-strongly convex with respect to the \(\ell_1\)-norm, which implies
\( 
    D_i(x,y)
    \ge
    \frac{1}{2\rho_i}
    \|x-y\|_1^2.
\) 
This completes the proof.
\end{proof}

\subsection{Properties of OOMD with smoothed entropy}
\begin{lemma}\label{lem:prox-lip}
Fix player $i$, a base point $y\in\Delta(A_i)$, and $g,g'\in\R^{A_i}$.  Define
\[
  P_y(g):=
  \argmin_{x\in\Delta(A_i)}
  \left\{
    \eta\inner{g}{x}+D_i(x,y)
  \right\}.
\]
Then
\[
  \norm{P_y(g)-P_y(g')}_1
  \le
  \eta\rho_i\norm{g-g'}_\infty.
\]
In particular,
\[
  \norm{P_y(g)-y}_1
  \le
  \eta\rho_i\norm{g}_\infty.
\]
\end{lemma}

\begin{proof}
The first claim follows from Proposition B.4 in \citet{mertikopoulos2019optimistic}. 
The second claim follows by taking $g'=0$ and observing that $P_y(0)=y$, where $P_y$ is a well-defined function.
\end{proof}
\begin{lemma}
\label{lem:three-point}
Let \(\mathcal X\subseteq \mathbb R^{A_i}\) be a nonempty closed convex set, and let
\(\Psi:\mathcal X\to\mathbb R\) be a differentiable strictly convex function.
Further, let
\( 
D(x,y)
:=
\Psi(x)-\Psi(y)-\inner{\nabla\Psi(y)}{x-y},
\; x,y\in\mathcal X,
\) denote the associated Bregman divergence. Now, fix \(g\in\mathbb R^d\) and \(y\in\mathcal X\), and define
\(  
x^+
:=
\argmin_{x\in\mathcal X}
\left\{
    \eta\inner{g}{x}
    +
    D(x,y)
\right\}.
\)  
Then, for every $u\in\mathcal X$,
\[
\eta\inner{x^+-u}{g}
\le
D(u,y)
-
D(u,x^+)
-
D(x^+,y).
\]
\end{lemma}

\begin{proof}

The claim follows from Lemma 3.2. in \citet{chenTeboulle1993}. 

\end{proof}

\begin{lemma}\label{lem:weighted-oomd}
Let $x_t,\widetilde x_t\in\Delta(A_i)$ be generated by
\( 
  x_t
  =
  \argmin_{x\in\Delta(A_i)}
  \left\{\eta\inner{M_t}{x}+D(x,\widetilde x_{t-1})\right\},
\) 
\( 
  \widetilde x_t
  =
  \argmin_{x\in\Delta(A)}
  \left\{\eta\inner{\ell_t}{x}+D(x,\widetilde x_{t-1})\right\},
\) 
where $M_t=\ell_{t-1}$ and $\ell_0=0$.  Suppose that the regularizer chosen for the iterations is the smoothed entropy regularizer with parameters $\lambda_i,\rho_i=1+|A_i|\lambda_i$, and Bregman diameter at most $B$. Then, for any $u\in\Delta(A_i)$ and any weights $q_t\in[0,1]$,
\[
  \sum_{t=1}^T q_t\inner{x_t-u}{\ell_t}
  \le
  \frac{B(q_1+\TV(q))}{\eta}
  +
  \frac{\eta\rho_i}{2}
  \sum_{t=1}^T q_t\norm{\ell_t-\ell_{t-1}}_\infty^2,
\]
where we define 
\( 
  \TV(q):=\sum_{t=1}^{T-1}|q_{t+1}-q_t|.
\) 
\end{lemma}

\begin{proof}
Throughout the proof, we follow steps similar those in \citet{RakhlinSridharan2013, SyrgkanisAgarwalLuoSchapire2015, ErezLancewickiShermanKorenMansour2023}.

\vspace{6pt}
\noindent 
Since
\( 
\widetilde x_t
=
\argmin_{x\in\Delta(A_i)}
\left\{
    \eta\inner{\ell_t}{x}
    +
    D(x,\widetilde x_{t-1})
\right\},
\) 
from Lemma \ref{lem:three-point}, for every
$u\in\Delta(A_i)$,
\begin{align}\label{eq:quaso15}
\inner{\widetilde x_t-u}{\ell_t}
\le
\frac{
    D(u,\widetilde x_{t-1})
    -
    D(u,\widetilde x_t)
}{\eta}
-
\frac{
    D(\widetilde x_t,\widetilde x_{t-1})
}{\eta}.
\end{align}
Next, decompose
\( 
\inner{x_t-u}{\ell_t}
=
\inner{\widetilde x_t-u}{\ell_t}
+
\inner{x_t-\widetilde x_t}{\ell_t}.
\) 
Substituting \eqref{eq:quaso15} leads to
\begin{align}\label{eq:quaso16}
\inner{x_t-u}{\ell_t}
\le 
\frac{
    D(u,\widetilde x_{t-1})
    -
    D(u,\widetilde x_t)
}{\eta}
+
\inner{x_t-\widetilde x_t}{\ell_t}
-
\frac{
    D(\widetilde x_t,\widetilde x_{t-1})
}{\eta}.
\end{align}
Now split
\begin{align}\label{eq:quaso17}
\inner{x_t-\widetilde x_t}{\ell_t}
=
\inner{x_t-\widetilde x_t}{\ell_t-M_t}
+
\inner{x_t-\widetilde x_t}{M_t}.
\end{align}
Since
\(
x_t
=
\argmin_{x\in\Delta(A_i)}
\left\{
    \eta\inner{M_t}{x}
    +
    D(x,\widetilde x_{t-1})
\right\},
\)
another application of Lemma \ref{lem:three-point}, with comparator
$u=\widetilde x_t$, yields
\begin{align}\label{eq:quaso18}
\inner{x_t-\widetilde x_t}{M_t}
\le
\;
\frac{
    D(\widetilde x_t,\widetilde x_{t-1})
}{\eta}
-
\frac{
    D(\widetilde x_t,x_t)
}{\eta}
-
\frac{
    D(x_t,\widetilde x_{t-1})
}{\eta}.
\end{align}
Substituting \eqref{eq:quaso17} and \eqref{eq:quaso18} into \eqref{eq:quaso16} cancels out the term
\(
D(\widetilde x_t,\widetilde x_{t-1})
\),
yielding
\[
\begin{aligned}
\inner{x_t-u}{\ell_t}
\le 
\frac{
    D(u,\widetilde x_{t-1})
    -
    D(u,\widetilde x_t)
}{\eta}
+
\inner{x_t-\widetilde x_t}{\ell_t-M_t}
-
\frac{
    D(\widetilde x_t,x_t)
}{\eta}
-
\frac{
    D(x_t,\widetilde x_{t-1})
}{\eta}.
\end{aligned}
\]
Finally, since Bregman divergences are nonnegative, we may drop the final negative term to obtain
\[
\begin{aligned}
    \inner{x_t-u}{\ell_t}
    \le
    \frac{
        D(u,\widetilde x_{t-1})
        -
        D(u,\widetilde x_t)
    }{\eta}
    +
    \inner{x_t-\widetilde x_t}{\ell_t-M_t}
    -
    \frac{
        D(\widetilde x_t,x_t)
    }{\eta}.
\end{aligned}
\]
Due to Lemma \ref{lem:smoothed-entropy-properties}, $\Psi_i$ is $1/\rho_i$-strongly convex with respect to $\ell_1$,
\( 
  D(\widetilde x_t,x_t)
  \ge
  \frac{1}{2\rho_i}\norm{x_t-\widetilde x_t}_1^2.
\) 
Therefore, by Young's inequality,
\[
\begin{aligned}
  \inner{x_t-\widetilde x_t}{\ell_t-M_t}
  -\frac{D(\widetilde x_t,x_t)}{\eta}
  &\le
  \norm{x_t-\widetilde x_t}_1\norm{\ell_t-M_t}_\infty
  -\frac{1}{2\eta\rho_i}\norm{x_t-\widetilde x_t}_1^2 \le
  \frac{\eta\rho_i}{2}\norm{\ell_t-M_t}_\infty^2.
\end{aligned}
\]
Since $M_t=\ell_{t-1}$,
\[
  \inner{x_t-u}{\ell_t}
  \le
  \frac{D(u,\widetilde x_{t-1})-D(u,\widetilde x_t)}{\eta}
  +
  \frac{\eta\rho_i}{2}\norm{\ell_t-\ell_{t-1}}_\infty^2.
\]
Multiplying by $q_t$ and summing up yields
\[
  \sum_{t=1}^T q_t\inner{x_t-u}{\ell_t}
  \le
  \frac1\eta\sum_{t=1}^T q_t(D_{t-1}-D_t)
  +
  \frac{\eta\rho_i}{2}\sum_{t=1}^Tq_t\norm{\ell_t-\ell_{t-1}}_\infty^2,
\]
where $D_t:=D(u,\widetilde x_t)$.  Since $0\le D_t\le B$, leveraging summation by parts we get, 
\[
\begin{aligned}
  \sum_{t=1}^T q_t(D_{t-1}-D_t)
  &=q_1D_0-q_TD_T+
    \sum_{t=1}^{T-1}(q_{t+1}-q_t)D_t \le Bq_1+B\TV(q).
\end{aligned}
\]
This completes the proof.
\end{proof}

\subsection{First order path of the played policies}

\begin{lemma}\label{lem:movement}
For every player $i$, layer $h$, state $s\in S_h$, and episode $t$,
\[
  \norm{\pi^i_{t+1,h}(\cdot\mid s)-\pi^i_{t,h}(\cdot\mid s)}_1
  \le
  3\rho_iH\eta_h.
\]
\end{lemma}

\begin{proof}
Write $x_t=x^{i,s}_{t,h}$ and $\widetilde x_t=\widetilde x^{i,s}_{t,h}$.  By Lemma~\ref{lem:prox-lip}, applied with base $\widetilde x_t$ and loss $Q^i_t(s,\cdot)$,
\[
  \norm{x_{t+1}-\widetilde x_t}_1
  \le
  \eta_h\rho_i\norm{Q^i_{t,h}(s,\cdot)}_\infty
  \le
  \eta_h\rho_iH.
\]
Again applying Lemma~\ref{lem:prox-lip} to compare the reference point $\widetilde{x}_t$ and the played point $x_t$, both generated from the common base point $\widetilde{x}_{t-1}$ using the losses $Q_t^i(s,\cdot)$ and $Q_{t-1}^i(s,\cdot)$, respectively, we obtain:
\[
  \norm{\widetilde x_t-x_t}_1
  \le
  \eta_h\rho_i
  \norm{Q^i_{t,h}(s,\cdot)-Q^i_{t-1,h}(s,\cdot)}_\infty
  \le
  2\eta_h\rho_iH,
\]
because both $Q$-functions lie in $[0,H]^{A_i}$.  The triangle inequality gives
\( 
\norm{x_{t+1}-x_t}_1\le 3\rho_iH\eta_h.
\) 
Since $x_t=\pi^i_{t,h}(\cdot\mid s)$, the claim follows.
\end{proof}
\noindent 
For later use, we define
\( 
  \psi^i_{t,h}:=
  \max_{s\in S_h}
  \norm{\pi^i_{t+1,h}(\cdot\mid s)-\pi^i_{t,h}(\cdot\mid s)}_1.
\) 
Therefore, Lemma~\ref{lem:movement} gives
\(  \psi^i_{t,h}
  \le
  3\rhobar H\eta_h.
\) 

\subsection{History occupancy variation}

\begin{lemma}
\label{lem:history-tv}
Fix a player $i$, a layer \(h\), and a policy $\mu_i\in\Pi_i^{\mathrm{gen}}$.
Then, we have,
\[
  \TV^{i,\mu}_h
  :=
  \sum_{t=1}^{T-1}
  \norm{d^{i,\mu}_{t+1,h}-d^{i,\mu}_{t,h}}_1
  \le
  3m\rhobar HT
  \sum_{\ell=1}^{h-1}\eta_\ell .
\]
\end{lemma}

\begin{proof}
Given $\tau_\ell\in\Hist_\ell$, $a_i\in A_i$, $\bm a_{-i}\in A_{-i}$, and
$s'\in S_{\ell+1}$, define the next history, 
\( 
  \tau_{\ell+1}
  =
  \bigl(\tau_\ell, a_i, \bm a_{-i}, s'\bigr)
  \in \Hist_{\ell+1}.
\) For episode $t$, define the history transition kernel
\( 
  K^{i,\mu}_{t,\ell}
  :
  \Hist_\ell \to \Delta(\Hist_{\ell+1})
\) 
by
\[
  K^{i,\mu}_{t,\ell}(\tau_{\ell+1}\mid \tau_\ell)
  :=
  \sum_{\substack{a_i\in A_i,\;a_{-i}\in A_{-i},\;s'\in S_{\ell+1}:\\
  \tau_{\ell+1}=\bigl(\tau_\ell, a_i, \bm a_{-i}, s'\bigr)}}
  \mu^i(a_i\mid \tau_\ell)\,
  \bm{\pi}^{-i}_{t,\ell}(\bm a_{-i}\mid s(\tau_\ell))\,
  P_\ell(s'\mid s(\tau_\ell),a_i,\bm a_{-i}),
\]
First, we demonstrate that $K^{i,\mu}_{t,\ell}$ is a stochastic kernel. For every fixed
$\tau_\ell\in\Hist_\ell$,
\[
\begin{aligned}
  \sum_{\tau_{\ell+1}\in\Hist_{\ell+1}}
  K^{i,\mu}_{t,\ell}(\tau_{\ell+1}\mid\tau_\ell)
  &=
  \sum_{a_i\in A_i}
  \sum_{\bm a_{-i}\in A_{-i}}
  \sum_{s'\in S_{\ell+1}}
  \mu^i(a_i\mid\tau_\ell)
  \bm{\pi}^{-i}_{t,\ell}(\bm a_{-i}\mid s(\tau_\ell))
  P_\ell(s'\mid s(\tau_\ell),\bm a_i,a_{-i})
  \\
  &=
  \sum_{a_i\in A_i}
  \mu^i(a_i\mid\tau_\ell)
  \sum_{\bm a_{-i}\in A_{-i}}
  \bm{\pi}^{-i}_{t,\ell}(\bm a_{-i}\mid s(\tau_\ell))
  \sum_{s'\in S_{\ell+1}}
  P_\ell(s'\mid s(\tau_\ell),\bm a_i,a_{-i})
  \\
  &=1.
\end{aligned}
\]
Now, let $d^{i,\mu}_{t,h}\in\Delta(\Hist_h)$ denote the distribution over histories
at the beginning of layer $h$ induced by player $i$ using $\mu^i$ and the
opponents using $\bm{\pi}^{-i}_t$. Denoting by
\( 
    \delta_s \in \Delta(\mathcal S)
\) 
the probability distribution placing unit
mass on \(s\), we have
\( 
  d^{i,\mu}_{t,1}=\delta_{s_1},
\) 
For $h\ge2$,
\[ 
  d^{i,\mu}_{t,h}
  =
  \delta_{s_1}
  K^{i,\mu}_{t,1}
  K^{i,\mu}_{t,2}
  \cdots
  K^{i,\mu}_{t,h-1}.
\]
Here we view distributions as row vectors and kernels as row-stochastic
matrices. We next compare $d^{i,\mu}_{t+1,h}$ with $d^{i,\mu}_{t,h}$. For notational
simplicity, write
\( 
  K_\ell := K^{i,\mu}_{t,\ell},
  \;
  K'_\ell := K^{i,\mu}_{t+1,\ell}.
\) 
Then
\begin{align*}
  d^{i,\mu}_{t+1,h}-d^{i,\mu}_{t,h}
  &=
  \delta_{s_1}
  \left(
    K'_1K'_2\cdots K'_{h-1}
    -
    K_1K_2\cdots K_{h-1}
  \right)
  \\&=\delta_{s_1}\left( \sum_{\ell=1}^{h-1}
  K_1\cdots K_{\ell-1}
  (K'_\ell-K_\ell)
  K'_{\ell+1}\cdots K'_{h-1}\right).
\end{align*}
Stochastic kernels are $\ell_1$-contractions on signed measures, i.e., it holds that
\( 
  \norm{\nu K}_1\le \norm{\nu}_1
\) 
for every signed measure $\nu$ and every stochastic kernel $K$. Hence
\begin{align*}
  \norm{
    \delta_{s_1}
    K_1\cdots K_{\ell-1}
    (K'_\ell-K_\ell)
    K'_{\ell+1}\cdots K'_{h-1}
  }_1
  &\le
  \norm{
    \delta_{s_1}
    K_1\cdots K_{\ell-1}
    (K'_\ell-K_\ell)
  }_1
  \le
  \norm{K'_\ell-K_\ell}_{\infty,1},
\end{align*}
where
\( 
  \norm{K'_\ell-K_\ell}_{\infty,1}
  :=
  \max_{\tau_\ell\in\Hist_\ell}
  \sum_{\tau_{\ell+1}\in\Hist_{\ell+1}}
  \left|
    K'_\ell(\tau_{\ell+1}\mid\tau_\ell)
    -
    K_\ell(\tau_{\ell+1}\mid\tau_\ell)
  \right|.
\) Therefore,
\begin{align}\label{eq:quaso2}
  \norm{d^{i,\mu}_{t+1,h}-d^{i,\mu}_{t,h}}_1
  &\le
  \sum_{\ell=1}^{h-1}
  \norm{
    \delta_{s_1}
    K_1\cdots K_{\ell-1}
    (K'_\ell-K_\ell)
    K'_{\ell+1}\cdots K'_{h-1}
  }_1 \le
  \sum_{\ell=1}^{h-1}
  \norm{K^{i,\mu}_{t+1,\ell}-K^{i,\mu}_{t,\ell}}_{\infty,1}.
\end{align}
\noindent
It remains to control the one-layer kernel differences \(\norm{K^{i,\mu}_{t+1,\ell}-K^{i,\mu}_{t,\ell}}_{\infty,1}\). Fix
$\tau_\ell\in\Hist_\ell$ and write $s=s(\tau_\ell)$. Since $\mu^i$ and
$P_\ell$ are the same in episodes $t$ and $t+1$, only the opponents' product
distribution changes. Therefore,
\[
\begin{aligned}
\norm{
  K^{i,\mu}_{t+1,\ell}(\cdot\mid\tau_\ell)
  -
  K^{i,\mu}_{t,\ell}(\cdot\mid\tau_\ell)
}_1
 & =\!\!\!\!\!\!
\sum_{\tau_{\ell+1}\in\Hist_{\ell+1}}
\!\!\!\!\!\left|
  K^{i,\mu}_{t+1,\ell}(\tau_{\ell+1}\mid\tau_\ell)
  -
  K^{i,\mu}_{t,\ell}(\tau_{\ell+1}\mid\tau_\ell)
\right|
\\
& = 
\sum_{a_i}
\sum_{\bm a_{-i}}
\sum_{s'} 
\mu^i(a_i\mid\tau_\ell)
P_\ell(s'\mid s,a_i,\bm a_{-i})
\left|
  \bm{\pi}^{-i}_{t+1,\ell}(\bm a_{-i}\mid s)
  -
  \bm{\pi}^{-i}_{t,\ell}(\bm a_{-i}\mid s)
\right|
\\
& =
\sum_{\bm a_{-i}} \left|
  \bm{\pi}^{-i}_{t+1,\ell}(\bm a_{-i}\mid s)
  -
  \bm{\pi}^{-i}_{t,\ell}(\bm a_{-i}\mid s)
\right|
\sum_{a_i}
\mu^i(a_i\mid\tau_\ell)
\sum_{s'}
P_\ell(s'\mid s,a_i,\bm a_{-i})
\\
& =
\norm{
  \bm{\pi}^{-i}_{t+1,\ell}(\cdot\mid s)
  -
  \bm{\pi}^{-i}_{t,\ell}(\cdot\mid s)
}_1 \le
  \max_{s\in S_\ell}
  \norm{
    \bm{\pi}^{-i}_{t+1,\ell}(\cdot\mid s)
    -
    \bm{\pi}^{-i}_{t,\ell}(\cdot\mid s)
  }_1.
\end{aligned}
\]
For product distributions, we have 
\[
  \norm{
    \bm{\pi}^{-i}_{t+1,\ell}(\cdot\mid s)
    -
    \bm{\pi}^{-i}_{t,\ell}(\cdot\mid s)
  }_1
  \le
  \sum_{j\ne i}
  \norm{
    \pi^j_{t+1,\ell}(\cdot\mid s)
    -
    \pi^j_{t,\ell}(\cdot\mid s)
  }_1.
\]
Since, 
\( 
  \psi^j_{t,\ell}
  :=
  \max_{s\in S_\ell}
  \norm{
    \pi^j_{t+1,\ell}(\cdot\mid s)
    -
    \pi^j_{t,\ell}(\cdot\mid s)
  }_1,
\) 
it holds that
\( 
  \norm{K^{i,\mu}_{t+1,\ell}-K^{i,\mu}_{t,\ell}}_{\infty,1}
  \le
  \sum_{j\ne i}\psi^j_{t,\ell}.
\) 
For all \(h>1\), combining this with \eqref{eq:quaso2} gives
\[
  \norm{d^{i,\mu}_{t+1,h}-d^{i,\mu}_{t,h}}_1
  \le
  \sum_{\ell=1}^{h-1}\sum_{j\ne i}\psi^j_{t,\ell}.
\]
In the special case $h=1$, the distribution is always $d^{i,\mu}_{t,1}=\delta_{s_1}$, and hence
\( 
  \norm{d^{i,\mu}_{t+1,1}-d^{i,\mu}_{t,1}}_1=0.
\) Finally, summing over $t=1,\ldots,T-1$ gives
\begin{align*}
  \TV^{i,\mu}_h
  &=
  \sum_{t=1}^{T-1}
  \norm{d^{i,\mu}_{t+1,h}-d^{i,\mu}_{t,h}}_1
  \le
  \sum_{t=1}^{T-1}
  \sum_{\ell=1}^{h-1}
  \sum_{j\ne i}
  \psi^j_{t,\ell}
  \le
  \sum_{\ell=1}^{h-1}
  \sum_{j\ne i}
  3\rhobar HT\eta_\ell \le
  3m\rhobar HT
  \sum_{\ell=1}^{h-1}\eta_\ell.
\end{align*}
This completes the proof.
\end{proof}

\subsection{Variation bounds on the Q-values}

\begin{lemma}
\label{lem:q-perturb}
Fix a player $i$ and two Markov joint policies
\( 
  \bm{\pi}=(\pi^1,\ldots,\pi^m),
  \;
  \bm{\pi}'=(\pi^{1\prime},\ldots,\pi^{m\prime}).
\) Define
\( 
  \Delta
  :=
  \sum_{j=1}^m
  \norm{\pi^j-\pi^{j\prime}}_{\infty,1},
\) 
where
\( 
  \norm{\pi^j-\pi^{j\prime}}_{\infty,1}
  :=
  \max_{h,s}
  \norm{
    \pi^j_{h}(\cdot\mid s)
    -
    \pi^{j\prime}_{h}(\cdot\mid s)
  }_1.
\) 
Then, for every layer $h$ and state $s$,
\( 
  \norm{
    Q^{i,\bm{\pi}}_{h}(s,\cdot)
    -
    Q^{i,\bm{\pi}'}_{h}(s,\cdot)
  }_\infty
  \le
  2H^2\Delta.
\) 
\end{lemma}

\begin{proof}
Fix player $i$. Define
\( 
  \Delta_h
  :=
  \max_s
  \norm{
    Q^{i,\bm{\pi}}_{h}(s,\cdot)
    -
    Q^{i,\bm{\pi}'}_{h}(s,\cdot)
  }_\infty.
\) It can be seen that, at layer $H+1$, both value functions are zero, and hence
\( 
  \Delta_{H+1}=0.
\) 
Now, fix $h\le H$, $s\in S_h$, and $a_i\in A_i$. Then,
\[
\begin{aligned}
&
Q^{i,\bm{\pi}}_{h}(s,a_i)
-
Q^{i,\bm{\pi}'}_{h}(s,a_i)
=
\E_{\bm a_{-i}\sim\bm{\pi}^{-i}_{h}}
\Big[
  \ell^i_h
  +
  \E_{s'}V^{i,\bm{\pi}}_{h+1}(s')
\Big]
-
\E_{\bm a_{-i}\sim\bm{\pi}^{\prime -i}_{h}}
\Big[
  \ell^i_h
  +
  \E_{s'}V^{i,\bm{\pi}'}_{h+1}(s')
\Big].
\end{aligned}
\]
Add and subtract
\( 
\E_{\bm a_{-i}\sim\bm{\pi}^{-i}_{h}}
\Big[
  \ell^i_h
  +
  \E_{s'}V^{i,\bm{\pi}'}_{h+1}(s')
\Big].
\) 
Then,
\[
\begin{aligned}
&
\left|
Q^{i,\bm{\pi}}_{h}(s,a_i)
-
Q^{i,\bm{\pi}'}_{h}(s,a_i)
\right|
\le
\underbrace{
\E_{\bm a_{-i}\sim\bm{\pi}^{-i}_{h}}
\left|
\E_{s'}
\left[
  V^{i,\bm{\pi}}_{h+1}(s')
  -
  V^{i,\bm{\pi}'}_{h+1}(s')
\right]
\right|
}_{\text{future-value difference}}+
\underbrace{
\left|
\E_{\bm a_{-i}\sim\bm{\pi}^{-i}_{h}}f(\bm a_{-i})
-
\E_{\bm a_{-i}\sim\bm{\pi}^{\prime -i}_{h}}f(\bm a_{-i})
\right|
}_{\text{distribution shift}},
\end{aligned}
\]
where
\( 
  f(\bm a_{-i})
  :=
  \ell^i_h(s,a_i,\bm a_{-i})
  +
  \E_{s'}V^{i,\bm{\pi}'}_{h+1}(s').
\) 
For the future-value term, note that
\[
\begin{aligned}
\left|
V^{i,\bm{\pi}}_{h+1}(s')
-
V^{i,\bm{\pi}'}_{h+1}(s')
\right| 
&=
\left|
\left\langle
\pi^i_{h+1}(\cdot\mid s'),
Q^{i,\bm{\pi}}_{h+1}(s',\cdot)
\right\rangle
-
\left\langle
\pi^{i\prime}_{h+1}(\cdot\mid s'),
Q^{i,\bm{\pi}'}_{h+1}(s',\cdot)
\right\rangle
\right| \\
&\le
\left|
\left\langle
\pi^i_{h+1}(\cdot\mid s'),
Q^{i,\bm{\pi}}_{h+1}(s',\cdot)
-
Q^{i,\bm{\pi}'}_{h+1}(s',\cdot)
\right\rangle
\right| +
\left|
\left\langle
\pi^i_{h+1}(\cdot\mid s')
-
\pi^{i\prime}_{h+1}(\cdot\mid s'),
Q^{i,\bm{\pi}'}_{h+1}(s',\cdot)
\right\rangle
\right| \\
&\le
\Delta_{h+1}
+
\left\|
\pi^i_{h+1}(\cdot\mid s')
-
\pi^{i\prime}_{h+1}(\cdot\mid s')
\right\|_1
\left\|
Q^{i,\bm{\pi}'}_{h+1}(s',\cdot)
\right\|_\infty .
\end{aligned}
\]
Since per-layer losses lie in $[0,1]$, the remaining loss from layer $h+1$
to layer $H$ is at most $H-h$. Therefore
\( 
\left\|
Q^{i,\bm{\pi}'}_{h+1}(s',\cdot)
\right\|_\infty
\le H-h \le H.
\) 
Thus,
\[
\left|
V^{i,\bm{\pi}}_{h+1}(s')
-
V^{i,\bm{\pi}'}_{h+1}(s')
\right|
\le
\Delta_{h+1}
+
H
\left\|
\pi^i_{h+1}(\cdot\mid s')
-
\pi^{i\prime}_{h+1}(\cdot\mid s')
\right\|_1
\le
\Delta_{h+1}+H\Delta .
\]
Therefore, 
\( 
\E_{\bm a_{-i}\sim\bm{\pi}^{-i}_h}
\left|
\E_{s'}
[
V^{i,\bm{\pi}}_{h+1}(s')
-
V^{i,\bm{\pi}'}_{h+1}(s')
]
\right|
\le
\Delta_{h+1}+H\Delta .
\) For the distribution-shift term, since \( \|f\|_\infty\le H \), we have
\[
\begin{aligned}
&
\left|
\E_{\bm a_{-i}\sim\bm{\pi}^{-i}_{h}}f(\bm a_{-i})
-
\E_{\bm a_{-i}\sim\bm{\pi}^{\prime -i}_{h}}f(\bm a_{-i})
\right|
\le
H
\,
\norm{
  \bm{\pi}^{-i}_{h}(\cdot\mid s)
  -
  \bm{\pi}_{h}^{\prime -i}(\cdot\mid s)
}_1.
\end{aligned}
\]
Using the product-distribution inequality,
\[
\norm{
  \bm{\pi}^{-i}_{h}(\cdot\mid s)
  -
  \bm{\pi}_{h}^{\prime -i}(\cdot\mid s)
}_1
\le
\sum_{j\ne i}
\norm{
  \pi^j_{h}(\cdot\mid s)
  -
  \pi_{h}^{j\prime}(\cdot\mid s)
}_1
\le
\Delta.
\]
Hence
\( 
  \Delta_h
  \le
  \Delta_{h+1}
  +
  2H\Delta.
\) 
Since
\(
  \Delta_{H+1}=0,
\) 
backward induction yields
\( 
  \Delta_h
  \le
  2H(H-h+1)\Delta.
\) 
In particular, it holds that
\( 
  \Delta_h\le 2H^2\Delta.
\) 
\end{proof}

\begin{lemma}\label{lem:qvar}
For every player $i$, state $s$, layer $h$, and episode $t\ge2$,
\( 
  \norm{Q^i_{t,h}(s,\cdot)-Q^i_{t-1,h}(s,\cdot)}_\infty^2
  \le
  36m^2\rhobar^2H^6\eta_H^2,
\) 
and, with $Q^i_{0,h}(s,\cdot)=0$,
\( 
  \norm{Q^i_{1,h}(s,\cdot)-Q^i_{0,h}(s,\cdot)}_\infty^2\le H^2.
\) 
\end{lemma}

\begin{proof}
The Lemma \ref{lem:q-perturb} bound gives
\[
  \norm{Q^i_{t,h}(s,\cdot)-Q^i_{t-1,h}(s,\cdot)}_\infty
  \le
  2H^2\sum_{j=1}^m
  \norm{\pi^j_t-\pi^j_{t-1}}_{\infty,1}
  =
  2H^2\sum_{j=1}^m \max_{h,s\in S_h}
  \norm{\pi^j_{t,h}(\cdot\mid s)-\pi^j_{t-1,h}(\cdot\mid s)}_1.
\]
By Lemma~\ref{lem:movement} and the monotonicity $\eta_1\le\cdots\le\eta_H$, it holds that
\( 
  \norm{\pi^j_t-\pi^j_{t-1}}_{\infty,1}
  \le
  3\rhobar H\eta_H,
\) 
which directly implies the desired result.
\end{proof}

\subsection{Weighted value-difference decomposition}
\begin{lemma}
\label{lem:history-vdiff}
Fix a player $i$, an episode $t$, and a policy
$\mu_i\in\Pi_i^{\mathrm{gen}}$. Fix the opponents' profile
$\bm{\pi}^{-i}_t$. Then,
\[
\begin{aligned}
  V^{i,\pi}_{t,1}(s_1)
  -
  V^{i,\mu_i}_{t,1}(s_1)
  =
  \sum_{h=1}^H
  \sum_{\tau_h\in\Hist_h}
  d^{i,\mu}_{t,h}(\tau_h)
  \left\langle
    Q^i_{t,h}(s(\tau_h),\cdot),
    \pi^i_{h}(\cdot\mid s(\tau_h))
    -
    \mu_{i,h}(\cdot\mid\tau_h)
  \right\rangle .
\end{aligned}
\]
\end{lemma}

\begin{proof}
Let $W_h(\tau_h)$ denote the expected \emph{continuation} loss from history
$\tau_h\in\Hist_h$ when player $i$ follows the comparator $\mu_i$ from layer
$h$ onward and the opponents follow $\bm{\pi}^{-i}_t$. Thus
\( 
  W_{H+1}(\tau_{H+1})=0,
\) 
and for $\tau_h\in\Hist_h$ with current state $s=s(\tau_h)$,
\[
\begin{aligned}
  W_h(\tau_h)
  =
  \sum_{a_i\in A_i}
  \mu_{i,h}(a_i\mid\tau_h)
  \Big[
     \ell^i_{t,h}(s,a_i)
    +
    \sum_{s'\in S_{h+1}}
     P_{t,h}(s'\mid s,a_i)
    W_{h+1}(\tau_{h+1})
  \Big].
\end{aligned}
\]
where $\tau_{h+1} = (\tau_h, a_i, \bm a_{-i}, s')$ denotes the next history after player $i$ chooses
$a_i$, 
the opponents' policies are according to $\bm{\pi}^{-i}_t$,
and the next state is $s'$.  Then, for every history $\tau_h\in\Hist_h$ with $s=s(\tau_h)$, we have 
\[
\begin{aligned}
  V^{i,\pi}_{t,h}(s)\!-\!W_h(\tau_h)
  \!&=\!
  \left\langle
    Q^i_{t,h}(s,\cdot),
    \pi^i_{h}(\cdot\!\mid\! s)
  \right\rangle\!
  -\!\!\!
  \sum_{a_i\in A_i}\!\!
  \mu_{i,h}(a_i\!\mid\!\tau_h)
  \!\Big[
     \ell^i_{t,h}(s,a_i)\!
    +\!\!\!\!\!
    \sum_{s'\in S_{h+1}}\!\!\!\!\!
     P_{t,h}(s'\!\!\mid\! s,a_i)
    W_{h+1}(\tau_{h+1})
  \Big].
  \\&=
  \left\langle
    Q^i_{t,h}(s,\cdot),
    \pi^i_{h}(\cdot\mid s)
    -
    \mu_{i,h}(\cdot\mid\tau_h)
  \right\rangle
  \\
  &\quad+
  \sum_{a_i\in A_i}
  \mu_{i,h}(a_i\mid\tau_h)
  \sum_{s'\in S_{h+1}}
   P_{t,h}(s'\mid s,a_i)
  \left[
    V^{i,\pi^i}_{t,h+1}(s')
    -
    W_{h+1}(\tau_{h+1})
  \right].
    \\&=
  \left\langle
    Q^i_{t,h}(s(\tau_h),\cdot),
    \pi^i_{h}(\cdot\mid s(\tau_h))
    -
    \mu_{i,h}(\cdot\mid\tau_h)
  \right\rangle
  \\
  &\quad+
  \E\left[
    V^{i,\pi}_{t,h+1}(s(\tau_{h+1}))
    -
    W_{h+1}(\tau_{h+1})
    \,\middle|\,
    \tau_h,\mu_i,\bm{\pi}^{-i}_t
  \right].
\end{aligned}
\]
Now taking expectation with respect to
$\tau_h\sim d^{i,\mu}_{t,h}$, we get,
\[
\begin{aligned}
\E_{\tau_h\sim d^{i,\mu}_{t,h}}
\left[
  V^{i,\pi}_{t,h}(s(\tau_h))-W_h(\tau_h)
\right]
= &
\sum_{\tau_h\in\Hist_h}
d^{i,\mu}_{t,h}(\tau_h)
\left\langle
  Q^i_{t,h}(s(\tau_h),\cdot),
  \pi^i_{h}(\cdot\mid s(\tau_h))
  -
  \mu_{i,h}(\cdot\mid\tau_h)
\right\rangle
\\
& +
\E_{\tau_{h+1}\sim d^{i,\mu}_{t,h+1}}
\left[
  V^{i,\pi}_{t,h+1}(s(\tau_{h+1}))
  -
  W_{h+1}(\tau_{h+1})
\right].
\end{aligned}
\]
Summing this identity over $h=1,\ldots,H$ telescopes the terms.
At the initial layer, $d^{i,\mu}_{t,1}=\delta_{s_1}$, and thus
\[
  \E_{\tau_1\sim d^{i,\mu}_{t,1}}
  \left[
    V^{i,\pi}_{t,1}(s(\tau_1))-W_1(\tau_1)
  \right]
  =
  V^{i,\pi}_{t,1}(s_1)
  -
  V^{i,\mu_i}_{t,1}(s_1).
\]
At layer $H+1$, both continuation values are zero, i.e. 
\( 
  V^{i,\pi}_{t,H+1} \equiv 0,
  \;
  W_{H+1} \equiv 0.
\) 
Therefore,
\[
\begin{aligned}
  V^{i,\pi}_{t,1}(s_1)
  -
  V^{i,\mu^i}_{t,1}(s_1)
  =
  \sum_{h=1}^H
  \sum_{\tau_h\in\Hist_h}
  d^{i,\mu}_{t,h}(\tau_h)
  \left\langle
    Q^i_{t,h}(s(\tau_h),\cdot),
    \pi^i_{h}(\cdot\mid s(\tau_h))
    -
    \mu_{i,h}(\cdot\mid\tau_h)
  \right\rangle .
\end{aligned}
\]
This completes the proof.
\end{proof}

\subsection{Proof of Theorem \ref{thm:main}}

Fix a player \(i\) and an arbitrary comparator policy \(\mu_i\in\Pi_i^{\mathrm{gen}}\). For each layer $h$, define
\[
\begin{aligned}
  R_{i,h}(\mu_i)
  :=
  \sum_{\tau_h\in\Hist_h}\sum_{t=1}^T
  d^{i,\mu}_{t,h}(\tau_h)
  \left\langle
    Q^i_{t,h}(s(\tau_h),\cdot),
    \pi^i_{t,h}(\cdot\mid s(\tau_h))-\mu_{i,h}(\cdot\mid\tau_h)
  \right\rangle .
\end{aligned}
\]
Now, apply Lemma~\ref{lem:weighted-oomd} separately for every history $\tau_h\in\Hist_h$, with
\( 
  x_t=\pi^i_{t,h}(\cdot\mid s(\tau_h)),
  \;
  u=\mu_{i,h}(\cdot\mid\tau_h),\;
  \ell_t=Q^i_{t,h}(s(\tau_h),\cdot),
  \;
  q_t=d^{i,\mu}_{t,h}(\tau_h),
\) 
and learning rate $\eta_h$.  Summing over $\tau\in\Hist_h$ yields
\[
\begin{aligned}
  R_{i,h}(\mu_i)
  &\le
  \frac{B_i}{\eta_h}
  \sum_{\tau_h\in\Hist_h}
  \left(
    d^{i,\mu}_{1,h}(\tau_h)
    +
    \sum_{t=1}^{T-1}
    |d^{i,\mu}_{t+1,h}(\tau_h)-d^{i,\mu}_{t,h}(\tau_h)|
  \right) \\
  &\quad+
  \frac{\eta_h\rho_i}{2}
  \sum_{\tau_h\in\Hist_h}\sum_{t=1}^T
  d^{i,\mu}_{t,h}(\tau_h)
  \norm{Q^i_{t,h}(s(\tau_h),\cdot)-Q^i_{t-1,h}(s(\tau_h),\cdot)}_\infty^2.
\end{aligned}
\]
Since $\sum_{\tau_h\in\Hist_h}d^{i,\mu}_{1,h}(\tau_h)=1$ and is the same at every episode, Lemma~\ref{lem:qvar} gives
\[
  R_{i,h}(\mu_i)
  \le
  \frac{B_i(1+\TV^{i,\mu}_h)}{\eta_h}
  +
  \frac12\eta_h\rho_iH^2
  +
  18\eta_h\rho_i m^2\rhobar^2H^6\eta_H^2T.
\]
We now sum over $h$.  First,
\[
  \sum_{h=1}^H \frac{B_i}{\eta_h}
  =
  \frac{B_i}{\eta_0}\sum_{h=1}^H T^{\alpha_h}
  \le
  \frac{HB_i}{\eta_0}T^\beta.
\]
Second, for $h\ge2$, Lemma~\ref{lem:history-tv} gives
\[
  \frac{B_i\TV^{i,\mu}_h}{\eta_h}
  \le
  3mB_i\rhobar H T\frac{\sum_{\ell=1}^{h-1}\eta_\ell}{\eta_h}.
\]
Since $\eta_1\le\cdots\le\eta_H$,
\( 
  \sum_{\ell=1}^{h-1}\eta_\ell
  \le
  H\eta_{h-1}.
\) 
Moreover,
\( 
  \frac{\eta_{h-1}}{\eta_h}
  =
  T^{-(\alpha_{h-1}-\alpha_h)}
  =
  T^{-3/(3H+1)}.
\) 
Therefore,
\[
  \frac{B_i\TV^{i,\mu}_h}{\eta_h}
  \le
  3mB_i\rhobar H^2 T^{1-3/(3H+1)}
  =
  3mB_i\rhobar H^2 T^\beta,
\]
and hence
\[
  \sum_{h=1}^H\frac{B_i\TV^{i,\mu}_h}{\eta_h}
  \le
  3mB_i\rhobar H^3T^\beta.
\]
Third,
\[
  \sum_{h=1}^H \frac12\eta_h\rho_iH^2
  \le
  \frac12\rho_iH^2\cdot H\eta_H
  \le
  \frac12\eta_0\rho_iH^3T^\beta,
\]
where we used $T^{-\alpha_H}\le T^\beta$ for $T\ge1$.  Fourth,
\begin{align*}
  \sum_{h=1}^H
  18\eta_h\rho_i m^2\rhobar^2H^6\eta_H^2T
  &=
  18\rho_i m^2\rhobar^2H^6T\eta_H^2\sum_{h=1}^H\eta_h \le
  18\rho_i m^2\rhobar^2H^6T\eta_H^2\cdot H\eta_H \\
  &\le
  18\eta_0^3\rho_i m^2\rhobar^2H^7T^{1-3\alpha_H}
  \\&\le 18\eta_0^3\rho_i m^2\rhobar^2H^7T^\beta,
\end{align*}
where the last line follows from $\alpha_H=1/(3H+1)$,
\(  
  1-3\alpha_H=\frac{3H-2}{3H+1}=\beta.
\) Choosing
\( 
    \lambda_i=\frac{1}{|A_i|},
\) 
gives
\( 
    \rho_i
    =
    1+|A_i|\lambda_i
    =
    2,
    \;
    \bar\rho
    =
    2,
\) 
and
\( 
    B_i
    =
    2\log(|A_i|+1).
\) Combining all four pieces results in
\begin{align*}
  \sum_{h=1}^H R_{i,h}(\mu_i)
  &\le
  \left(
    \frac{HB_i}{\eta_0}
    +3mB_i\rhobar H^3
    +\frac12\eta_0\rho_iH^3
    +18\eta_0^3\rho_i m^2\rhobar^2H^7
  \right)T^\beta.
    \\&\le
    \left(
    \frac{2H\log(|A_i|+1)}{\eta_0}
    +
    12mH^3\log(|A_i|+1)
    +
    \eta_0H^3
    +
    144\eta_0^3m^2H^7
    \right)T^\beta .
\end{align*}
Define
\[
    C_i^{\mathrm{SE}}
    :=
    \frac{2H\log(|A_i|+1)}{\eta_0}
    +
    12mH^3\log(|A_i|+1)
    +
    \eta_0H^3
    +
    144\eta_0^3m^2H^7 .
\]
By Lemma~\ref{lem:history-vdiff}, we have
\[
\sum_{t=1}^T
\Bigl(
V^{i,\pi}_{t,1}(s_1)-V^{i,\mu_i}_{t,1}(s_1)
\Bigr)
=
\sum_{h=1}^H R_{i,h}(\mu_i).
\]
The left-hand side is precisely the regret of player~\(i\) with respect to the comparator policy \(\mu_i\).
Since the bound on
\( 
\sum_{t=1}^T
\Bigl(
V^{i,\pi}_{t,1}(s_1)-V^{i,\mu_i}_{t,1}(s_1)
\Bigr)
\) 
holds uniformly for all \(\mu_i\in\Pi_i^{\mathrm{gen}}\), taking the supremum over
\(\mu_i\in\Pi_i^{\mathrm{gen}}\) yields
\[
\Reg^{\mathrm{gen},i}_T
=
\sup_{\mu_i\in\Pi_i^{\mathrm{gen}}}
\sum_{t=1}^T
\Bigl(
V^{i,\pi}_{t,1}(s_1)-V^{i,\mu_i}_{t,1}(s_1)
\Bigr)
\le
C_i^{\mathrm{SE}} T^\beta,
\qquad
\beta=\frac{3H-2}{3H+1}.
\]
Now we let
\( 
    C_*^{\mathrm{SE}}:=\max_{i\in[m]} C_i^{\mathrm{SE}},
    \;
    T_\varepsilon
    :=
    \left\lceil
    \left(
        \frac{C_*^{\mathrm{SE}}}{\varepsilon}
    \right)^{(3H+1)/3}
    \right\rceil .
\) 
In particular, since \(\beta-1=-3/(3H+1)\), this choice ensures that, for every player
\(i\in[m]\),
\[
    \frac{\Reg^{\mathrm{gen},i}_{T_\varepsilon}}{T_\varepsilon}
    \le
    C_i^{\mathrm{SE}}T_\varepsilon^{-3/(3H+1)}
    \le
    C_*^{\mathrm{SE}}T_\varepsilon^{-3/(3H+1)}
    \le
    \varepsilon .
\]
Equivalently,
\( 
    \Reg^{\mathrm{gen},i}_{T_\varepsilon}
    \le
    \varepsilon T_\varepsilon .
\) 
The \(\varepsilon\)-CCE guarantee follows by applying the standard
no-regret-to-CCE equivalence to the empirical distribution
\( 
    \widehat{\Pi}_{T_\varepsilon}
\)
\citep{CesaBianchiLugosi2006}.

\section{Convergence Analysis of Algorithm \ref{alg:blocked-bandit-oomd}}
\label{app:episodic-partial-feedback}

\subsection{Comparator smoothing}
\label{app:pf-comparator-smoothing}

Algorithm \ref{alg:blocked-bandit-oomd} plays policies in the restricted simplex
\(\Delta_i^\zeta\).  We therefore first compare an arbitrary history-dependent
policy with its \(\zeta\)-smoothed version.  For
\(\mu^i\in\Pi_i^{\mathrm{gen}}\), define
\[
    \mu_h^{i,\zeta}(\cdot\mid\tau_h)
    :=
    (1-\zeta)\mu_h^i(\cdot\mid\tau_h)
    +
    \zeta\,\Unif(A_i),
    \qquad h\in[H],\ \tau_h\in\Hist_h .
\]
Then \(\mu_h^{i,\zeta}(\cdot\mid\tau_h)\in\Delta_i^\zeta\) for every admissible 
history, and we have the following guarantee for the \emph{smoothing}.

\begin{lemma}
\label{lem:pf-comparator-smoothing}
Fix player \(i\), a product Markov opponent profile \(\bm{\pi}^{-i}\), and a
history-dependent policy \(\mu^i\in\Pi_i^{\mathrm{gen}}\). Consequently, for every product Markov profile \(\bm{\pi}\), it holds that
\[
    V^{i,\pi}_{t,1}(s_1)-V^{i,\mu^i}_{t,1}(s_1)
    \le
    V^{i,\pi}_{t,1}(s_1)-V^{i,\mu^{i,\zeta}}_{t,1}(s_1)
    +2\zeta H^2 .
\]
\end{lemma}

\begin{proof}
For \(r=0,1,\ldots,H\), define a policy \(\nu^{i,r}\) by
\[
    \nu_h^{i,r}(\cdot\mid\tau_h)
    :=
    \begin{cases}
        \mu_h^{i,\zeta}(\cdot\mid\tau_h), & h\le r,\\
        \mu_h^i(\cdot\mid\tau_h), & h>r.
    \end{cases}
\]
Then \(\nu^{i,0}=\mu^i\) and \(\nu^{i,H}=\mu^{i,\zeta}\).  Hence
\begin{align}\label{eq:quaso4}
    &V_{t,1}^{i,\mu^{i,\zeta}}(s_1)
    -
    V_{t,1}^{i,\mu^i}(s_1)  =
    \sum_{r=1}^H
    \left[
        V_{t,1}^{i,\nu^{i,r}}(s_1)
        -
        V_{t,1}^{i,\nu^{i,r-1}}(s_1)
    \right].
\end{align}
The policies \(\nu^{i,r}\) and \(\nu^{i,r-1}\) agree on all layers \(\ell<r\), and their only possible difference is at layer \(r\). For \(\tau_r\in\Hist_r\), let
\(
    Q^i_{t,r}(s(\tau_r),\cdot)
\)
denote the Q-function whose \(a_i\)-th coordinate is the expected
payoff obtained by taking \(a_i\) at \(\tau_r\), then following \( \pi^{i}\), while the opponents follow \(\bm\pi^{-i}\). Since the remaining total cost is at most \(H\),
\(
    \|Q^i_{t,r}(s(\tau_r),\cdot)\|_\infty\le H .
\)
Moreover, by the definition of the smoothed policy,
\(
    \mu_r^{i,\zeta}(\cdot\mid\tau_r)
    -
    \mu_r^i(\cdot\mid\tau_r)
    =
    \zeta\bigl(\Unif(A_i)-\mu_r^i(\cdot\mid\tau_r)\bigr).
\)
Hence,
\[
    \left\|
        \mu_r^{i,\zeta}(\cdot\mid\tau_r)
        -
        \mu_r^i(\cdot\mid\tau_r)
    \right\|_1
    =
    \zeta
    \left\|
        \Unif(A_i)-\mu_r^i(\cdot\mid\tau_r)
    \right\|_1
    \le 2\zeta .
\]
Thus, by Lemma \ref{lem:history-vdiff}
\[
\begin{aligned}
    \left|
        V_{t,1}^{i,\nu^{i,r}}(s_1)
        -
        V_{t,1}^{i,\nu^{i,r-1}}(s_1)
    \right|                                     
    &\le
    \sum_{\tau_r\in\Hist_r}
    d_{t,r}^{i,\nu^{i,r-1}}(\tau_r)
    \|Q_{t,r}(\tau_r,\cdot)\|_\infty
    \left\|
        \mu_r^{i,\zeta}(\cdot\mid\tau_r)
        -
        \mu_r^i(\cdot\mid\tau_r)
    \right\|_1                                                  \\
    &\le
    \sum_{\tau_r\in\Hist_r} d_{t,r}^{i,\nu^{i,r-1}}(\tau_r)\cdot H\cdot 2\zeta
    =2\zeta H .
\end{aligned}
\]
Summing over \(r=1,\ldots,H\) by \eqref{eq:quaso4} we obtain
\[
    \left|
        V_{t,1}^{i,\mu^{i,\zeta}}(s_1)
        -
        V_{t,1}^{i,\mu^i}(s_1)
    \right|
    \le
    2\zeta H^2 .
\]
The final claim follows
by adding and subtracting
\(V_{t,1}^{i,\mu^{i,\zeta}}(s_1)\) to \(  V^{i,\pi}_{t,1}(s_1)-V^{i,\mu^i}_{t,1}(s_1) \).
\end{proof}

\subsection{Uniform estimation of \texorpdfstring{\(Q\)}{Q}-values}
\label{app:pf-q-estimation}

At block \(k\), conditional on the past, the block policy \(\pi_k\) is fixed
and the \(B\) trajectories in that block are independent.  The reachability condition, together with the action-probability floor for player~\(i\), gives a uniform lower bound on the probability of observing each state--player-\(i\)-action pair. We state the resulting estimation guarantee for the player-\(i\) \(Q\)-function estimates as follows.

\begin{lemma}
\label{lem:pf-uniform-q-estimation}
Fix \(K,B\in\mathbb N\), \(\delta\in(0,1)\), and
\(\xi>0\).  
\( 
    M_K
    :=
    KmA_{\max}S,
    \;
    u_K:=\log\frac{4M_K}{\delta},\; 
    n_Q
    :=
    \left\lceil
        \frac{H^2}{2\xi^2}u_K
    \right\rceil,
    \;
    p_{\min}:=\frac{\kappa\zeta}{A_{\max}} .
\) 
If
\( 
    B
    \ge
    \left\lceil
        \frac{8}{p_{\min}}(n_Q+u_K)
    \right\rceil,
\) 
then, with probability at least \(1-\delta\),
\[
    \max_{\substack{k\in[K],i\in[m],h\in[H],\\ s\in S_h, a\in A_i}}
    \left|
        \widehat Q_{k,h}^i(s,a)-Q_{k,h}^{i}(s,a)
    \right|
    \le
    \xi .
\]
\end{lemma}

\begin{proof}
First, fix a tuple \((k,i,h,s,a_i)\), with \(s\in S_h\).  Then, let \(\mathcal F_{k-1}\) be
the sigma-field generated by all randomness before block \(k\). Let \(\mathbb P\) denote the probability measure induced by Algorithm \ref{alg:blocked-bandit-oomd} and the game itself. Conditional on \(\mathcal F_{k-1}\) under \(\mathbb P\), the policy \(\pi_k\) is fixed and the \(B\) trajectories in block \(k\) are independent.
For episode \(r\in[B]\), define
\( 
    X_r:=\mathbf 1\{s_h^{k,r}=s,\ a_h^{i,k,r}=a_i\} , 
\) thus  \(N_{k,i}(s,a_i)=\sum_{r=1}^B X_r\) by \eqref{eq:quasocount}. Then, by Definition \ref{eq:quasocount} conditional on \(\mathcal F_{k-1}\), the random variables \(X_1,\ldots,X_B\) are i.i.d. Bernoulli random variables with success probability
\( 
p_{k,i,h,s,a_i}
:=
\mathbb P(s_h^{k,r}=s, a_h^{i,k,r}=a_i\mid \mathcal F_{k-1}).
\) 
By the reachability assumption and the action-floor condition,
\[
p_{k,i,h,s,a_i}
=
\mathbb P(s_h^{k,r}=s\mid \mathcal F_{k-1})
\pi_{k,h}^i(a_i\mid s)
\ge
\kappa\frac{\zeta}{|A_i|}
\ge
\frac{\kappa\zeta}{A_{\max}}
=:p_{\min}
\]
\(\mathbb P\)-almost surely. Hence, under the conditional law \(\mathbb P(\cdot\mid\mathcal F_{k-1})\),
\[
N_{k,i}(s,a_i)\sim \mathrm{Bin}(B,p_{k,i,h,s,a_i}),
\qquad
p_{k,i,h,s,a_i}\ge p_{\min}\quad \mathbb P\text{-a.s.}
\]
Since the lower-tail probability of a binomial random variable is nonincreasing in its success probability, it follows that, for every threshold \(n\),
\[
\mathbb P\!\left(N_{k,i}(s,a_i)<n\mid \mathcal F_{k-1}\right)
\le
\Pr\!\left(\mathrm{Bin}(B,p_{\min})<n\right)
\]
\(\mathbb P\)-almost surely. Since \(Bp_{\min}\ge 8(n_Q+u_K)\), we have
\(n_Q\le Bp_{\min}/2\). Therefore, Chernoff's inequality gives
\[
\begin{aligned}
    \mathbb P\!\left(N_{k,i}(s,a_i)<n_Q\mid\mathcal F_{k-1}\right)
    &\le
    \Pr\!\left(\mathrm{Bin}(B,p_{\min})<\frac12 Bp_{\min}\right)  \\
    &\le
    \exp\!\left(-\frac{Bp_{\min}}8\right)
    \le
    e^{-u_K},
\end{aligned}
\]
\(\mathbb P\)-almost surely. Now, let
\( 
    A:=\{N_{k,i}(s,a_i)\ge n_Q\}.
\) 
On \(A\), the estimator is
\[
    \widehat Q_{k,h}^i(s,a_i)
    =
    \frac{1}{N_{k,i}(s,a_i)}
    \sum_{r=1}^B X_rY_h^{i,k,r}.
\]
Fix a realization \(x_1,\ldots,x_B\in\{0,1\}\) with
\(\sum_{r=1}^B x_r\ge n_Q\). Conditional on
\(\mathcal F_{k-1}\) and on \(\{X_r=x_r,\ r\in[B]\}\), the variables
\(\{Y_h^{i,k,r}:x_r=1\}\) are independent, since the \(B\) trajectories in block \(k\) are conditionally independent. Moreover, for each \(r\) with \(x_r=1\),
\[
\begin{aligned}
\mathbb E\!\left[
    Y_h^{i,k,r}
    \,\middle|\,
    \mathcal F_{k-1},\ s_h^{k,r}=s,\ a_h^{i,k,r}=a_i
\right]
&=
\mathbb E\!\left[
    \sum_{\tau=h}^H \ell_\tau^i(s_\tau^{k,r},\bm a_\tau^{k,r})
    \,\middle|\,
    \mathcal F_{k-1},\ s_h^{k,r}=s,\ a_h^{i,k,r}=a_i
\right] \\
&=
Q_{k,h}^i(s,a_i).
\end{aligned}
\]
as conditional on
\(\mathcal F_{k-1}\), the policy \(\pi_k\) is fixed, so the conditional expected tail cost from layer \(h\) onward is precisely the state-action value \(Q_{k,h}^i(s,a_i)\). Also,
\(Y_h^{i,k,r}\in[0,H]\). Therefore, Hoeffding's inequality gives
\[
\mathbb P\!\left(
    \left|
        \frac{1}{\sum_{r=1}^B x_r}
        \sum_{r=1}^B x_rY_h^{i,k,r}
        -
        Q_{k,h}^i(s,a_i)
    \right|>\xi
    \,\middle|\,
    \mathcal F_{k-1},\ X_r=x_r,\ r\in[B]
\right)
\le
2\exp\!\left(-\frac{2n_Q\xi^2}{H^2}\right).
\]
Averaging over all realizations \(x_1,\ldots,x_B\) with
\(\sum_{r=1}^B x_r\ge n_Q\), we obtain
\[
\mathbb P\!\left(
    \left|
        \widehat Q_{k,h}^i(s,a_i)-Q_{k,h}^i(s,a_i)
    \right|>\xi,\ A
    \,\middle|\,
    \mathcal F_{k-1}
\right)
\le
2\exp\!\left(-\frac{2n_Q\xi^2}{H^2}\right)
\le
2e^{-u_K}.
\]
Combining the bound on \(A^c=\{N_{k,i}(s,a_i)<n_Q\}\) with the bound on the estimation error on \(A\), we obtain, for the fixed tuple \((k,i,h,s,a_i)\),
\[
\begin{aligned}
\mathbb P\!\left(
        \left|
            \widehat Q_{k,h}^i(s,a_i)-Q_{k,h}^{i}(s,a_i)
        \right|>\xi
        \,\middle|\,\mathcal F_{k-1}
    \right) 
& \le
\mathbb P\!\left(A^c\mid\mathcal F_{k-1}\right)
+
\mathbb P\!\left(
    \left|
        \widehat Q_{k,h}^i(s,a_i)-Q_{k,h}^{i}(s,a_i)
    \right|>\xi,\ A
    \,\middle|\,
    \mathcal F_{k-1}
\right)  \\
& \le
e^{-u_K}+2e^{-u_K}
=
3e^{-u_K},
\end{aligned}
\]
\(\mathbb P\)-almost surely. Taking expectations on both sides gives the unconditional bound
\[
\mathbb P\!\left(
        \left|
            \widehat Q_{k,h}^i(s,a_i)-Q_{k,h}^{i}(s,a_i)
        \right|>\xi
    \right)
\le
3e^{-u_K}.
\]
Now take a union bound over all tuples
\( 
(k,i,h,s,a_i)
\;\text{with}\;
k\in[K],\ i\in[m],\ h\in[H],\ s\in S_h, a_i\in A_i .
\) 
The number of such tuples is at most \(M_K=KmA_{\max}S\). Therefore,
\[
\begin{aligned}
&\mathbb P\!\left(
    \max_{\substack{k\in[K],i\in[m],h\in[H],\\ s\in S_h,\ a_i\in A_i}}
    \left|
        \widehat Q_{k,h}^i(s,a_i)-Q_{k,h}^{i}(s,a_i)
    \right|
    >\xi
\right)\le
3M_Ke^{-u_K}
=
3M_K\frac{\delta}{4M_K}
\le
\delta .
\end{aligned}
\]
Equivalently, with probability at least \(1-\delta\),
\[
    \max_{\substack{k\in[K],i\in[m],h\in[H],\\ s\in S_h,\ a_i\in A_i}}
    \left|
        \widehat Q_{k,h}^i(s,a_i)-Q_{k,h}^{i}(s,a_i)
    \right|
    \le
    \xi .
\]
This proves the result.

\end{proof}

\subsection{Regret with uniformly accurate estimates}
\label{app:pf-regret-accurate-estimates}

We next prove a deterministic regret bound for Algorithm~\ref{alg:blocked-bandit-oomd} on the event that all estimated \(Q\)-vectors are uniformly accurate. Define
\( 
    \mathcal E_\xi
    :=
    \left\{
    \max_{k,i,h,s,a_i}
    \left|
        \widehat Q_{k,h}^i(s,a_i)-Q_{k,h}^i(s,a_i)
    \right|
    \le \xi
    \right\}.
\) 
On this event, the regret analysis of the estimated-\(Q\) functions differs from the full-feedback case only through the estimation-error terms that appear.

\begin{proposition}
\label{prop:pf-accurate-estimates}
Suppose Algorithm~\ref{alg:blocked-bandit-oomd} is run for \(K\) blocks.  On
the event \(\mathcal E_\xi\), for every player \(i\in[m]\),
\[
\begin{aligned}
    \sup_{\mu^i\in\Pi_i^{\mathrm{gen}}}
    \frac1K
    \sum_{k=1}^K
    \left[
        V^{i,\pi}_{k,1}(s_1)
        -
        V^{i,\mu^i}_{k,1}(s_1)
    \right]
    \le\;&
    \widetilde C_i^{\mathrm{SE}}K^{-3/(3H+1)}
    +2H\xi
    +2\zeta H^2
    +8\eta_0H\xi^2K^{-\alpha_H},
\end{aligned}
\]
where
\( 
    \widetilde C_i^{\mathrm{SE}}
    :=
        \frac{2H\log(|A_i|+1)}{\eta_0}
        +12mH^3\log(|A_i|+1)
        +\eta_0H^3
        +288\eta_0^3m^2H^7.
\) 
\end{proposition}

\begin{proof}
Fix player \(i\) and a comparator \(\mu^i\in\Pi_i^{\mathrm{gen}}\).  Let
\(\mu^{i,\zeta}\) be its \(\zeta\)-smoothed version defined as in Appendix \ref{app:pf-comparator-smoothing}.  By
Lemma~\ref{lem:pf-comparator-smoothing}, it is enough to control the regret against
\(\mu^{i,\zeta}\) and then add \(2\zeta H^2\) to the average regret.

\vspace{6pt}
\noindent 
For block \(k\), let
\(d_{k,h}^{i,\mu^\zeta}\in\Delta(\Hist_h)\) be the history distribution at layer
\(h\) generated when player \(i\) uses \(\mu^{i,\zeta}\) and the opponents use
\(\bm \pi_k^{-i}\). Applying Lemma~\ref{lem:history-vdiff} to the block profile
\(\bm \pi_k\) gives
\[
\begin{aligned}
    &V^{i,\pi^i}_{k,1}(s_1)
    -
    V^{i,\mu^{i,\zeta}}_{k,1}(s_1)         =
    \sum_{h=1}^H
    \sum_{\tau_h\in\Hist_h}
    d_{k,h}^{i,\mu^\zeta}(\tau_h)
    \left\langle
        Q_{k,h}^i(s(\tau_h),\cdot),
        \pi_{k,h}^i(\cdot\mid s(\tau_h))
        -
        \mu_h^{i,\zeta}(\cdot\mid\tau_h)
    \right\rangle,
\end{aligned}
\]
On the event \(\mathcal E_\xi\), for every history \(\tau_h\), it holds almost surely that
\[
\begin{aligned}
    &\left|
    \left\langle
        Q_{k,h}^i(s(\tau_h),\cdot)-\widehat Q_{k,h}^i(s(\tau_h),\cdot),
        \pi_{k,h}^i(\cdot\mid s(\tau_h))
        -
        \mu_h^{i,\zeta}(\cdot\mid\tau_h)
    \right\rangle
    \right|                                          \le
    \xi
    \left\|
        \pi_{k,h}^i(\cdot\mid s(\tau_h))
        -
        \mu_h^{i,\zeta}(\cdot\mid\tau_h)
    \right\|_1
    \le 2\xi .
\end{aligned}
\]
Since \(d_{k,h}^{i,\mu^\zeta}\in\Delta(\Hist_h)\) has total mass one at each layer, replacing
\(Q_{k,h}^i\) by \(\widehat Q_{k,h}^i\) costs at most \(2HK\xi\) in cumulative regret.

\vspace{6pt}
\noindent 
It remains to bound the corresponding weighted regret expression with
\(\widehat Q_{k,h}^i\) in place of \(Q_{k,h}^i\).  Lemma~\ref{lem:weighted-oomd} is
stated for the full simplex, but its proof uses only first-order optimality over
a convex decision set, the Bregman diameter bound \(D_i\le B_i\), and
\(1/\rho_i\)-strong convexity of the smoothed entropy regularizer.  Hence the same bound applies verbatim on the
restricted simplex \(\Delta_i^\zeta\).  Applying Lemma~\ref{lem:weighted-oomd}
separately to every history \(\tau\in\Hist_h\), with
\[
    x_k=\pi_{k,h}^i(\cdot\mid s(\tau_h)),
    \quad
    u=\mu_h^{i,\zeta}(\cdot\mid\tau_h),
    \quad
    \ell_k=\widehat Q_{k,h}^i(s(\tau_h),\cdot), \quad
    q_k=d_{k,h}^{i,\mu^\zeta}(\tau_h),
    \quad
    M_k=\widehat Q_{k-1,h}^i(s(\tau_h),\cdot),
\]
and summing over histories gives
\begin{align} \label{eq:oomd-qhat}
    R_{i,h}^{\mathrm{est}}(\mu^{i,\zeta})
    \le\;&
    \frac{B_i(1+\TV_h^{i,\mu^\zeta})}{\eta_h}       +
    \frac{\eta_h\rho_i}{2}
    \sum_{\tau\in\Hist_h}
    \sum_{k=1}^K
    d_{k,h}^{i,\mu^\zeta}(\tau)
    \left\|
        \widehat Q_{k,h}^i(s(\tau),\cdot)
        -
        \widehat Q_{k-1,h}^i(s(\tau),\cdot)
    \right\|_\infty^2,
\end{align}
where
\( 
    \TV_h^{i,\mu^\zeta}
    :=
    \sum_{k=1}^{K-1}
    \left\|
        d_{k+1,h}^{i,\mu^\zeta}-d_{k,h}^{i,\mu^\zeta}
    \right\|_1 .
\) 

\vspace{6pt}
\noindent
We next bound first order path of the policies generated by estimated \(Q\)-functions. The proof of Lemma~\ref{lem:movement}, which is based on
Lemma~\ref{lem:prox-lip}, only uses the fact that the feedback vectors lie in
\([0,H]^{A_j}\).  Since \(\widehat Q_{k,h}^j\in[0,H]^{A_j}\), repeating the same calculation, for every player \(j\), layer \(h\), state \(s\in S_h\), and block
\(k\), we obtain
\begin{align} \label{eq:quaso5}
    \left\|
        \pi_{k+1,h}^j(\cdot\mid s)
        -
        \pi_{k,h}^j(\cdot\mid s)
    \right\|_1
    \le
    3\rho_jH\eta_h
    \le
    3\rhobar H\eta_h .
\end{align}
Applying the Lemma~\ref{lem:history-tv} to the block sequence
\(\pi_1,\ldots,\pi_K\), and using \eqref{eq:quaso5}, yields
\[
    \TV_h^{i,\mu^\zeta}
    \le
    3m\rhobar HK\sum_{\ell<h}\eta_\ell .
\]
We now control the variation of the estimated \(Q\)-function.  For the exact values, Lemma~\ref{lem:q-perturb} and \eqref{eq:quaso5} imply, for \(k\ge2\),
\[
\begin{aligned}
    \left\|
        Q_{k,h}^i(s,\cdot)-Q_{k-1,h}^i(s,\cdot)
    \right\|_\infty
    &\le
    2H^2\sum_{j=1}^m
    \|\pi_{k,h}^j-\pi_{k-1,h}^j\|_{\infty,1}                         \\
    &\le
    2H^2\cdot m\cdot 3\rhobar H\eta_H
    =6m\rhobar H^3\eta_H .
\end{aligned}
\]
 On the good event \(\mathcal E_\xi\), for \(k\ge2\),
\[
\begin{aligned}
    \left\|
        \widehat Q_{k,h}^i(s,\cdot)-\widehat Q_{k-1,h}^i(s,\cdot)
    \right\|_\infty
    &\le
    \left\|
        Q_{k,h}^i(s,\cdot)-Q_{k-1,h}^i(s,\cdot)
    \right\|_\infty
    +2\xi \le
    6m\rhobar H^3\eta_H+2\xi 
\end{aligned}
\] Then, we have
\[
\begin{aligned}
    \left\|
        \widehat Q_{k,h}^i(s,\cdot)-\widehat Q_{k-1,h}^i(s,\cdot)
    \right\|_\infty^2
    \le 72m^2\rhobar^2H^6\eta_H^2+8\xi^2,
\end{aligned}
\]
where the last line follows from the inequality \((a+b)^2\le2a^2+2b^2\).
On the first block, \(\widehat Q_{0,h}^i\equiv0\) and
\(\widehat Q_{1,h}^i\in[0,H]^{A_i}\), so
\( 
    \left\|
        \widehat Q_{1,h}^i(s,\cdot)-\widehat Q_{0,h}^i(s,\cdot)
    \right\|_\infty^2
    \le H^2 .
\) 
Substituting these estimates into the weighted OOMD bound \eqref{eq:oomd-qhat} gives, on every
layer \(h\),
\begin{align} \label{eq:quaso6}
    R_{i,h}^{\mathrm{est}}(\mu^{i,\zeta})
    \le\;&
    \frac{B_i}{\eta_h}
    +
    \frac{B_i\TV_h^{i,\mu^\zeta}}{\eta_h}
    +
    \frac12\eta_h\rho_iH^2  +
    36\eta_h\rho_i m^2\rhobar^2H^6\eta_H^2K
    +
    4\eta_h\rho_i\xi^2K .
\end{align}
We now sum over \(h\in[H]\), using the same bounds on the learning-rate related terms as in the
proof of Theorem~\ref{thm:main} in Appendix~\ref{app:proof-main}.  Namely, with
\(\beta=(3H-2)/(3H+1)\),
\begin{align}\label{eq:step-block1}
    \sum_{h=1}^H\frac{B_i}{\eta_h}
    \le
    \frac{HB_i}{\eta_0}K^\beta,\quad 
    \sum_{h=1}^H\frac{B_i\TV_h^{i,\mu^\zeta}}{\eta_h}
    \le
    3mB_i\rhobar H^3K^\beta,
    \quad 
    \sum_{h=1}^H\frac12\eta_h\rho_iH^2
    \le
    \frac12\eta_0\rho_iH^3K^\beta,
\end{align}
and
\begin{align}\label{eq:step-block2}
    \sum_{h=1}^H
    36\eta_h\rho_i m^2\rhobar^2H^6\eta_H^2K
    \le
    36\eta_0^3\rho_i m^2\rhobar^2H^7K^\beta .
\end{align}
Then the final term in \eqref{eq:quaso6} satisfies
\begin{align}\label{eq:step-block3}
    \sum_{h=1}^H 4\eta_h\rho_i\xi^2K
    \le
    4\eta_0\rho_iH\xi^2K^{1-\alpha_H} .
\end{align}
Choosing
\( 
    \lambda_i=\frac{1}{|A_i|},
\) 
gives
\( 
    \rho_i
    =
    1+|A_i|\lambda_i
    =
    2,
    \;
    \bar\rho
    =
    2,
\) 
and
\( 
    B_i
    =
    2\log(|A_i|+1)
\). Then,  substituting these bounds into the \eqref{eq:step-block1}, \eqref{eq:step-block2}, \eqref{eq:step-block3} yields, 
\begin{align} \label{eq:onemore}
    R_{i,h}^{\mathrm{est}}(\mu^{i,\zeta})
    \le\;&
    \widetilde C_i^{\mathrm{SE}}K^{\beta}
    +8\eta_0H\xi^2K^{1-\alpha_H}.
\end{align}
Finally, combining \eqref{eq:onemore} with the error from comparing the played policies against \(\mu^{i,\zeta}\), namely \(2\zeta H^2K\), and the error from \(Q\)-function estimation, namely \(2HK\xi\), we obtain,
\[
\begin{aligned}
    \sum_{k=1}^K
    \left[
        V_{k,1}^{i,\pi}(s_1)
        -
        V_{k,1}^{i,\mu^i}(s_1)
    \right]
    \le\;&
    \widetilde C_i^{\mathrm{SE}}K^\beta
    +2HK\xi
    +2\zeta H^2K  +
    8\eta_0H\xi^2K^{1-\alpha_H} .
\end{aligned}
\]
Dividing by \(K\) and taking the supremum over
\(\mu^i\in\Pi_i^{\mathrm{gen}}\) proves the proposition.
\end{proof}

\subsection{Proof of Theorem~\ref{thm:episodic-partial-feedback}}
\label{app:proof-episodic-partial-feedback}

\begin{proof}
Apply Lemma~\ref{lem:pf-uniform-q-estimation} with
\(   
    K=K_\varepsilon,
    \;
    B=B_\varepsilon,
    \;
    \zeta=\zeta_\varepsilon,
    \;
    \xi=\xi_\varepsilon .
\) 
By the definition of \(B_\varepsilon\), the event \(\mathcal E_{\xi_\varepsilon}\)
holds with probability at least \(1-\delta\).  On this event,
Proposition~\ref{prop:pf-accurate-estimates} gives, for every player
\(i\in[m]\),
\[
\begin{aligned}
    \sup_{\mu^i\in\Pi_i^{\mathrm{gen}}}
    \frac1{K_\varepsilon}
    \sum_{k=1}^{K_\varepsilon}
    \left[
        V_{k,1}^{i,\pi}(s_1)
        -
        V_{k,1}^{i,\mu^i}(s_1)
    \right]
    \le\;&
    \widetilde C_i^{\mathrm{SE}}K_\varepsilon^{-3/(3H+1)}
    +2H\xi_\varepsilon
    +2\zeta_\varepsilon H^2  +
    8\eta_0H\xi_\varepsilon^2K_\varepsilon^{-\alpha_H} .
\end{aligned}
\]
\begin{itemize}
\item By the definition of \(K_\varepsilon\),
\( 
    \widetilde C_i^{\mathrm{SE}}K_\varepsilon^{-3/(3H+1)}
    \le
    \widetilde C_*^{\mathrm{SE}}K_\varepsilon^{-3/(3H+1)}
    \le
    \frac\varepsilon4 .
\) 
\item By the definition of \(\xi_\varepsilon\),
\( 
    2H\xi_\varepsilon\le\frac\varepsilon4 .
\) 
\item By the definition of \(\zeta_\varepsilon\),
\( 
    2\zeta_\varepsilon H^2=\frac\varepsilon4 .
\)
\item Finally, since \(K_\varepsilon^{-\alpha_H}\le1\), \(\rho_i=2\), and
\(\xi_\varepsilon^2\le\varepsilon/(32\eta_0 H)\),
\( 
    8\eta_0H\xi_\varepsilon^2K_\varepsilon^{-\alpha_H}
    \le
    8\eta_0 H\xi_\varepsilon^2
    \le
    \frac\varepsilon4 .
\) 
\end{itemize}
Thus, with probability at least \(1-\delta\), for every player \(i\in[m]\),
\begin{align}\label{eq:epsilonblock}
    \sup_{\mu^i\in\Pi_i^{\mathrm{gen}}}
    \frac1{K_\varepsilon}
    \sum_{k=1}^{K_\varepsilon}
    \left[
        V_{k,1}^{i,\pi}(s_1)
        -
        V_{k,1}^{i,\mu^i}(s_1)
    \right]
    \le
    \varepsilon .
\end{align}
It remains to translate this block-level guarantee into a guarantee over the
actual trajectory episodes. By construction, during block \(k\) the algorithm
plays the same product Markov policy \(\pi_k\) for all
\(B_\varepsilon\) episodes in that block. Therefore
\[
\begin{aligned}
    \frac1{N_\varepsilon}
    \sum_{n=1}^{N_\varepsilon}
    \left[
        V_n^{i,\pi}(s_1)
        -
        V_n^{i,\mu^i}(s_1)
    \right]
    &=
    \frac1{K_\varepsilon B_\varepsilon}
    \sum_{k=1}^{K_\varepsilon}
    \sum_{r=1}^{B_\varepsilon}
    \left[
        V_{k,1}^{i,\pi}(s_1)
        -
        V_{k,1}^{i,\mu^i}(s_1)
    \right]=
    \frac1{K_\varepsilon}
    \sum_{k=1}^{K_\varepsilon}
    \left[
        V_{k,1}^{i,\pi}(s_1)
        -
        V_{k,1}^{i,\mu^i}(s_1)
    \right].
\end{aligned}
\]
Hence the bound given in \eqref{eq:epsilonblock} holds when the average is taken over
the \(N_\varepsilon=K_\varepsilon B_\varepsilon\) actual trajectory episodes. Moreover, the empirical distribution policies over \(\pi_1,\ldots,\pi_{{N_{\varepsilon}}}\), is
identical to the empirical distribution over the block policies
\(\pi_1,\ldots,\pi_{K_\varepsilon}\), since each policy \(\pi_k\) appears
exactly \(B_\varepsilon\) times. Thus, the empirical distribution over block
policies is an \(\varepsilon\)-approximate CCE.

\vspace{6pt}
\noindent
Finally, the total number of episodes is \(N_\varepsilon=K_\varepsilon B_\varepsilon\). By the definitions of
\(K_\varepsilon\), \(B_\varepsilon\), \(\zeta_\varepsilon\), \(n_\varepsilon\), \(u_\varepsilon\), and using the fact \(\xi_\varepsilon\le \varepsilon/(8H)\),
\[
N_\varepsilon
=
\widetilde O\left(
    \frac{A_{\max}H^6}{\kappa\varepsilon^3}
    \left(
        \frac{\widetilde C_*^{\mathrm{SE}}}{\varepsilon}
    \right)^{(3H+1)/3}
\right),
\]
where \(\widetilde O(\cdot)\) hides logarithmic factors in
\((m,S_{\mathrm{nt}},A_{\max},K_\varepsilon,1/\delta)\), and
\[
\widetilde C_*^{\mathrm{SE}}
\le
\max_{i\in[m]}
\left\{
    \frac{2H\log(|A_i|+1)}{\eta_0}
    +12mH^3\log(|A_i|+1)
    +\eta_0H^3
    +288\eta_0^3m^2H^7
\right\}.
\]
This completes the proof.

\end{proof}

\section{Proofs for the Approximation of Discounted Markov Games (Section \ref{sec:discounted-truncation})}
\label{app:discounted-truncation}

This appendix proves the truncation statements used in
Section~\ref{sec:discounted-truncation}. Throughout this appendix, we fix the discount factor \(\gamma\in(0,1)\).The argument is written once for the
unified discounted model with original horizon
\(\bar H\in\mathbb N_+\cup\{\infty\}\).  The case \(\bar H<\infty\) corresponds to the
finite-horizon discounted games, while \(\bar H=\infty\) corresponds to the discounted
infinite-horizon games.

\begin{lemma}
\label{lem:discounted-tail-unified}
For every \(\bar H\)-horizon policy profile \(\bm\sigma\in\Pi^{\mathrm{gen},\bar H}\), every
\(L\le\bar H\), and every player \(i\in[m]\),
\[
    0
    \le
    J_{i,\bar H}^\gamma(\bm{\sigma})
    -
    J_{i,L}^\gamma(\bm\sigma^{[L]})
    \le
    \tau_{L,\bar H}(\gamma)
    \le
    \frac{\gamma^L}{1-\gamma}.
\]
\end{lemma}

\begin{proof}
The policy profile \(\bm \sigma\) and its truncation \(\bm\sigma^{[L]}\) have
the same first \(L\) decision rules.  Therefore they induce the same
distribution over trajectories up to the beginning of layer \(L+1\).  Hence,
\[
\begin{aligned}
    J_{i,\bar H}^\gamma(\bm{\sigma})
    -
    J_{i,L}^\gamma(\bm\sigma^{[L]})
    & =
    \mathbb E^{\bm{\sigma}}
    \left[
        \sum_{h=L+1}^{\bar H}
        \gamma^{h-1}\ell_h^i(s_h,\bm a_h)
        \;\middle|\; s_1
    \right].
\end{aligned}
\]
The right-hand side is nonnegative.  Since \(0\le \ell_h^i\le 1\), we have
\[
    \sum_{h=L+1}^{\bar H}\gamma^{h-1} \ell_h^i(s_h,\bm a_h) \le 
    \sum_{h=L+1}^{\infty}\gamma^{h-1}
    =
    \frac{\gamma^L}{1-\gamma},
\]
which concludes the proof. 
\end{proof}

\begin{lemma}
\label{lem:discounted-single-policy-transfer}
Let \(\bm{\pi}^{[L]}\in\Pi^{\mathrm{markov},L}\), be an \(L\)-step product Markov profile, and define
\( 
    \bar{\bm{\pi}}:=\mathrm{Ext}_L(\bm{\pi}^{[L]}).
\) 
Then,
\[
\begin{aligned}
    J_{i,\bar H}^\gamma(\bar{\bm{\pi}})
    -
    J_{i,\bar H}^\gamma(\mu^i\odot\bar{\bm{\pi}}^{-i})
    \le\;&
    J_{i,L}^\gamma(\bm{\pi}^{[L]})
    -
    J_{i,L}^\gamma(\mu^{i,[L]}\odot\bm{\pi}^{[L],-i}) +
    \tau_{L,\bar H}(\gamma).
\end{aligned}
\]
\end{lemma}

\begin{proof}
Applying Lemma~\ref{lem:discounted-tail-unified} with \(\bm{\pi}=\bar{\bm{\pi}}\), we obtain
\begin{align}\label{eq:tot-disc1}
    J_{i,\bar H}^\gamma(\bar{\bm{\pi}})
    \le
    J_{i,L}^\gamma(\bm{\pi}^{[L]})
    +
    \tau_{L,\bar H}(\gamma).
\end{align}
On the other hand, the first \(L\) decision rules of
\(\mu^i\odot\bar{\bm{\pi}}^{-i}\) are exactly those of
\(\mu^{i,[L]}\odot\bm{\pi}^{[L],-i}\).  Since all costs are nonnegative, ignoring the
layers after \(L\) can only decrease the term \( 
    J_{i,\bar H}^\gamma(\mu^i\odot\bar{\bm{\pi}}^{-i})
\).  Therefore,
\begin{align}\label{eq:tot-disc2}
    J_{i,\bar H}^\gamma(\mu^i\odot\bar{\bm{\pi}}^{-i})
    \ge
    J_{i,L}^\gamma(\mu^{i,[L]}\odot\bm{\pi}^{[L],-i}).
\end{align}
Then, subtracting \eqref{eq:tot-disc2} from \eqref{eq:tot-disc1} proves the claim.
\end{proof}

\begin{proof}[Proof of Theorem~\ref{thm:discounted-transfer}]
Fix a player \(i\in[m]\) and a policy \(\mu^i\in\Pi_i^{\mathrm{gen},\bar H}\).
For each output profile \(\bm{\pi}_n^{[L]}\in\Pi^{\mathrm{markov},L}\), define
\(\bar{\bm{\pi}}_n:=\mathrm{Ext}_L(\pi_n^{[L]})\). 
By Lemma~\ref{lem:discounted-single-policy-transfer}, for every \(n\),
\[
\begin{aligned}
    J_{i,\bar H}^\gamma(\bar{\bm{\pi}}_n)
    -
    J_{i,\bar H}^\gamma(\mu^i\odot\bar{\bm{\pi}}_n^{-i})
    \le\;&
    J_{i,L}^\gamma(\bm{\pi}_n^{[L]})
    -
    J_{i,L}^\gamma(\mu^{i,[L]}\odot\bm{\pi}_n^{[L],-i}) +
    \tau_{L,\bar H}(\gamma).
\end{aligned}
\]
Averaging over \(n=1,\ldots,N\) gives

\begin{align}\label{eq:quaso8}
    \frac1N\sum_{n=1}^N
    \left[
        J_{i,\bar H}^\gamma(\bar{\bm{\pi}}_n)
        -
        J_{i,\bar H}^\gamma(\mu^i\odot\bar{\bm{\pi}}_n^{-i})
    \right]
    \le\;& \underbrace{
    \frac1N\sum_{n=1}^N
    \left[
        J_{i,L}^\gamma(\bm{\pi}_n^{[L]})
        -
        J_{i,L}^\gamma(\mu^{i,[L]}\odot\bm{\pi}_n^{[L],-i})
    \right]}_* +
    \tau_{L,\bar H}(\gamma).
\end{align}
Since \(\mu^{i,[L]}\in\Pi_i^{\mathrm{gen},L}\),
the assumed finite-horizon regret guarantee implies that the term \((*)\) in \eqref{eq:quaso8} is at most \(\varepsilon_{\mathrm{alg}}\),
simultaneously for all players \(i\in[m]\).  
The resulting bound is uniform over
\(\mu^i\), and thus taking the supremum over \(\Pi_i^{\mathrm{gen},\bar H}\) proves
the statement.
\end{proof}

\subsection{Proofs for the full-feedback case}\label{subsec:corollaryproof-1}

We first note that Algorithm~1 can be applied to game \(G_{\gamma,\bar H}^{[L_\varepsilon]}\), as \(G_{\gamma,\bar H}^{[L_\varepsilon]}\) is simply an undiscounted finite-horizon Markov game with costs \(c_h^i(s,a)=\gamma^{h-1}\ell_h^i(s,a)\), and these costs remain in \([0,1]\). Hence Theorem~\ref{thm:main} applies directly with the horizon parameter \(H\) replaced by \(L_\varepsilon\).

\begin{proof}[Proof of Corollary~\ref{cor:disc-full}]
We apply Theorem~\ref{thm:discounted-transfer} when
Algorithm~\ref{alg:layered-omd} runs on the truncated discounted game
\(G_{\gamma,\bar H}^{[L_\varepsilon]}\).  By Theorem~\ref{thm:main}, applied
with horizon \(L_\varepsilon\), the finite-horizon average regret on the truncated
game \(G_{\gamma,\bar H}^{[L_\varepsilon]}\) is bounded by
\( 
    \varepsilon_{\mathrm{alg}}
    =
    C_*^{\mathrm{SE}}(L_\varepsilon)
    T_\varepsilon^{-3/(3L_\varepsilon+1)}.
\) 
The definition of \(T_\varepsilon\) gives
\[
    C_*^{\mathrm{SE}}(L_\varepsilon)
    T_\varepsilon^{-3/(3L_\varepsilon+1)}
    \le
    \frac{\varepsilon}{2}.
\]
It remains to bound the truncation error.  If
\(L_\varepsilon=\bar H<\infty\), then the truncated game \(G_{\gamma,\bar H}^{[L_\varepsilon]}\), is the same as the original game \(G_{\gamma,\bar H}\) and
\( 
    \tau_{L_\varepsilon,\bar H}(\gamma)=0.
\) 
Otherwise, by the definition of
\( 
    L_\varepsilon=\min\{\bar H,H_\varepsilon^\gamma\},
\) 
we have \(L_\varepsilon=H_\varepsilon^\gamma\).  In that case, by Lemma \ref{lem:discounted-tail-unified}
\[
    \tau_{L_\varepsilon,\bar H}(\gamma)
    \le
    \frac{\gamma^{L_\varepsilon}}{1-\gamma}
    \le
    \frac{\varepsilon}{2},
\]
where the last inequality follows from the choice
\( 
    H_\varepsilon^\gamma
    =
    \left\lceil
        \frac{\log\!\left(\frac{2}{(1-\gamma)\varepsilon}\right)}
             {\log(1/\gamma)}
    \right\rceil .
\) 
Therefore, for every player
\(i\in[m]\), Theorem~\ref{thm:discounted-transfer} gives
\[
    \sup_{\mu^i\in\Pi_i^{\mathrm{gen},\bar H}}
    \sum_{t=1}^{T_\varepsilon}
    \left[
        J_{i,\bar H}^\gamma(\bar{\bm\pi}_t)
        -
        J_{i,\bar H}^\gamma(\mu^i\odot\bar{\bm\pi}_t^{-i})
    \right]
    \le
    \varepsilon T_\varepsilon .
\]
The discounted CCE conclusion in Corollary \ref{cor:disc-full} follows from the standard no-regret-to-CCE
equivalence \citep{CesaBianchiLugosi2006} applied to the empirical distribution of joint policies over
\(\bar{\bm\pi}_1,\ldots,\bar{\bm\pi}_{T_\varepsilon}\).

\vspace{6pt}
\noindent 
We now justify the upper bound for the total number of episodes. Since
\( 
    \log(1/\gamma)\ge 1-\gamma,
\) 
we have
\[
    L_\varepsilon
    \le
    L_{\gamma,\varepsilon}^{\bar H}
    :=
    \min\left\{
        \bar H,\,
        1+
        \frac{1}{1-\gamma}
        \log\!\left(\frac{2}{(1-\gamma)\varepsilon}\right)
    \right\},
\]
with the convention that when \(\bar H=\infty\),
\( 
    L_{\gamma,\varepsilon}^{\bar H}
    =
    1+
    \frac{1}{1-\gamma}
    \log\!\left(\frac{2}{(1-\gamma)\varepsilon}\right).
\) 
Under the choice \(\lambda_i=1/|A_i|\), from Theorem \ref{thm:main} we have,
\[
    C_*^{\mathrm{SE}}(L)
    \le
    \left(
        \frac{2A_{\log}}{\eta_0}
        +
        12mA_{\log}
        +
        \eta_0
        +
        144\eta^3_0m^2
    \right)L^7,
    \qquad
    A_{\log}:=\max_{i\in[m]}\log(|A_i|+1).
\]
Substituting this estimate into the definition of \(T_\varepsilon\), using
\(\lceil x\rceil\le 1+x\), gives, for every fixed \(\gamma<1\),
\[
    T_\varepsilon
    \le
    \left[
        \frac{
            \left(
                \frac{A_{\log}}{\eta_0}
                +
                mA_{\log}
                +
                m^2
                +
                1
            \right)
            \min\{\bar H,\log(1/\varepsilon)\}^7
        }{
            \varepsilon
        }
    \right]^{
        O(\min\{\bar H,\log(1/\varepsilon)\})
    }.
\]
This concludes the proof.
\end{proof}

\subsection{Proof of the partial-feedback case}\label{subsec:corollaryproof-2}

Now, similar to Subsection \ref{subsec:corollaryproof-1}, Theorem~\ref{thm:episodic-partial-feedback} can be applied to the game \(G_{\gamma,\bar H}^{[L_\varepsilon]}\), as its scaled costs \(c_h^i(s,a)=\gamma^{h-1}\ell_h^i(s,a)\) still lie in \([0,1]\). 
Hence  Theorem~\ref{thm:episodic-partial-feedback} applies directly with the horizon parameter \(H\) replaced by \(L_\varepsilon\) and the \((\kappa,\zeta_{\bar\varepsilon,L})\)-reachability condition for \(G_{\gamma,\bar H}^{[L_\varepsilon]}\).

\begin{proof}[Proof of Corollary~\ref{cor:disc-bandit}]
Set
\( 
    L=L_\varepsilon,
    \;
    \bar\varepsilon=\frac{\varepsilon}{2}.
\) 
By assumption, the truncated discounted game
\( 
    G_{\gamma,\bar H}^{[L]}
\) 
is \((\kappa,\zeta_{\bar\varepsilon,L})\)-reachable.
Therefore, applying Theorem~\ref{thm:episodic-partial-feedback}
to \(G_{\gamma,\bar H}^{[L]}\) with stated parameters,
we obtain that, with probability at least \(1-\delta\), for every player
\(i \in [m]\),
\[
    \sup_{\nu^i\in\Pi_i^{\mathrm{gen},L}}
    \frac1{N_{\varepsilon,\bar H,\delta}}
    \sum_{k=1}^{K_{\bar\varepsilon,L_\varepsilon}}
    \sum_{t=1}^{B_{\varepsilon,L_\varepsilon,\delta}}
    \left[
        J_{i,L}^{\gamma}({\bm\pi}_k^{[L]})
        -
        J_{i,L}^{\gamma}(\nu^i\odot{\bm\pi}_k^{[L],-i})
    \right]
    \le
    \bar\varepsilon
    =
    \frac{\varepsilon}{2}.
\]
By Theorem~\ref{thm:discounted-transfer}, the average regret of the
extended block policies \( 
    \bar{\bm\pi}_k
    :=
    \mathrm{Ext}_L(\bm\pi_k^{[L]})
\) 
is at most 
\( 
    \frac{\varepsilon}{2}
    +
    \tau_{L,\bar H}(\gamma).
\) 
It remains to bound the truncation error.  If
\(L_\varepsilon=\bar H<\infty\), then the truncation is the original game and
\( 
    \tau_{L_\varepsilon,\bar H}(\gamma)=0.
\)  Hence
\[
    \tau_{L,\bar H}(\gamma)
    \le
    \frac{\gamma^L}{1-\gamma}
    \le
    \frac{\varepsilon}{2}
\]
by the Lemma \ref{lem:discounted-tail-unified}.  Consequently, with probability at
least \(1-\delta\), for every player \(i\in[m]\), Theorem~\ref{thm:discounted-transfer} gives, 
\[
    \sup_{\mu^i\in\Pi_i^{\mathrm{gen},\bar H}}
    \frac1{N_{\varepsilon,\bar H,\delta}}
    \sum_{k=1}^{K_{\bar\varepsilon,L_\varepsilon}}
    \sum_{t=1}^{B_{\varepsilon,L_\varepsilon,\delta}}
    \left[
        J_{i,\bar H}^{\gamma}(\bar\pi_k)
        -
        J_{i,\bar H}^{\gamma}(\mu^i\odot\bar\pi_k^{-i})
    \right]
    \le
    \varepsilon .
\]
The discounted CCE statement follows from the standard no-regret-to-CCE
equivalence \citep{CesaBianchiLugosi2006} applied to the empirical distribution over the extended block
policies.

\vspace{6pt}
\noindent 
It remains to justify the upper bound on
\(N_{\varepsilon,\bar H,\delta}^{\mathrm{band}}\).  
We have,
\( 
    L_\varepsilon
    =
    O(\mathcal L_{\bar H,\varepsilon}),
    \;
    \mathcal L_{\bar H,\varepsilon}
    :=
    \min\{\bar H,\log(1/\varepsilon)\},
\) 
with the convention
\(\mathcal L_{\infty,\varepsilon}=\log(1/\varepsilon)\).  Then,
\begin{align}\label{eq:quaso10}
    K_{\bar\varepsilon,L_\varepsilon}
    =
    \left\lceil
        \left(
            \frac{4\widetilde C_*^{\mathrm{SE}}(L_\varepsilon)}
                 {\bar\varepsilon}
        \right)^{(3L_\varepsilon+1)/3}
    \right\rceil
    \le
    \left[
        \frac{
            \mathcal C_{\mathrm{band}}
            \mathcal L_{\bar H,\varepsilon}^{7}
        }{
            \varepsilon
        }
    \right]^{
        O(\mathcal L_{\bar H,\varepsilon})
    },
\end{align}
where, \(\mathcal C_{\mathrm{band}}
    :=
    \frac{2A_{\log}}{\eta_0}
    +12mA_{\log}
    +\eta_0
    +288m^2 \eta^3_0\) . 
Next, by the definitions
\( 
    \zeta_{\bar\varepsilon,L}
    =
    \frac{\bar\varepsilon}{8L^2},
    \;
    \xi_{\bar\varepsilon,L}
    =
    \min\left\{
        \frac{\bar\varepsilon}{8L},
        \sqrt{\frac{\bar\varepsilon}{32\eta_0 L}}
    \right\},
\) 
we have, up to universal constants,
\begin{align}\label{eq:quaso11}
    \frac{1}{\zeta_{\bar\varepsilon,L_\varepsilon}}
    \lesssim
    \frac{\mathcal L_{\bar H,\varepsilon}^{2}}{\varepsilon},
    \qquad
    \frac{L_\varepsilon^2}{\xi_{\bar\varepsilon,L_\varepsilon}^{2}}
    \lesssim
    \frac{\mathcal L_{\bar H,\varepsilon}^{4}}{\varepsilon^2}
    +
    \frac{\eta_0\mathcal L_{\bar H,\varepsilon}^{3}}{\varepsilon}.
\end{align}
Using the definitions
\( 
    M_{\bar\varepsilon,L}
    =
    K_{\bar\varepsilon,L}
    m A_{\max}S,
    \;
    u_{\bar\varepsilon,L,\delta}
    =
    \log\frac{4M_{\bar\varepsilon,L}}{\delta},
\) 
we obtain
\begin{align}\label{eq:quaso12}
    u_{\bar\varepsilon,L_\varepsilon,\delta}
    \le
    O\left(
        \mathcal L_{\bar H,\varepsilon}
        \log
        \frac{
            \mathcal C_{\mathrm{band}}
            \mathcal L_{\bar H,\varepsilon}^{7}
        }{
            \varepsilon
        }
        +
        \log
        \frac{
            mA_{\max}S
        }{
            \delta
        }
    \right).
\end{align}
Finally, substituting \eqref{eq:quaso11},\eqref{eq:quaso12} into
\( 
    B_{\bar\varepsilon,L_\varepsilon,\delta}
    =
    \left\lceil
        \frac{8A_{\max}}
             {\kappa\zeta_{\bar\varepsilon,L_\varepsilon}}
        \left(
            n_{\bar\varepsilon,L_\varepsilon,\delta}
            +
            u_{\bar\varepsilon,L_\varepsilon,\delta}
        \right)
    \right\rceil,
    \;
    n_{\bar\varepsilon,L_\varepsilon,\delta}
    =
    \left\lceil
        \frac{L_\varepsilon^2}{2\xi_{\bar\varepsilon,L_\varepsilon}^2}
        u_{\bar\varepsilon,L_\varepsilon,\delta}
    \right\rceil,
\) up to universal constants we get; 
\begin{align} \label{eq:quaso9}
B_{\bar\varepsilon,L_\varepsilon,\delta} \lesssim \left\lceil
        \frac{A_{\max}}
             {\kappa}\frac{\mathcal L_{\bar H,\varepsilon}^{2}}{\varepsilon} \left(\frac{\mathcal L_{\bar H,\varepsilon}^{4}}{\varepsilon^2}
    +
    \frac{\eta_0\mathcal L_{\bar H,\varepsilon}^{3}}{\varepsilon}+1\right)
            u_{\bar\varepsilon,L_\varepsilon,\delta}
    \right\rceil,
\end{align}
and multiplying \eqref{eq:quaso9} by \eqref{eq:quaso10}, yields 
\[
\begin{aligned}
    N_{\varepsilon,\bar H,\delta}
    \le
    \left[
        \frac{
            A_{\max}\mathcal L_{\bar H,\varepsilon}^{2}
        }{
            \kappa\varepsilon
        }\!
        \left(
            1
            +
            \frac{\mathcal L_{\bar H,\varepsilon}^{4}}{\varepsilon^2}
            +
            \frac{\eta_0\mathcal L_{\bar H,\varepsilon}^{3}}{\varepsilon}\!
        \right)\!\!
        \times\!\!
        \left(\!
            \mathcal L_{\bar H,\varepsilon}\!
            \log
            \frac{
                \mathcal C_{\mathrm{band}}
                \mathcal L_{\bar H,\varepsilon}^{7}
            }{
                \varepsilon
            }
            \!+\!
            \log
            \frac{\!
                mA_{\max}S\!
            }{
                \delta
            }
        \right)\!
    \right]\!\!
    \left[
        \frac{\!
            \mathcal C_{\mathrm{band}}
            \mathcal L_{\bar H,\varepsilon}^{7}\!
        }{
            \varepsilon
        }
    \right]^{
        O(\mathcal L_{\bar H,\varepsilon})
    }
\end{aligned}
\]
This concludes the proof.  
\end{proof}

\section{Proofs for the Computational Lower Bounds (Section \ref{sec:computational-lower-bound})}
\label{app:discounted-lower-bound}

\noindent
For an encoded object \(z\), we write \(\bits(z)\) for the length of its binary encoding.
Rational numbers are encoded in reduced form: a rational number \(q\in\mathbb{Q}\) is written
as \(q=p/r\), where \(p\in\mathbb{Z}\), \(r\in\mathbb{N}\), and \(\gcd(|p|,r)=1\). We define
\(
\bits(q)
:=
1+\left\lceil \log_2(|p|+1)\right\rceil
+\left\lceil \log_2(r+1)\right\rceil .
\) 

\subsection{Proof of Algorithms~\ref{alg:layered-omd} and
\ref{alg:blocked-bandit-oomd} solving the sparse discounted-CCE
problem}
\label{subsec:proof-search-problem-output}

\begin{proof}[Proof of
Proposition~\ref{prop:algorithms-solve-discounted-sparse-cce}]
We verify separately the policy-class requirement and the discounted
CCE requirement appearing in
Definition~\ref{def:discounted-sparse-markov-cce}.

\vspace{6pt}
\noindent 
First consider Algorithm~\ref{alg:layered-omd}. At every episode
\(t\), player \(i\), layer \(h\), and state \(s\), the algorithm sets
\( 
    \pi_{t,h}^{i,[L_\varepsilon]}(\cdot\mid s)
    =
    x_{t,h}^{i,s}
    \in \Delta(\mathcal A_i)
\) independently for each instance.
Hence
\( 
    \pi_t^{i,[L_\varepsilon]}
    =
    \bigl\{
        \pi_{t,h}^{i,[L_\varepsilon]}
    \bigr\}_{h=1}^{L_\varepsilon}
    \in
    \Pi_i^{\mathrm{markov},L_\varepsilon},
\) 
and therefore
\(       \bm\pi_t^{[L_\varepsilon]}
    =
    \bigl(
        \pi_t^{1,[L_\varepsilon]},\ldots,
        \pi_t^{m,[L_\varepsilon]}
    \bigr)
    \in
    \Pi^{\mathrm{markov},L_\varepsilon}.
\) 
The continuation used by
\(\operatorname{Ext}_{L_\varepsilon}\) is itself a fixed product Markov
policy. Consequently,
\( 
    \bar{\bm\pi}_t
    =
    \operatorname{Ext}_{L_\varepsilon}
    \bigl(\bm\pi_t^{[L_\varepsilon]}\bigr)
    \in
    \Pi^{\mathrm{markov},\bar H}
\) 
for every \(t\in[T_\varepsilon]\). Corollary~\ref{cor:disc-full} gives, for every player \(i\in[m]\),
\( 
    \sup_{\mu^i\in\Pi_i^{\mathrm{gen},\bar H}}
    \sum_{t=1}^{T_\varepsilon}
    \left[
        J_{i,\bar H}^{\gamma}
        (\bar{\bm\pi}_t)
        -
        J_{i,\bar H}^{\gamma}
        \bigl(
            \mu^i\odot
            \bar{\bm\pi}_t^{-i}
        \bigr)
    \right]
    \leq
    \varepsilon T_\varepsilon.
\) 
Thus the sequence
\( 
    \bigl(
        \bar{\bm\pi}_1,\ldots,
        \bar{\bm\pi}_{T_\varepsilon}
    \bigr)
\) 
satisfies both requirements of
Definition~\ref{def:discounted-sparse-markov-cce}, and therefore solves
\( 
    (T_\varepsilon,\varepsilon)\text{-}
    \mathsf{DiscSparseMarkovCCE}_{\gamma}^{\mathrm{gen}}.
\) 

\vspace{6pt}
\noindent 
We next consider
Algorithm~\ref{alg:blocked-bandit-oomd}. At every block \(k\), player
\(i\), layer \(h\), and state \(s\), the algorithm sets
\( 
    \pi_{k,h}^{i,[L_\varepsilon]}(\cdot\mid s)
    =
    x_{k,h}^{i,s}
    \in
    \Delta_i^\zeta
    \subseteq
    \Delta(\mathcal A_i)
\) independently for each instance.
Therefore,
\( 
    \bm\pi_k^{[L_\varepsilon]}
    \in
    \Pi^{\mathrm{markov},L_\varepsilon},
    \;
    \bar{\bm\pi}_k
    =
    \operatorname{Ext}_{L_\varepsilon}
    \bigl(\bm\pi_k^{[L_\varepsilon]}\bigr)
    \in
    \Pi^{\mathrm{markov},\bar H}.
\) 
Since the policy \(\bar{\bm\pi}_k\) is played throughout
block \(k\), every entry of the sequence
\( 
    \bm\sigma^{(k-1)B_\varepsilon+r}
    :=
    \bar{\bm\pi}_k
\) 
belongs to \(\Pi^{\mathrm{markov},\bar H}\). By Corollary~\ref{cor:disc-bandit}, with probability at least
\(1-\delta\), for any player \(i\in[m]\),
\( 
    \sup_{\mu^i\in\Pi_i^{\mathrm{gen},\bar H}}
    \sum_{n=1}^{N_\varepsilon}
    \left[
        J_{i,\bar H}^{\gamma}
        (\bm\sigma^{n})
        -
        J_{i,\bar H}^{\gamma}
        \bigl(
            \mu^i\odot
            \bm\sigma^{n,-i}
        \bigr)
    \right]
    \leq
    \varepsilon N_\varepsilon.
\) 
Hence the episode-level sequence solves
\( 
    (N_\varepsilon,\varepsilon)\text{-}
    \mathsf{DiscSparseMarkovCCE}_{\gamma}^{\mathrm{gen}}
\) 
on this event.
\end{proof}

\subsection{Game size and encoding conventions}
\label{app:lower-bound-search-problems}

For a finite-horizon Markov game
\( 
G_H
=
\Bigl(
H,\{S_h\}_{h=1}^{H+1},
\{A_i\}_{i=1}^m,
\{P_h\}_{h=1}^{H},
\{\ell_h^i\}_{i\in[m],h\in[H]},
s_1
\Bigr),
\) 
we define the term $\beta(G_H)$ as 
\[
\begin{aligned}
\beta(G_H)
:=
\max\Biggl\{
\max_{\substack{
h\in[H], s\in S_h, \bm a\in A,\
s'\in S_{h+1}
}}
\bits\left(P_h(s'\mid s,\bm a)\right),\max_{\substack{
i\in[m], h\in[H], s\in S_h, \bm a\in A
}}
\bits\left(\ell_h^i(s,\bm a)\right)
\Biggr\}.
\end{aligned}
\]
Thus, \(\beta(G_H)\) is the largest bit length of one transition-probability entry or one loss entry.
We define
\( 
|G_H|
:=
\max\left\{
H,
S,
A_{\max},
\beta(G_H)
\right\}.
\) 
Consequently, \(G_H\) can be described using \(|G_H|^{O(1)}\) bits for fixed \(m\). For the special case of $H=1$, we simply denote the size of a normal-form game by $|G|$ under this convention.

\vspace{6pt}
\noindent
Similarly, for a discounted infinite-horizon Markov game
\( 
G_\gamma
=
\Bigl(
S,\{A_i\}_{i=1}^m,
P,\{\ell^i\}_{i=1}^m,
\gamma,s_1
\Bigr),
\) 
we define the term $\beta(G_\gamma)$ as
\[
\begin{aligned}
\beta(G_\gamma)
:=
\max\Biggl\{\bits(\gamma),
\max_{\substack{s\in S, \bm a\in A, s'\in S}}
\bits\left(P(s'\mid s,\bm a)\right),
\max_{\substack{i\in[m], s\in S, \bm a\in A}}
\bits\left(\ell^i(s,\bm a)\right)
\Biggr\}.
\end{aligned}
\]
We then define
\( 
|G_\gamma|
:=
\max\left\{
S,
A_{\max},
\beta(G_\gamma)
\right\}.
\) 
If \(\gamma\) is fixed before the input is given, then \(\bits(\gamma)=O(1)\), and including
\(\bits(\gamma)\) in \(\beta(G_\gamma)\) changes \(|G_\gamma|\) only by a fixed constant factor
inside the maximum.

\vspace{6pt}
\noindent
For a finite-horizon discounted Markov game \(G_{\gamma,\bar H}\), we use the same convention as
above with the horizon parameter included:
\( 
|G_{\gamma,\bar H}|
:=
\max\left\{
\bar H,
S,
A_{\max},
\beta(G_{\gamma,\bar H})
\right\},
\) 
and
\[
\begin{aligned}
\beta(G_{\gamma,\bar H})
:=
\max\Biggl\{\bits(\gamma),
\max_{\substack{
h\in[\bar H], s\in S_h, \bm a\in A,\
s'\in S_{h+1}
}}\!\!\!\!\!\!
\bits\left(P_h(s'\mid s,\bm a)\right),
\max_{\substack{
i\in[m], h\in[\bar H], s\in S_h, \bm a\in A
}}\!\!\!
\bits\left(\ell_h^i(s,\bm a)\right)
\Biggr\}.
\end{aligned}
\]

\noindent
Throughout Section~\ref{sec:computational-lower-bound}, by \(N\), we denote an upper bound on the game-size parameters. Specifically,
\( 
    |G_H| \le N,\;
    |G_{\gamma,\bar H}| \le N,\;
    |G_\gamma| \le N,
\) 
for finite-horizon, finite-horizon discounted, and infinite-horizon discounted Markov games, respectively. The letter \(n\) is reserved
for the number of actions per player in the corresponding games. 
 We particularly name the \(\widehat M\) from Lemma \ref{lemma:quaso} as the \emph{source} game after the bimatrix game described in Lemma 8.3 of \citet{Rubinstein2016}, which is used in the reduction underlying Theorem~\ref{thm:rubinstein-qpoly-hardness}. Furthermore, we emphasize that, to remain consistent with the previous sections, we will use these games under the \emph{loss} convention instead of the \emph{utility-payoff} convention, since this conversion neither changes the equilibrium points nor the size of the games in view of Lemma~\ref{lem:clock-size-app} and the subsequent results. For brevity, for further details regarding the exact construction of the game, we refer to \citet{Rubinstein2016}. 

\subsection{The logarithmic-horizon undiscounted Markov game}\label{subsect:quaso0}

We use a finite-horizon Markov game construction from \citet{noah_hardness}. Let \(M=(L^1,L^2)\) be a bimatrix \(n\times n\) game. For \(M\), we identify a corresponding finite-horizon Markov game
\(F_H(M)\) as follows, where \(H\ge 2\) is a fixed even (numbered) horizon length.

\vspace{6pt}
\noindent 
In \(F_H(M)\), both players share the same state and action sets. The admissible action set is \([n]\) at every non-terminal layer.  For odd (numbered) layers \(h\), the state set is a singleton \(S_h=\{s_h^\star\}\), and for even (numbered)
layers \(h\), the state set is \(S_h=[n]\times[n]\).

\vspace{6pt}
\noindent 
For each layer \(h\), the state transition probability kernels  \(P_h\times [n]\times[n]\to \Delta(S_{h+1}) \) are ``deterministic'' and are defined as follows. For an odd layer  \(h \), if the current state is  \(s_h^\star \) and the players choose the joint action  \((a,b)\in[n]\times[n] \), then the next state records this joint action:
 \(
P_h\bigl((a,b)\mid s_h^\star,(a,b)\bigr)=1 .
 \) For an even layer \(h<H\), the transition probability is independent of both the current state and the joint action: for every  \((x,y)\in[n]\times[n] \) and every  \((a,b)\in[n]\times[n] \),
\( 
P_h\bigl(s_{h+1}^\star\mid (x,y),(a,b)\bigr)=1 .
\) 
All other transition probabilities are zero.

\vspace{6pt}
\noindent 
Furthermore, at an odd layer \(h\), if the players choose the action profile \((a,b)\in[n]\times[n]\), then the player \(i\) incurs cost
\( 
    \ell_h^i(s_h^\star,(a,b))
    :=
    \frac{1}{H}L^i(a,b),
\) 
At an even layer \(h<H\), all costs are zero and the next state is the next
odd-layer singleton \(s_{h+1}^\star\).

\vspace{6pt}
\noindent 
Let
\( 
    \mathcal O_H:=\{1,3,5,\ldots,H-1\}
\) 
be the set of odd layers.  For a given product Markov profile \(\bm\pi_t\) and
odd layer \(h\in\mathcal O_H\), define the corresponding local mixed strategies of players as
\( 
    x_{t,h}:=\pi^1_{t}(\cdot\mid s_h^\star),
\)
and
\(
    y_{t,h}:=\pi^2_{t}(\cdot\mid s_h^\star).
\) 
Finally, it can be clearly seen that the described construction takes polynomial time.

\begin{lemma}[Lemma D.4 \citet{noah_hardness}]
\label{lem:fgk-extraction-app}
Let \(M\) be a bimatrix game. There exists an absolute constant \(c_F>0\), independent of any problem parameter, such that the following holds.  Let
\(H\ge 2\) be even, \(T \in \mathbb N_+\), and let \(F_H(M)\) be the finite-horizon Markov game constructed as described in Appendix \ref{subsect:quaso0}.
Suppose,
\( 
    L_H=(\bm\pi_1,\ldots,\bm\pi_T)
\) 
is an \(\varepsilon_\star/4\)-approximate \(T\)-sparse CCE
of \(F_H(M)\), against non-Markov deviations.  If
\( 
    T<\exp(c_F\varepsilon_\star^2H),
\) 
then there exist \(t\in[T]\) and \(h\in\mathcal O_H\) such that
\( 
    (x_{t,h},y_{t,h})
\) 
is an \(\varepsilon_\star\)-approximate Nash equilibrium of \(M\), where \( 
    \mathcal O_H
\) 
is the set of odd numbered layers. 
\end{lemma}

\subsection{Embedding into a fixed-discount game}\label{subsect:quaso1}

We extend the construction of \(F_H(M)\) to the infinite-horizon setting. This enables us to establish a connection between the notion of CCE in discounted infinite-horizon Markov games and the notion of Nash equilibrium in the associated normal-form game \(M\).

\vspace{6pt}
\noindent 
For a finite-horizon Markov game \( 
G_H
=
\Bigl(
H,\{S_h\}_{h=1}^{H+1},
\{A_i\}_{i=1}^m,
\{P_h\}_{h=1}^{H},
\{\ell_h^i\}_{i\in[m],h\in[H]},
s_1
\Bigr),
\)  we define a corresponding infinite-horizon discounted Markov game
\( 
    \widetilde G_{H,\gamma}
    =
    (\widetilde S,\{A_i\}_{i=1}^m,
      \widetilde P,\{\widetilde\ell^i\}_{i=1}^m,
      \gamma,s_1)
\), where \(\gamma\in(0,1) \cap \mathbb Q\) as follows.

\vspace{6pt}
\noindent
The state space is
\( 
    \widetilde S
    :=
    \{(h,s):h\in[H],\ s\in S_h\}\cup\{\bot\},
\) 
where \(\bot\) is an absorbing state.  The action sets of the players stay the same. In this configuration, we define the transition kernel of the players in $\widetilde G_{H,\gamma}$ as,   
\[
\widetilde P(x' \mid x,a)
=
\begin{cases}
P_h(s' \mid s,a),
& \text{if } x=(h,s),\ h<H,\ x'=(h+1,s'),\\[4pt]
1,
& \text{if } x=(H,s),\ x'=\bot,\\[4pt]
1,
& \text{if } x=\bot,\ x'=\bot,\\[4pt]
0,
& \text{otherwise}.
\end{cases}
\]
Choose a scaling parameter
\( 
    \alpha:=\gamma^{H-1},
\) 
and define the cost function of player \(i\) as
\[
    \widetilde\ell^i((h,s),a)
    :=
    \alpha\gamma^{-(h-1)}\ell_h^i(s,a)
    =
    \gamma^{H-h}\ell_h^i(s,a),
    \qquad h\in[H],
\]
while \(\widetilde\ell^i(\bot, \cdot): \equiv 0\). Moreover, as \(\widetilde G_{H,\gamma}\) can be constructed from \(G_H\) in polynomial time, since it only adds one absorbing state and rescales the finitely many transition and cost entries of \(G_H\). Then, we have the following result.
\begin{lemma}
\label{lem:discounted-clock-transfer-app}
For every policy profile \(\bm\sigma \in \Pi^{\mathrm{markov}}\) in \(\widetilde G_{H,\gamma}\) and every
player \(i\),
\( 
    J_i^\gamma(\bm\sigma;\widetilde G_\gamma)
    =
    \alpha
    V^{i,\bm\sigma^{[H]} }(s_1;G_H).
\) 
Moreover, for every sparse list
\(\widetilde L=(\bm\sigma_1,\ldots,\bm\sigma_T)\) where, for all \(t \in [T]\) \(\sigma_t   \in \Pi^{\mathrm{markov,\infty}}\), it holds that, if \(\widetilde L\) is an \(\varepsilon\)-approximate discounted CCE
of \(\widetilde G_{H,\gamma}\), then \(\widetilde L^{[H]}\) is an
\((\varepsilon/\alpha)\)-approximate CCE of \(G_H\).
\end{lemma}

\begin{proof}
For the first \(H\) layers, the infinite-horizon game \(\widetilde G_{H,\gamma}\) exactly  emulates \(G_H\). By construction, after the layer \(H\), the state of the players in the game \(\widetilde G_{H,\gamma}\) permanently moves to
\(\bot\), where all future costs of the players are zero.  Along every simulated state-action pair, i.e. \(((h,s_h),a_h)\), for \(h\le H\), we obtain
\[
    \gamma^{h-1}\widetilde\ell^i((h,s_h),a_h)
    =
    \gamma^{h-1}\alpha\gamma^{-(h-1)}\ell_h^i(s_h,a_h)
    =
    \alpha\ell_h^i(s_h,a_h).
\]
Taking the expectation of the expression above, we obtain
\[
    J_i^\gamma(\bm\sigma;\widetilde G_\gamma)
    =
    \alpha
    \mathbb E^{\bm\sigma^{[H]}}
    \left[\sum_{h=1}^H \ell_h^i(s_h,a_h)\right]
    =
    \alpha V^{i,\bm\sigma^{[H]}}(s_1;G_H).
\]
Now fix a list \(\widetilde L\). Then, for any infinite-horizon deviation
\(\mu^i \in \Pi_i^{\mathrm{gen},\infty}\) we get the following equality
\begin{align}\label{eq:quaso19}
    &\frac1T\sum_{t=1}^T
    \left[
        J_i^\gamma(\bm\sigma_{t};\widetilde G_\gamma)
        -
        J_i^\gamma(\mu^i\odot\bm\sigma_t^{-i};\widetilde G_\gamma)
    \right]  =
    \alpha\frac1T\sum_{t=1}^T
    \left[
        V^{i,\bm\sigma_t^{[H]}}(s_1;G_H)
        -
        V^{i,\mu^{i,[H]} \odot \bm\sigma_t^{[H],-i}}(s_1;G_H)
    \right].
\end{align}
It can be seen that, every \(H\)-step deviation can be extended arbitrarily to infinite horizon
after the absorbing layer \(H\), and the discounted value of the players remains the same because all post-\(H\) costs are zero. Thus, taking the supremum over deviations
on the LHS of \eqref{eq:quaso19} is equivalent to taking the supremum over \(H\)-step deviations on
the right.  This completes the proof. 
\end{proof}

\vspace{6pt}
\noindent
Now, we describe the expression sizes of the source game \(\widehat M\) and its corresponding encoding  \(\widetilde F_{H,\gamma}(\widehat M)\) in the next two lemmas, Lemmas \ref{lem:rubinstein-source} and \ref{lem:clock-size-app}.

\begin{lemma}[Size of the source game]
\label{lem:rubinstein-source}
Let \(\widehat M\) be the \emph{source} two-player game  described in Lemma 8.3 of \citet{Rubinstein2016}, and let \(n\) be the number of pure actions per player in \(\widehat M\).  Then, there exist absolute constants \(q_{\rm src},n_{\rm src}\ge 1\), such that
\( 
    n
    \le
    |\widehat M|
    \le
    (n+2)^{q_{\rm src}}
\)
for all \(n\ge n_{\rm src}\).  
\end{lemma}

\begin{proof}[Proof outline]
The lower bound follows immediately from the definition of \(|\widehat M|\). For the upper bound, observe that the payoff functions of the bimatrix game
constructed in Lemma~8.3 of \citet{Rubinstein2016} consist of two types of
components: a scaled bipartite polymatrix game
constructed in Section~7 of \citet{Rubinstein2016}, and the payoffs of the two auxiliary
\emph{Alth\"ofer} games.

\vspace{6pt}
\noindent
We first consider the polymatrix component. The bipartite polymatrix game
constructed in Section~7 of \citet{Rubinstein2016} contains number of
players that is at most polynomial in
\(n\) for each \emph{side} of the bipartite polymatrix game. Moreover, as
described in Sections~7.2--7.4 of \citet{Rubinstein2016}, each polymatrix payoff entry is computed
from a number of quantities that is polynomially bounded
by \(n\). Moreover, all numerical quantities are represented at the
fixed precision. Therefore, there exist an
absolute constant \(r_1\) and a sufficiently large absolute threshold \(n_{\rm src}\) such that every payoff entry of the bipartite polymatrix
game has binary encoding length at most
\(
    (n+2)^{r_1}
\)
whenever \(n\ge n_{\rm src}\).

\vspace{6pt}
\noindent
Then, the main payoff component is obtained by multiplying the
bipartite polymatrix payoff by factors of the form
\(\lambda n_A\) or \(\lambda n_B\), where \(\lambda\) is an absolute
constant and \(n_A,n_B\) are the numbers of players on the two sides of
the bipartite polymatrix game. It follows that the binary encoding length
of every payoff entry arising from the scaled polymatrix component is at
most
\(
    (n+2)^{r_2}
\)
for some absolute constant \(r_2\) and all \(n\ge n_{\rm src}\), after
sufficiently large \(n_{\rm src}\).

\vspace{6pt}
\noindent
We next consider the two \emph{Alth\"ofer}-game components used in
Lemma~8.3 of \citet{Rubinstein2016}. Their payoff entries belong to
\(\{0,1\}\), and therefore have constant binary encoding length. Since
each final bimatrix payoff is the sum of one scaled polymatrix payoff and
a constant number of Alth\"ofer payoffs, there exists an absolute constant
\(r_0\) such that
\(
    \beta(\widehat M)\le (n+2)^{r_0}
\)
for all \(n\ge n_{\rm src}\), where \(\beta(\widehat M)\) denotes the maximum
binary encoding length of any payoff entry of \(M\). Consequently,
\( 
\begin{aligned}
    |\widehat M|
    =\max\{n,\beta(\widehat M)\}
    \le
    \max\bigl\{n,(n+2)^{r_0}\bigr\}
    \le
    (n+2)^{q_{\rm src}}
\end{aligned}
\)
for a sufficiently large absolute constant
\(
    q_{\rm src}.
\) 

\vspace{6pt}
\noindent
Finally, the payoff range of the \emph{source} game need not initially be
normalized to \([0,1]\). Since the payoffs in the construction are bounded
by absolute constants, such a normalization can be carried out by an
affine transformation with fixed coefficients. This changes the payoff
bit complexity by at most an absolute additive amount and therefore does
not affect the asserted polynomial relation for sufficiently large \(n\).
\end{proof}

\begin{lemma}[Size of the discounted embedding]
\label{lem:clock-size-app}
Fix a rational discount factor \(\gamma\in(0,1)\cap\mathbb Q\). Let
\(q_{\mathrm{src}}\ge 1\) be the absolute constant from Lemma~\ref{lem:rubinstein-source}. Define
\(
    d_0 := \max\{4,q_{\mathrm{src}}+2\}.
\) 
Then, for every finite constant \(B_H\ge 2\), there exists an integer
\(n_{\mathrm{size}}=n_{\mathrm{size}}(\gamma)\) such that
the following holds for every \(n\ge n_{\mathrm{size}}\). Let \(\widehat M\) be the \emph{source} two-player game  described in Lemma 8.3 of \citet{Rubinstein2016}, with \(n\) pure actions per player, and let \(H\ge 2\) be an even integer
satisfying
\( 
    H \le B_H\log(n+2).
\) 
Let \(F_H(\widehat M)\) be the finite-horizon Markov game constructed according to Subsection \ref{subsect:quaso0}, and let
\(\widetilde F_{H,\gamma}(\widehat M)\) be the discounted game obtained from \(F_H(\widehat M)\) as described in Subsection \ref{subsect:quaso1}. Then
\( 
    n \le |\widetilde F_{H,\gamma}(\widehat M)| \le (n+2)^{d_0}.
\) 
\end{lemma}

\begin{proof}
To construct \(\widetilde F_{H,\gamma}(\widehat M)\) from \(\widehat M\), we first denote it in \emph{loss} convention through well known conversions from payoffs to losses. By Lemma~\ref{lem:rubinstein-source}, the source game satisfies
\( 
    |\widehat M| \le (n+2)^{q_{\mathrm{src}}}.
\) 
Then, every loss entry of \(\widehat M\) has bit length at most
\[
    \max_{i\in\{1,2\}}\max_{a,b\in[n]}\mathrm{bits}(L^i(a,b))
    \le (n+2)^{q_{\mathrm{src}}}.
\]
We now bound the size of \(\widetilde F_{H,\gamma}(\widehat M)\). First, consider the state and action sets. In \(F_H(\widehat M)\), every odd layer is a
singleton and every even layer has state set \([n]\times[n]\). Hence,
\( 
    \sum_{h=1}^{H}|S_h|
    =
    \frac{H}{2}+\frac{H}{2}n^2.
\) 
Recall that the state space of the players in the game \(\widetilde F_{H,\gamma}(\widehat M)\) is
\( 
    \widetilde S
    :=
    \{(h,s):h\in[H],\,s\in S_h\}\cup\{\bot\}.
\) 
Therefore
\[
    |\widetilde S|
    =
    1+\frac{H}{2}+\frac{H}{2}n^2
    \le 1+H(n+2)^2.
\]
Since \(H \le B_H \log(n+2)\), we have
\( 
    H \le n+2
\) 
whenever \(n \ge n_{\mathrm{size}}\). Consequently,
\begin{align}\label{eq:quaso13}
    |\widetilde S|
    \le 1+(n+2)^3
    \le (n+2)^4.
\end{align}
Moreover, both players in  \(\widetilde F_{H,\gamma}(\widehat M)\) have the action set \([n]\), and thus
\( 
    A_{\max}=n.
\) 
This already gives the lower bound
\[
    |\widetilde F_\gamma(\widehat M)|
    =
    \max\{|\widetilde S|,A_{\max},\beta(\widetilde G_\gamma(\widehat M)\}
    \ge A_{\max}
    =
    n.
\]
It remains to bound \(\beta(\widetilde F_{H,\gamma}(\widehat M))\). The transition
probabilities of the players in \(\widetilde F_{H,\gamma}(\widehat M)\) are all either \(0\) or \(1\), and hence the
transition entries have constant bit length. We only need to bound the loss
entries. For \(h\in[H]\), the discounted embedding uses losses of the form
\( 
    \widetilde \ell^i((h,s),a)
    =
    \gamma^{H-h}\ell^i_h(s,a),
\) 
and \(\widetilde \ell^i(\bot,a)=0\).
Thus every nonzero loss entry in \(\widetilde F_{H,\gamma}(\widehat M)\) has the form
\( 
    \gamma^{H-h}\frac{1}{H}L^i(a,b).
\) 
Since \(\gamma\in(0,1)\cap\mathbb Q\) is fixed, there is a constant
\(c_\gamma\ge 1\), depending only on \(\gamma\), such that
\( 
    \mathrm{bits}(\gamma^k)\le c_\gamma(k+1)
    \; \text{for all }k\ge 0.
\) 
Also,
\( 
    \mathrm{bits}\!\left(\frac{1}{H}\right)
    \le O(\log(H+1)).
\) 
Furthermore, there is a
universal constant \(c_{\mathrm{rat}}\ge 1\) such that, for every finite
collection of rationals \(r_1,\ldots,r_k\),
\( 
        \operatorname{bits}\!\left(\prod_{\ell=1}^k r_\ell\right)
        \le
        c_{\mathrm{rat}}
        \left(
            1+\sum_{\ell=1}^k \operatorname{bits}(r_\ell)
        \right).
\)
Then, there exists a universal constant \(c_{\mathrm{rat}}\ge 1\) such that
every loss entry of \(\widetilde F_{H,\gamma}(\widehat M)\) satisfies
\[
\begin{aligned}
    \mathrm{bits}\!\left(
        \gamma^{H-h}\frac{2}{H}L^i(a,b)
    \right)
    &\le
    c_{\mathrm{rat}}
    \left(
        c_\gamma(H+1)
        + O(\log(H+1))
        + \mathrm{bits}(L^i(a,b))
        + 1
    \right)                                                   \\
    &\le
    c_{\mathrm{rat}}
    \left(
        c_\gamma(H+1)
        + O(\log(H+1))
        + (n+2)^{q_{\mathrm{src}}}
        + 1
    \right).
\end{aligned}
\]
Since \(H\le B_H\log(n+2)\), for sufficiently large \(n_{\mathrm{size}}\)
we have
\begin{align}\label{eq:quaso14}
    \beta(\widetilde F_{H,\gamma}(\widehat M))
    \le (n+2)^{q_{\mathrm{src}}+2}
\end{align}
for all \(n\ge n_{\mathrm{size}}\). Combining \eqref{eq:quaso13} and \eqref{eq:quaso14} we obtain
\[
\begin{aligned}
    |\widetilde F_\gamma(\widehat M)|
    \le
    \max\{(n+2)^4,n,(n+2)^{q_{\mathrm{src}}+2}\}   \le
    (n+2)^{d_0},
\end{aligned}
\]
where \(d_0=\max\{4,q_{\mathrm{src}}+2\}\). Together with
\(|\widetilde F_{H,\gamma}(\widehat M)|\ge n\), we have
\( 
    n \le |\widetilde F_{H,\gamma}(\widehat M)| \le (n+2)^{d_0}.
\) 
\end{proof}

\section{Details for the Numerical Experiments (Section \ref{sec:num_results})}
\label{app:numerical-details}

\subsection{Implementation and evaluation protocol}
\label{app:numerical-protocol}

Both finite-horizon experiments use the four games specified below without any
changes.  The two experiments differ only in the feedback available to
the players.  

\subsection{Finite-Horizon Episodic Games}
\label{app:episodic-games}

Below, we list the transition and cost tables used for the finite-horizon experiments.

\subsubsection{Markov battle-of-the-sexes chain}

This is a two-player, two-action, horizon-$2$ game with state layers
\(  
    S_1=\{0\},\; S_2=\{1,2\},\; S_3=\{3\}.
\) 
At the initial state $0$, the transition and costs are
\begin{center}
\begin{tabular}{c|c|c|c}
\toprule
Joint action $a$ & Transition to $(1,2)$ & $c_1(0,a)$ & $c_2(0,a)$ \\
\midrule
$(0,0)$ & $(0.9,0.1)$ & $0.05$ & $0.20$ \\
$(1,1)$ & $(0.1,0.9)$ & $0.20$ & $0.05$ \\
$(0,1),(1,0)$ & $(0.5,0.5)$ & $0.55$ & $0.55$ \\
\bottomrule
\end{tabular}
\end{center}
At the second layer, the game terminates deterministically.  Explicitly, for $s=1$, costs are
\begin{align*}
    c_1(1,a)&=\begin{cases}
    0.15,&a=(0,0),\\
    0.35,&a=(1,1),\\
    0.65,&a\in\{(0,1),(1,0)\},
    \end{cases}\qquad
    c_2(1,a)=\begin{cases}
    0.35,&a=(0,0),\\
    0.15,&a=(1,1),\\
    0.65,&a\in\{(0,1),(1,0)\}.
    \end{cases}
\end{align*}
For $s=2$,
\begin{align*}
    c_1(2,a)&=\begin{cases}
    0.15,&a=(1,1),\\
    0.35,&a=(0,0),\\
    0.65,&a\in\{(0,1),(1,0)\},
    \end{cases}\qquad
    c_2(2,a)=\begin{cases}
    0.35,&a=(1,1),\\
    0.15,&a=(0,0),\\
    0.65,&a\in\{(0,1),(1,0)\}.
    \end{cases}
\end{align*}

\subsubsection{Routing/congestion game}

This is a two-player, two-action, horizon-$3$ game with layers
\( 
    S_1=\{0\},\; S_2=\{1,2\},\; S_3=\{3,4\},\; S_4=\{5\}.
\) 
Actions are interpreted as route choices: $0$ is an upper/fast route and $1$ is a lower/safe route.  Congestion occurs when both players choose the same action.

\vspace{6pt}
\noindent
At $h=1$, let $n_0(a)=\mathbf 1\{a_1=0\}+\mathbf 1\{a_2=0\}$ and $\kappa(a)=\mathbf 1\{a_1=a_2\}$.  The transition is
\begin{equation}
    P(1\mid0,a)=
    \begin{cases}
    0.75,&n_0(a)\le1,\\
    0.25,&n_0(a)=2,
    \end{cases}
    \qquad P(2\mid0,a)=1-P(1\mid0,a),
\end{equation}
with player-$i$ cost
\( 
    c_i(0,a)=0.10+0.45\kappa(a)+0.05a_i.
\) 
At $h=2$, for $s\in\{1,2\}$,
\begin{equation}
    P(3\mid s,a)=
    \begin{cases}
    0.85,&s=1,\ \kappa(a)=0,\\
    0.35,&s=1,\ \kappa(a)=1,\\
    0.35,&s=2,\ \kappa(a)=0,\\
    0.15,&s=2,\ \kappa(a)=1,
    \end{cases}
    \qquad P(4\mid s,a)=1-P(3\mid s,a),
\end{equation}
with base cost $0.05$ at state $1$ and $0.20$ at state $2$:
\( 
    c_i(s,a)=\operatorname{base}(s)+0.50\kappa(a)+0.05\mathbf 1\{a_i=0\}.
\) 
At $h=3$, the game terminates.  Costs are
\(  c_i(s,a)=0.10+0.55\kappa(a)+p_i(s,a),
\) 
where $p_i(s,a)=0.10$ if $(s=3,a_i=1)$ or $(s=4,a_i=0)$, and $p_i(s,a)=0$ otherwise.

\subsubsection{Three-player public-goods game}

This game has $m=3$, binary actions, horizon $2$, and layers
\( 
    S_1=\{0\},\; S_2=\{1,2\},\; S_3=\{3\}.
\) 
Action $1$ means contribute.  Let $k(a)=a_1+a_2+a_3$ be the number of contributors.  At the first layer,
\begin{equation}
    P(1\mid0,a)=
    \begin{cases}
    0.85,&k(a)\ge2,\\
    0.20,&k(a)<2,
    \end{cases}
    \qquad P(2\mid0,a)=1-P(1\mid0,a),
\end{equation}
with
\(    c_i(0,a)=0.08+0.22a_i+0.35\frac{\max\{0,2-k(a)\}}{2}.
\) 
At the second layer, the game terminates.  If $s=1$, then
\begin{equation}
    c_i(1,a)=0.10+0.18a_i+0.08\mathbf 1\{k(a)<1\}.
\end{equation}
If $s=2$, then
\begin{equation}
    c_i(2,a)=0.55+0.15a_i-0.10\min\{k(a),2\}.
\end{equation}

\subsubsection{Transition-trap game}

This is a two-player, two-action, horizon-$3$ game with layers
\( 
    S_1=\{0\},\; S_2=\{1,2\},\; S_3=\{3,4\},\; S_4=\{5\}.
\) 
Action $0$ is a shortcut and action $1$ is safe.  The shortcut gives lower immediate cost but increases the probability of entering a dangerous or trapped state.

\vspace{6pt}
\noindent
At $h=1$, let $n_0(a)=\mathbf 1\{a_1=0\}+\mathbf 1\{a_2=0\}$.  The probability of moving to state $2$ is
\begin{equation}
    p_{\mathrm{danger}}(n_0)=
    \begin{cases}
    0.10,&n_0=0,\\
    0.50,&n_0=1,\\
    0.80,&n_0=2.
    \end{cases}
\end{equation}
The transition is $P(2\mid0,a)=p_{\mathrm{danger}}(n_0(a))$, and the immediate cost is $0.05$ for choosing action $0$ and $0.13$ for choosing action $1$.

\vspace{6pt}
\noindent
At $h=2$, from $s=1$ the trap probability is $(0.05,0.20,0.50)$ for $n_0=0,1,2$, while from $s=2$ it is $(0.10,0.50,0.80)$.  The immediate cost is $0.06$ for action $0$ and $0.20$ for action $1$, with an additional $0.02$ cost in state $2$.

\vspace{6pt}
\noindent
At $h=3$, the game terminates.  In state $3$, every action has cost $0.20$.  In the trap state $4$, choosing shortcut action $0$ costs $0.80$, while action $1$ costs $0.60$.

\subsection{Discounted LQ Markov Game}
\label{app:lq-details}

The discounted case is a two-player, two-state, zero-sum Markov game.  The state is $s\in\{0,1\}$, player~0 chooses $u\in[-U,U]$, and player~1 chooses $v\in[-U,U]$.  Player~0 minimizes the raw cost
\begin{equation}
\label{eq:lq-raw-app}
    g_s(u,v)=\frac r2u^2-\frac r2v^2+b_suv+a_su+d_sv+k_s,
\end{equation}
while player~1 minimizes $-g_s(u,v)$.  The transition is
\begin{equation}
\label{eq:lq-trans-app}
    P(s'=1\mid s,u,v)=p_s+\alpha_su+\beta_sv.
\end{equation}
For the numerical experiment, the raw zero-sum costs have been transformed by
\begin{equation}
    c_0(s,u,v)=\frac12+\frac{g_s(u,v)}{2M},
    \qquad
    c_1(s,u,v)=\frac12-\frac{g_s(u,v)}{2M}.
\end{equation}
The parameters are
\begin{center}
\begin{tabular}{c|ccccccccccc}
\toprule
$s$ & $r$ & $b_s$ & $a_s$ & $d_s$ & $k_s$ & $p_s$ & $\alpha_s$ & $\beta_s$ & $U$ & $\gamma$ & $M$ \\
\midrule
$0$ & $1.50$ & $0.35$ & $1.10$ & $0.90$ & $-0.35$ & $0.22$ & $0.07$ & $-0.055$ & $1$ & $0.90$ & $2.40$ \\
$1$ & $1.50$ & $-0.30$ & $-1.05$ & $-0.95$ & $0.45$ & $0.78$ & $-0.05$ & $0.060$ & $1$ & $0.90$ & $2.40$ \\
\bottomrule
\end{tabular}
\end{center}

\noindent 
Also, our infinite-horizon benchmark can be viewed as a finite-state stochastic
dynamic-game analogue of the classical two-player zero-sum LQ dynamic games
studied in \citet{basar_olsder_ch6}. 

\subsubsection{Closed-form saddle-point solution}

Let $V_s$ be the raw, unnormalized infinite-horizon discounted value of player~0 at state $s$, and define
\( 
    \Delta:=V_1-V_0.
\) 
The Bellman-Isaacs saddle-point equation is
\( 
    V_s=\min_u\max_v
    \left\{
    g_s(u,v)+\gamma\mathbb E[V_{s'}\mid s,u,v]
    \right\}.
\) 
Thus,
\begin{align}
    \mathbb E[V_{s'}\mid s,u,v]
    &=(1-p_s-\alpha_su-\beta_sv)V_0
      +(p_s+\alpha_su+\beta_sv)V_1=V_0+(p_s+\alpha_su+\beta_sv)\Delta.
\end{align}
Therefore the Bellman-Isaacs objective equals
\begin{equation}
    \gamma V_0+\gamma p_s\Delta+k_s
    +\frac r2u^2-\frac r2v^2+b_suv+A_s(\Delta)u+D_s(\Delta)v,
\end{equation}
where
\( 
    A_s(\Delta):=a_s+\gamma\alpha_s\Delta,
    \; D_s(\Delta):=d_s+\gamma\beta_s\Delta.
\)
Thus
\begin{equation}
\label{eq:Vs-phi-app}
    V_s=\gamma V_0+\gamma p_s\Delta+k_s+\phi_s(\Delta),
\end{equation}
where
\( 
    \phi_s(\Delta)
    :=\min_u\max_v
    \left\{
    \frac r2u^2-\frac r2v^2+b_suv+A_s(\Delta)u+D_s(\Delta)v
    \right\}.
\)
The first-order saddle-point equations are
\begin{equation}
    ru+b_sv+A_s(\Delta)=0,
    \qquad
    -rv+b_su+D_s(\Delta)=0.
\end{equation}
Solving these gives
\begin{equation}
\label{eq:u-star-app}
    u_s^*(\Delta)
    =-\frac{rA_s(\Delta)+b_sD_s(\Delta)}{r^2+b_s^2},
\quad 
    v_s^*(\Delta)
    =\frac{rD_s(\Delta)-b_sA_s(\Delta)}{r^2+b_s^2}.
\end{equation}
Substituting~\eqref{eq:u-star-app} into the optimized quadratic yields
\begin{equation}
\label{eq:phi-app}
    \phi_s(\Delta)=
    \frac12
    \frac{
    r\bigl(D_s(\Delta)^2-A_s(\Delta)^2\bigr)
    -2b_sA_s(\Delta)D_s(\Delta)
    }{r^2+b_s^2}.
\end{equation}
Subtracting~\eqref{eq:Vs-phi-app} for $s=0$ from the same equation for $s=1$ yields the scalar equation
\begin{equation}
\label{eq:delta-equation-app}
    \Delta
    =\gamma(p_1-p_0)\Delta+(k_1-k_0)
    +\phi_1(\Delta)-\phi_0(\Delta).
\end{equation}
Since $\phi_s$ is quadratic in $\Delta$, equation~\eqref{eq:delta-equation-app} is a scalar quadratic.  We solve this equation numerically, and select the feasible root whose saddle-point actions lie in $[-U,U]$, and then recovers
\( 
    V_0=\frac{\gamma p_0\Delta+k_0+\phi_0(\Delta)}{1-
    \gamma},
    \; V_1=V_0+\Delta.
\) 
For the transformed normalized costs, the corresponding values are
\begin{equation}
    J_0^*(s)=\frac12+\frac{(1-
    \gamma)V_s}{2M},
    \qquad
    J_1^*(s)=\frac12-\frac{(1-
    \gamma)V_s}{2M}.
\end{equation}

\end{document}